\documentclass[preprint,12pt]{elsarticle}

\usepackage{amssymb}
\usepackage{amsthm}
\newtheorem{theorem}{Theorem}
\usepackage{amsmath}
\usepackage{subfigure}

\usepackage{algorithm} 
\usepackage{algorithmic} 

\usepackage{multirow}
\usepackage{booktabs}
\usepackage{threeparttable} 

\newcommand{\tabincell}[2]{\begin{tabular}
		{@{}#1@{}}#2
\end{tabular}} 

\journal{Information Sciences}

\begin{document}

\begin{frontmatter}



\title{Coupling Chaotic System Based on Unit Transform and Its Applications in Image Encryption}


\author[label1]{Guozhen Hu}
\author[label1]{Baobin Li}

\address[label1]{School of Computer Science and Technology,
University of Chinese Academy of Sciences, Beijing,100049, China}

\begin{abstract}
Chaotic maps are very important for establishing chaos-based image encryption systems. This paper introduces a coupling chaotic system based on a certain unit transform, which can combine any two 1D chaotic maps to generate a new one with excellent performance. The chaotic behavior analysis has verified this coupling system’s effectiveness and progress. In particular, we give a specific strategy about selecting an appropriate unit transform function to enhance chaos of generated maps. Besides, a new chaos based pseudo-random number generator, shorted as CBPRNG, is designed to improve the distribution of chaotic sequences. We give a mathematical illustration on the uniformity of CBPRNG, and test the randomness of it. Moreover, based on CBPRNG, an image encryption algorithm  is  introduced. Simulation results and security analysis indicate that the proposed image encryption scheme is competitive with some advanced existing methods.
\end{abstract}


\begin{highlights}
\item A novel chaotic system named UT-CCS is proposed as a general framework which has covered several existing chaotification models. The proposed UT-CCS combines two chaotic maps to generate a new one with excellent performance. A feasible way on how to select UTFs in UT-CCS to generate preferable chaotic maps is discussed.
\item A new chaos-based pseudo random number generator (CBPRNG) with good randomness is proposed. The uniformity of proposed CBPRNG is proved theoretically.
\item A novel digital image encryption algorithm based on CBPRNG is designed and tested.
\end{highlights}


\end{frontmatter}


\section{Introduction}
\label{SEC_Intro}
With the rapid development of the Internet and communication technologies, especially the great expansion of transmission bandwidth, multimedia data mainly consisting of digital images, audios and videos, has taken a large proportion of information exchanged on the Internet. More and more people choose to communicate with each other through multimedia information instead of classical text messages due to its intuition and vividness. Wide applications of multimedia communication bring great convenience to people's daily life. In the meantime, a series of problems also arise in obtaining, transmitting, processing and storing the multimedia data. One of the most serious problems among these issues is the data security during the transmission on common channels, because  those multimedia data includes personal privacy, trade secret and even military secret. The examples of huge losses caused by data leakage are not rare, thus, 
it needs to be protected carefully against potential attackers \cite{chenSymmetricImageEncryption2004,pareekImageEncryptionUsing2006,liaoNovelImageEncryption2010,langImageEncryptionBased2012}.

Considering that videos are essentially composed of sequential image frames, the image information security on which this paper focused, is one of the main branches of multimedia information security. A straightforward approach to protect digital images in transmission is image encryption. That is, the sender encrypts the image into a meaningless cipher, typically a noise-like image, and the receiver decrypts the ciphertext image to recover the original image using pre-negotiated keys. According to the kerckhoffs principle \cite{petitcolasKerckhoffsPrinciple2011}, the security of this cryptosystem shouldn't depend on keeping secret for encrypting scheme, but secret keys only. Thus, it is often difficult to design an image encryption algorithm with excellent performance as attackers will know every details of your method.

From the perspective of computer storage and communication, the image data has no difference from text information, as they are all represented by binary bit streams. Therefore, traditional text encryption algorithms such as Data Encryption Standard \cite{coppersmithDataEncryptionStandard1994}, are applicable for digital image encryption. However, this way brings low efficiency and poor safety because image data has its own characteristics such as high correlation among adjacent pixels, lager amount of data and high redundancy. Over the years, a number of algorithms designed for image data encryption have been proposed  \cite{chenSymmetricImageEncryption2004,pareekImageEncryptionUsing2006,liaoNovelImageEncryption2010,langImageEncryptionBased2012,wuDesignImageCipher2014,
kangDoubleRandomScrambling2017,huaCosinetransformbasedChaoticSystem2019,kangRealityPreservingMultipleParameter2019,
guanChaosbasedImageEncryption2005,belaziNovelImageEncryption2016}. Among these algorithms, the chaos based methods in \cite{chenSymmetricImageEncryption2004,pareekImageEncryptionUsing2006,huaCosinetransformbasedChaoticSystem2019,guanChaosbasedImageEncryption2005,belaziNovelImageEncryption2016}, attract growing attentions from researchers owing to the good nature of chaotic system. A dynamical system is said to be chaotic when it is sensitive to initial value, aperiodic and non-convergence \cite{strogatzNonlinearDynamicsChaos2018}. These characteristics make the chaotic system meet all the needs for encryption \cite{ozkaynakBriefReviewApplication2018}.

In today's image encryption fields, it is no exaggeration to say that most image encryption algorithms are designed based on chaotic maps. A typical way is to set the chaotic map as a pseudo-random number generator (PRNG). As a result, the security of these encryption methods depends largely on the performance of their underlying chaotic maps. The existing chaotic maps with excellent performance (e.g. the Lorenz system and Chen system) are often relatively sophisticated, which increases the complexity of encryption process. Researchers tend to adapt simple 1D chaotic maps such as Logistic map to achieve high efficiency. However, these simple systems have some drawbacks which may cause security hazards to encryption process. Firstly, most of them have a narrow and discontinuous chaotic range on their control parameter that is often served as a secret key in encryption. It means that the key space is limited, which makes the cryptosystem vulnerable to be attacked. Secondly, some chaotic maps are quite sensitive to the finite precision effect of computer and may degrade to non-chaotic immediately. Besides, the chaotic sequences generated by iterating the chaotic maps, usually have a non-uniform distribution, which reduces the randomness of cipher.

To address above problems, researchers have proposed a series of new chaotic systems by modifying and combining existing chaotic maps. For example, Zhang constructed the mixed linear–nonlinear coupled map lattices  to gain outstanding cryptography features in dynamics \cite{zhangSymmetricImageEncryption2014}. Zhou proposed a parametric switching chaotic system  by combining three common 1D chaotic maps to enhance complexity \cite{zhouImageEncryptionUsing2013}. A homogenized Chebyshev-Arnold map was discussed by Luo in \cite{luoNewImageEncryption2018}, which greatly improved the chaotic behavior of the original Chebyshev map.
%
Hua and Zhou combined two 1D chaotic maps to construct a 2D chaotic system with superior chaotic property for image encryption and improved multiple versions \cite{hua2DSineLogistic2015,huaImageEncryptionUsing2016,hua2DLogisticSinecouplingMap2018}. Zahmoul proposed a new Beta-function-base chaotic map  for image encryption with high efficiency \cite{zahmoulImageEncryptionBased2017}. Hua modified the chaotic sine map to improve the chaotic behavior and achieved amazing results \cite{huaSineChaotificationModel2019}. Gayathri proposed spatiotemporal mixed linear–nonlinear coupling with the logistic-sine system for chaotic orbit generation \cite{asgari-chenaghluNovelImageEncryption2019}. All these work has proved the feasibility of chaos modification and laid a solid foundation for chaos application in image enryption.

As aforementioned, each case enumerated above developed a specific chaotic map with clear definition. Recently, some researchers are not content with improving concrete chaotic maps manually. Instead, they tend to build a general framework which could process arbitrary existing chaotic maps and generate new chaotic maps with excellent performance automatically. In this way, a number of new chaotic maps could be easily get for image encryption when we set different existing chaotic maps as seed maps. Up to now, several chaotic systems of this kind have been proposed. Zhou proposed a new 1D chaotic system by simply adding two existing maps followed by a modulo operation \cite{zhouNew1DChaotic2014}. Hua did further research on Zhou's method and improved it by transforming the sum of two chaotic maps through a cosine function to enhance nonlinearity \cite{huaCosinetransformbasedChaoticSystem2019}. Asgari-Chenaghlu proposed a polynomial  chaotic system to bind any number of seed chaotic maps by polynomial coupling and achieved complex chaotic behavior \cite{asgari-chenaghluNovelImageEncryption2019}. Alawida developed a new hybrid digital chaotic system by the composition of two chaotic maps as well as modulo operation \cite{alawidaNewHybridDigital2019}. Parvaz also combined different chaotic maps by composition and proposed a combination chaotic system using lots of parameters to improve performance \cite{parvazCombinationChaoticSystem2018}. Lan proposed an integrated chaotic system by mixing three existing maps with addition, composition, and modulo operation \cite{lanIntegratedChaoticSystems2018}. The research regarding chaotic frameworks is becoming a hotspot.

Although some chaotic frameworks which can generate new chaotic maps from existing seed maps have been proposed, they still have some defects that need to be improved. At first, many existing methods introduce parameters into their frameworks, but they cannot point out how these parameters will affect the generated chaotic maps. This makes the frameworks difficult to use, as users don't know how to set parameters properly. Next, the universality of some proposed chaotic systems is not good. That is, the generated new maps are of poor chaotic property when changing seed maps from given examples. Additionally, the uneven distribution of chaotic sequences gained by iterating chaotic maps are still unresolved thoroughly in most methods. At last, the chaotic behavior of some generated chaotic maps are not complex enough when applied in cryptography fields.

According to the analysis above, this work proposes a coupling chaotic system based on unit transform, named as UT-CCS, to address these problems. UT-CCS has a simple form and can combine any two 1D maps to gain a new chaotic map with excellent performance, with the aid of  an unit transform function(UTF). And in this chaotic framworks, the UTF plays an important role, and has a great influence on the behavior of generated chaotic maps. Thus,  we give a specific guidance on how to select an appropriate UTF  to enhance chaos of generated maps in UT-CCS. Based on this strategy, some new chotic maps have been constructed, and perform better chotic behaviors than some konwn cases in terms of Lyapunov exponent and sample entropy. Besides, we design a new PRNG based on three superior chaotic maps generated by UT-CCS for image encryption. The chaos based PRNG (CBPRNG) greatly improves the numerical distribution of chaotic sequence and the uniformity of CBPRNG is proved theoretically in this paper. At last, a digital image encryption algorithm based on proposed CBPRNG is developed as an application of UT-CCS. Some general tests for security analysis demonstrate the high security of our image encryption scheme and the progress compared to existing works.

The rest of this paper is organized as follows. Section \ref{sec:chaos} introduces the structure of UT-CCS, and analyzes the selection strategy of UTF. Moreover, a series of new chaotic maps are generated as examples. The performance of generated chaotic maps will be tested in Section \ref{sec:test_of_chaos}, which is also compared with other chaotic systems.
%
Section \ref{sec:PRNG} introduces, discusses, and evaluates the CBPRNG, which will be used to establish a new image encryption algorithm  in Section \ref{sec:IE}. 
Section \ref{sec:test_IE} analyzes the security through several tests and the results are compared with a number of existing methods, while Section \ref{sec:conclusion} concludes this paper.

\section{A new chaotic framework:  UT-CCS}
\label{sec:chaos}
In this section, a coupling chaotic system based on unit transform, shorted as UT-CCS, is proposed and described in details. 
To test the performance of the system, several new chaotic maps are generated as examples using three existing chaotic maps. Besides, a feasible way on how to select an appropriate UTF is discussed.

\subsection{Structure of UT-CCS}
\label{subsec:struct}
To build the structure of UT-CCS, we first introduce the concept of unit transform function(UTF) .  A function $f$ is named  a unit transform function if it
satisfying 
$$f:[0,1] \to [0,1].$$ 
That is, A UTF is a function which maps unit interval $[0,1]$ to itself. 

\begin{figure}[htbp]
	\centering
	\includegraphics[scale=0.8]{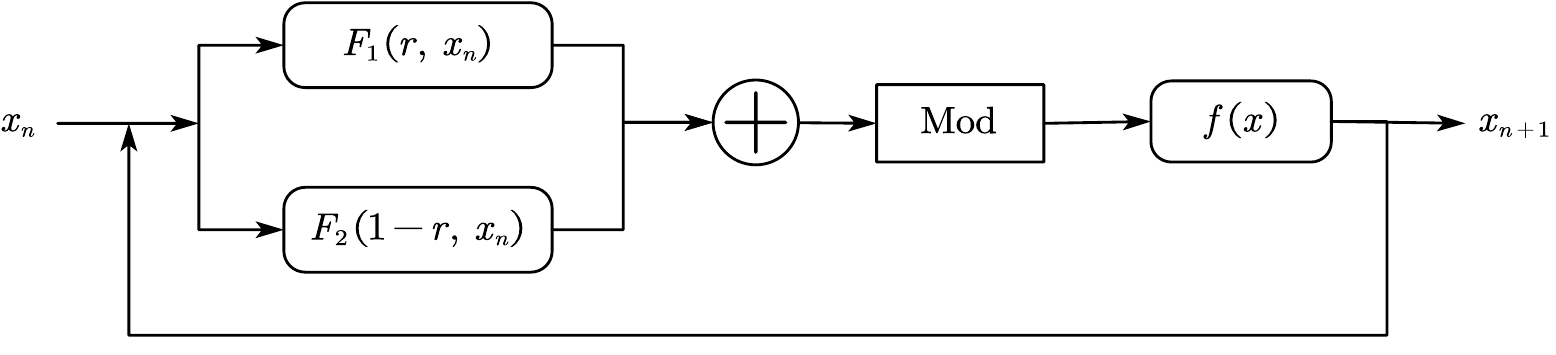}
	\caption{The structure of proposed chaotic system}
	\label{fig:framework_structure}
\end{figure}

For a unit transform function $f(x)$, the structure of  UT-CCS is   depicted in Figure \ref{fig:framework_structure}, which is defined as following:
\begin{equation}
\label{eq:chaos_def}
x_{n+1}=f((F_1(r,x_n)+F_2(1-r,x_n))\ mod \ 1),
\end{equation}
where $F_1$ and $F_2$ are two  1D chaotic maps (seed maps), $r$ and $1-r$ are their normalized control parameters and $r\in [0,1]$, and $n$ is the iteration number. In this proposed chaotic systeym, two 1D  chaotic maps are 
coupled by their normalized control parameter at first, which will mix the chaotic behaviors of these two seed maps, and  effectively expand the chaotic range of parameter $r$. Then, modulo  operation is performed to ensure the output value is in $[0,1]$. At last, a UTF $f(x)$ is used to further increase the complexity and nonlinearity of chaotic system.

In the above chaotic framwork, different seed maps and unit transforms functions will produce various chaotic maps with different behaviors.
Since the seed maps $F_1$, $F_2$ in Eq.(\ref{eq:chaos_def}) can be set as any existing chaotic maps, users can easily generate customized chaotic maps to meet various needs. Besides,  selecting different functions as UTF can also improve the performance of generated chaotic maps. Even for the same seed maps, when setting different UTFs in our proposed UT-CCS, the behaviors of generated chaotic systems change dramatically, which will be illustrated in the following section. And an effective guidance on how to choose a proper UTF will be also offered later in this paper. 


Actually, the UT-CCS is proposed as an unified chaotic framework which has covered several existing chaotification models. For example, when we choose the UTF as $f(x)=x$, then UT-CCS becomes the combined chaotic system proposed by Zhou et al. in \cite{zhouNew1DChaotic2014}; when setting $f(x)=\sin (\pi x)$, it is identified with the STBCS proposed by Hua et al. in \cite{huaSineTransformBasedChaoticSystem2018}; when $f(x)=\cos (\pi x)$ ,  the CTBCS proposed  in  \cite{huaCosinetransformbasedChaoticSystem2019} has the similar structure of UT-CCS. In this regard, our work is of great generality.
\subsection{Some Examples of UT-CCS}
\label{subsec:exp}
As examples, we select three common 1D chaotic maps as seed maps. The definitions of three seed maps with their parameters normalized are listed below:
\begin{enumerate}[1)]
	\item Logistic map: 
	\begin{equation}
	\label{eq:L}
	x_{n+1}=4rx_n(1-x_n),
	\end{equation}
	\item Tent map: 
	\begin{equation}
	\label{eq:T}
	x_{n+1}=\left\{
	\begin{aligned}
	&2rx_n, &x_n \in [0,0.5), \\
	&2r(1-x_n), &x_n \in [0.5,1], \\
	\end{aligned}
	\right.
	\end{equation}
	\item Sine map: 
	\begin{equation}
	\label{eq:S}
	x_{n+1}=r\sin(\pi x_n).
	\end{equation}
\end{enumerate}
In Eq.(\ref{eq:L}) - (\ref{eq:S}), $r$ is the normalized control parameter of chaotic maps, where $r \in [0,1]$. Besides, it can be trivially proved that these three maps always get outputs in $[0,1]$ if inputs are in $[0,1]$, which means the modulo 1 operation can be omitted when they are set as seed maps of the proposed UT-CCS.

Set any two of these maps as seed maps, then there are three combinations of coupling, naming Logistic-Sine coupling map(LSCM), Tent-Logistic coupling map(TLCM) and Sine-Tent coupling map(STCM). 
They are formally defined as Eq.(\ref{eq:CM}):
\begin{equation}
\label{eq:CM}
\begin{aligned}
&LSCM: x_{n+1}=4rx_n(1-x_n)+(1-r)\sin(\pi x_n), \\
&TLCM: x_{n+1}=\left\{
\begin{aligned}
&2rx_n+4(1-r)x_n(1-x_n), &x_n \in [0,0.5), \\
&2r(1-x_n)+4(1-r)x_n(1-x_n), &x_n \in [0.5,1], \\
\end{aligned}
\right.\\
&STCM: x_{n+1}=\left\{
\begin{aligned}
&r\sin(\pi x)+2(1-r)x_n, &x_n \in [0,0.5), \\
&r\sin(\pi x)+2(1-r)(1-x_n), &x_n \in [0.5,1]. \\
\end{aligned}
\right.\\
\end{aligned}
\end{equation}

The last step of UT-CCS is to select an appropriate UTF $f(x)$. For the sake of simplicity, now we consider five basic elementary functions (BEFs) as $f(x)$ respectively, then we can get 15 new chaotic maps totally. Note that $f(x)$ should be defined properly on unit interval [0,1]. The definitions of these new chaotic maps based on LSCM, TLCM, and STCM are listed as Table \ref{tab:def_LSCM} - \ref{tab:def_STCM}, respectively.

\begin{table}[htbp]
	\centering
	\caption{Definitions of five new chaotic maps based on LSCM}
	\renewcommand\arraystretch{1.5}
	\begin{tabular}{c|c}
		\hline
		$f(x)$ & corresponding LSCM \\
		\hline
		$f(x)=x$ & $x_{n+1}=4rx_n(1-x_n)+(1-r)\sin(\pi x_n)$ \\
		$f(x)=2^x-1$ & $x_{n+1}=2^{4rx_n(1-x_n)+(1-r)\sin(\pi x_n)}-1$ \\
		$f(x)=\frac{\ln(1+x)}{\ln 2}$ & $x_{n+1}=\frac{\ln (1+4rx_n(1-x_n)+(1-r)\sin(\pi x_n))}{\ln 2}$ \\
		$f(x)=\sin(\pi x)$ & $x_{n+1}=\sin(\pi(4rx_n(1-x_n)+(1-r)\sin(\pi x_n)))$ \\
		$f(x)=\frac{2}{\pi}\arcsin x$ & $x_{n+1}=\frac{2}{\pi}\arcsin(4rx_n(1-x_n)+(1-r)\sin(\pi x_n))$ \\[3mm]
		\hline
	\end{tabular}
	\label{tab:def_LSCM}
\end{table}

\begin{table}[htbp]
	\scriptsize
	\centering
	\caption{Definitions of five new chaotic maps based on TLCM}
	\renewcommand\arraystretch{4}
	\begin{tabular}{c|c}
		\hline
		$f(x)$ & corresponding TLCM \\
		\hline
		$f(x)=x$ & $x_{n+1}=\left\{
		\begin{aligned}
		&2rx_n+4(1-r)x_n(1-x_n), &x_n \in [0,0.5) \\
		&2r(1-x_n)+4(1-r)x_n(1-x_n), &x_n \in [0.5,1] \\
		\end{aligned}
		\right.$ \\
		
		$f(x)=2^x-1$ & $x_{n+1}=\left\{
		\begin{aligned}
		&2^{2rx_n+4(1-r)x_n(1-x_n)}-1, &x_n \in [0,0.5) \\
		&2^{2r(1-x_n)+4(1-r)x_n(1-x_n)}-1, &x_n \in [0.5,1] \\
		\end{aligned}
		\right.$ \\
		
		$f(x)=\frac{\ln(1+x)}{\ln 2}$ & $x_{n+1}=\left\{
		\begin{aligned}
		&\frac{\ln (1+2rx_n+4(1-r)x_n(1-x_n))}{\ln 2}, &x_n \in [0,0.5) \\
		&\frac{\ln (1+2r(1-x_n)+4(1-r)x_n(1-x_n))}{\ln 2}, &x_n \in [0.5,1] \\
		\end{aligned}
		\right.$ \\
		
		$f(x)=\sin(\pi x)$ & $x_{n+1}=\left\{
		\begin{aligned}
		&\sin(\pi(2rx_n+4(1-r)x_n(1-x_n))), &x_n \in [0,0.5) \\
		&\sin(\pi(2r(1-x_n)+4(1-r)x_n(1-x_n))), &x_n \in [0.5,1] \\
		\end{aligned}
		\right.$ \\
		
		$f(x)=\frac{2}{\pi}\arcsin x$ & $x_{n+1}=\left\{
		\begin{aligned}
		&\frac{2}{\pi}\arcsin(2rx_n+4(1-r)x_n(1-x_n)), &x_n \in [0,0.5) \\
		&\frac{2}{\pi}\arcsin(2r(1-x_n)+4(1-r)x_n(1-x_n)), &x_n \in [0.5,1] \\
		\end{aligned}
		\right.$ \\[3mm]
		
		\hline
	\end{tabular}
	\label{tab:def_TLCM}
\end{table}

\begin{table}[htbp]
	\scriptsize
	\centering
	\caption{Definitions of five new chaotic maps based on STCM}
	\renewcommand\arraystretch{4}
	\begin{tabular}{c|c}
		\hline
		$f(x)$ & corresponding TLCM \\
		\hline
		$f(x)=x$ & $x_{n+1}=\left\{
		\begin{aligned}
		&r\sin(\pi x)+2(1-r)x_n, &x_n \in [0,0.5) \\
		&r\sin(\pi x)+2(1-r)(1-x_n), &x_n \in [0.5,1] \\
		\end{aligned}
		\right.$ \\
		
		$f(x)=2^x-1$ & $x_{n+1}=\left\{
		\begin{aligned}
		&2^{r\sin(\pi x)+2(1-r)x_n}-1, &x_n \in [0,0.5) \\
		&2^{r\sin(\pi x)+2(1-r)(1-x_n)}-1, &x_n \in [0.5,1] \\
		\end{aligned}
		\right.$ \\
		
		$f(x)=\frac{\ln(1+x)}{\ln 2}$ & $x_{n+1}=\left\{
		\begin{aligned}
		&\frac{\ln (1+r\sin(\pi x)+2(1-r)x_n)}{\ln 2}, &x_n \in [0,0.5) \\
		&\frac{\ln (1+r\sin(\pi x)+2(1-r)(1-x_n))}{\ln 2}, &x_n \in [0.5,1] \\
		\end{aligned}
		\right.$ \\
		
		$f(x)=\sin(\pi x)$ & $x_{n+1}=\left\{
		\begin{aligned}
		&\sin(\pi(r\sin(\pi x)+2(1-r)x_n)), &x_n \in [0,0.5) \\
		&\sin(\pi(r\sin(\pi x)+2(1-r)(1-x_n))), &x_n \in [0.5,1] \\
		\end{aligned}
		\right.$ \\
		
		$f(x)=\frac{2}{\pi}\arcsin x$ & $x_{n+1}=\left\{
		\begin{aligned}
		&\frac{2}{\pi}\arcsin(r\sin(\pi x)+2(1-r)x_n), &x_n \in [0,0.5) \\
		&\frac{2}{\pi}\arcsin(r\sin(\pi x)+2(1-r)(1-x_n)), &x_n \in [0.5,1] \\
		\end{aligned}
		\right.$ \\[3mm]
		
		\hline
	\end{tabular}

	\label{tab:def_STCM}
\end{table}

\subsection{Selecting strategy of UTFs}
\label{subsec:select_UTF}
In above examples,  the cases for $f(x)=x,\ \sin (\pi x)$, have been discussed in \cite{zhouNew1DChaotic2014} and \cite{huaSineTransformBasedChaoticSystem2018}, respectively, which also show that different UTFs will bring distinct chaotic systems. Thus, how to select an appropriate UTF $f(x)$ is an important issue in our UT-CCS.
Here are some basic principles:
\begin{enumerate}[1)]
	\item $f(x)$ should be a surjection on unit interval [0,1], otherwise the system may degrade into non-chaotic with iteration.
	\item $f(x)$ should expand the chaotic range of parameter $r$. The new map is chaotic for all $r \in [0,1]$ is the best.
	\item $f(x)$ should make the new map exhibit more complex chaotic behaviors than its underlying seed maps.
	\item $f(x)$ should be as simple as possible so that chaotic sequence could be generated efficiently on computers.
\end{enumerate}
As we can see above, 1) and 4) are easy to meet while satisfying 2) and 3) need to do experiments on a number of potential $f(x)$. 

In order to get more specific guides on selecting $f(x)$, we analyze the UT-CCS in terms of Lyapunov exponent (LE) \cite{wolfDeterminingLyapunovExponents1985}, which is widely used to measure the chaotic level of a dynamical system. Suppose a dynamical system is defined by $x_{n+1}=F(x_n)$, where $F(x)$ is piecewise differentiable, then LE of the system is defined as:
\begin{equation}
\lambda_F = \lim\limits_{n \to \infty} \frac{1}{n} \sum_{i=0}^{n-1} \ln |F'(x_i)|.
\end{equation}
Generally, a positive LE demonstrates chaos of the system. Larger LE indicates faster diverging rate of two close trajectories, that is, the more complex chaotic behaviors of the system.
Therefore, the UTF $f(x)$ should be chosen to get large LE of the new chaotic map. Accordingly, a theorem regarding the selection of $f(x)$ is given below.
\begin{theorem}
	\label{th:1}
	Let $F_1$, $F_2$ be two chaotic maps generated by UT-CCS using the same seed maps, $f_1$, $f_2$ are their corresponding piecewise differentiable UTFs respectively. Then 
	$$\lambda_{F_1} \geqslant  \lambda_{F_2},$$ 
	if
	$$\forall x \in [0,1], |f'_1(x)| \geqslant |f'_2(x)|, $$
	where $\lambda_{F_i}$ denotes the Lyapunov exponent of $F_i$, $i=1,2$.
\end{theorem}
\begin{proof}
	Suppose $h(x)$, $g(x)$ are two seed maps. Denote $$\sigma(x) = (h(r,x)+g(1-r,x)) \  mod \ 1,$$  then $F_1(x)=f_1(\sigma(x))$, $F_2(x)=f_2(\sigma(x))$. Considering $\sigma(x) \in [0,1]$, from the condition of Theorem \ref{th:1}, it is easy to  get
	\begin{equation*}
	\begin{aligned}
	&\forall x \in [0,1],\  |\frac{\mathrm{d} f_1}{\mathrm{d} \sigma}| \geqslant |\frac{\mathrm{d} f_2}{\mathrm{d} \sigma}| \\
	\Rightarrow & \ln |\frac{\mathrm{d} f_1}{\mathrm{d} \sigma}| + \ln |\frac{\mathrm{d} \sigma}{\mathrm{d} x}| \geqslant \ln |\frac{\mathrm{d} f_2}{\mathrm{d} \sigma}| + \ln |\frac{\mathrm{d} \sigma}{\mathrm{d} x}| \\
	\Rightarrow & \forall x \in [0,1], \ \ln |F'_1(x)| \geqslant \ln |F'_2(x)| \\
	\Rightarrow & \lim\limits_{n \to \infty} \frac{1}{n} \sum_{i=0}^{n-1} \ln |F_1'(x_i)| \geqslant \lim\limits_{n \to \infty} \frac{1}{n} \sum_{i=0}^{n-1} \ln |F_2'(x_i)| \\
	\Rightarrow & \lambda_{F_1} \geqslant  \lambda_{F_2}
	\end{aligned}
	\end{equation*}
\end{proof}

According to Theorem \ref{th:1}, the larger $|f'(x)|$ brings larger LE of the new chaotic map. 
Thus, we can choose UTF $f(x)$ with greater $|f'(x)|$ on $[0,1]$ to improve the chaotic behavior of UT-CCS. In this way, an appropriate $f(x)$ is expected to be found easily.

Based on Theorem \ref{th:1}, we construct another three UTFs defined as following to generate preferable chaotic maps:
\begin{enumerate}[I:]
	\item $|f'(x)| \equiv 2$:
	\begin{equation}
	\label{eq:fx=2x}
		f(x) = \left\{
		\begin{array}{cl}
		2x, &x \in [0,0.5), \\
		-2x+2, &x \in [0.5,1], \\
		\end{array}
		\right.
	\end{equation} \\
	\item $|f'(x)| \equiv 4$:
	\begin{equation}
	\label{eq:fx=4x}
		f(x) = \left\{
		\begin{array}{cl}
		4x, &x \in [0,0.25), \\
		-4x+2, &x \in [0.25,0.5), \\
		4x-2, &x \in [0.5,0.75), \\
		-4x+4, &x \in [0.75,1], \\
		\end{array}
		\right.
	\end{equation} \\
	
	\item $|f'(x)| \equiv 8$:
	\begin{equation}
	\label{eq:fx=8x}
		f(x) = \left\{
		\begin{array}{cl}
		8x, &x \in [0,0.125), \\
		-8x+2, &x \in [0.125,0.25), \\
		8x-2, &x \in [0.25,0.375), \\
		-8x+4, &x \in [0.375,0.5), \\
		8x-4, &x \in [0.5,0.625), \\
		-8x+6, &x \in [0.625,0.75), \\
		8x-6, &x \in [0.75,0.875), \\
		-8x+8, &x \in [0.875,1]. \\
		\end{array}
		\right.
	\end{equation} \\
\end{enumerate}

For above three UTFs and three coupling maps---LSCM, TLCM, STCM, we can get another 9 new chaotic maps, which can be classified into three types according to which UTF above is used,
\begin{itemize}
\item[]Type I:\quad LSCM-I,\ TLCM-I,\ STCM-I;
\item[]Type II:\ \ LSCM-II,\ TLCM-II,\ STCM-II;
\item[]Type III:\ LSCM-III,\ TLCM-III,\ STCM-III.
\end{itemize}
That is, the generated chaotic maps LSCM-I, TLCM-I, and STCM-I that based on the first UTF (Eq.(\ref{eq:fx=2x})) are collectively called type-I coupling chaotic maps, and so on. The spceific mathematical definitions of these new maps are omitted.

\section{Performance analysis of generated chaotic maps}
\label{sec:test_of_chaos}
In this section, we analyze the chaotic behaviors of the generated chaotic maps in Section \ref{subsec:exp} - \ref{subsec:select_UTF} to demonstrate the effectiveness of UT-CCS. Since the type-III coupling chaotic maps are expected to gain the best performance, we will focus on them in all test items. We evaluate chaotic maps in terms of Lyapunov exponent (LE), bifurcation diagram (BD), cobweb diagram, and sample entropy (SE). The results are also compared with Zhou's method \cite{zhouNew1DChaotic2014}, Hua's CTBCS \cite{huaCosinetransformbasedChaoticSystem2019}, and the underlying seed maps.

\subsection{Lyapunov exponent}
The definition and implication of LE have been introduced in Section \ref{subsec:select_UTF}. Figure \ref{fig:LE_seed_maps} has shown the LEs of three seed maps with the change of normalized parameter $r$. Figure \ref{fig:LE_BEF} illustrates the results of the 15 new chaotic maps where five BEFs are selected as UTFs. As we can see, the LEs of Logistic map and Sine map are positive only for some $r \in [0.89,1]$, and the LE of Tent map is positive only when $r>0.5$.  That is, all the three seed maps have a narrow and discontinuous chaotic range. Figure \ref{fig:LE_BEF} shows that the 15 new chaotic maps have positive LEs for all $r \in [0,1]$, indicating that all the generated chaotic maps are superior than their underlying seed maps, which  also illustrates that our UT-CCS  is effective.


\begin{figure}[htbp]
	\centering
	\subfigure[]{
		\label{fig:LE_Logistic_map}
		\includegraphics[width=0.3\textwidth]{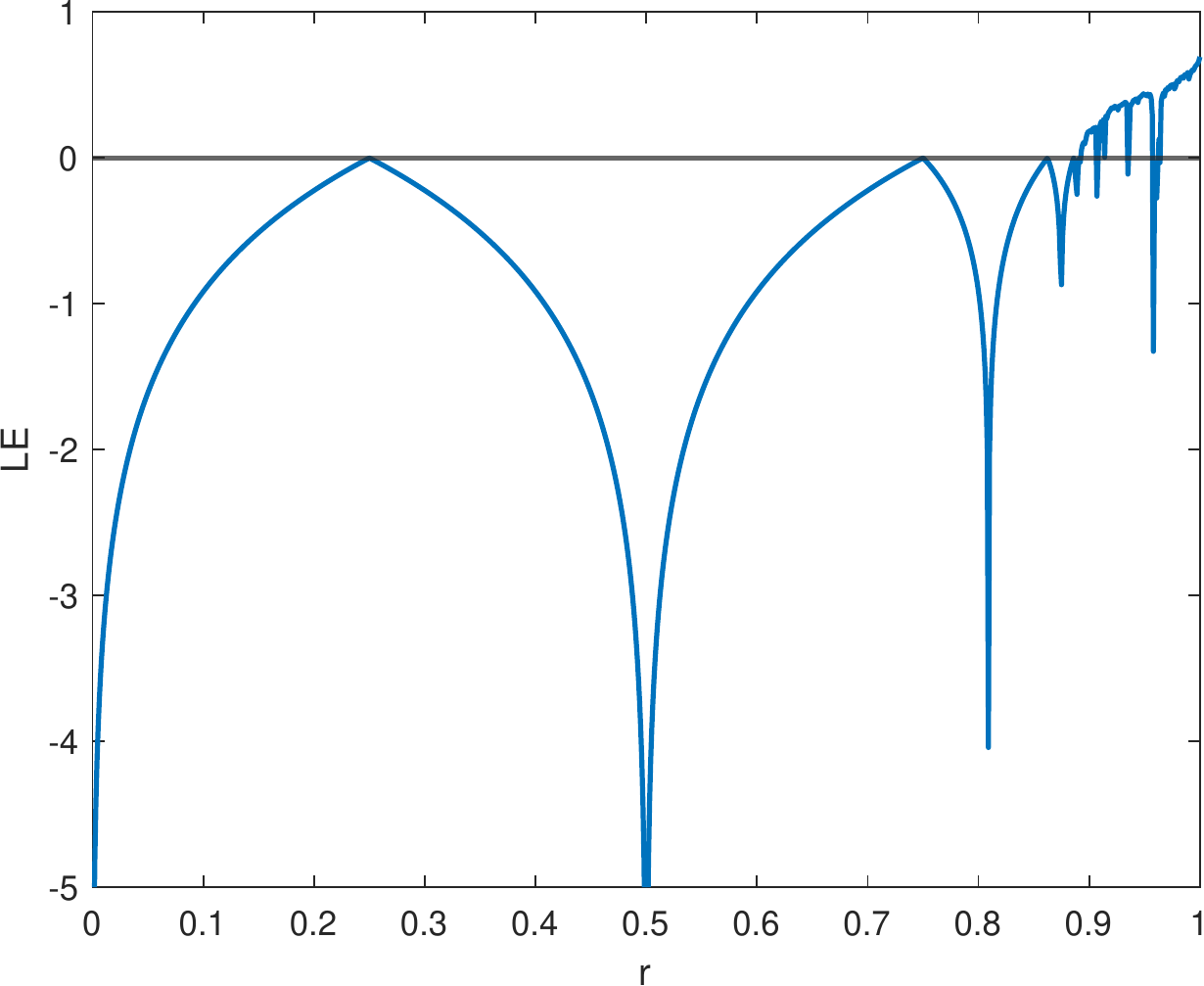}}
	\subfigure[]{
		\label{fig:LE_Tent_map}
		\includegraphics[width=0.3\textwidth]{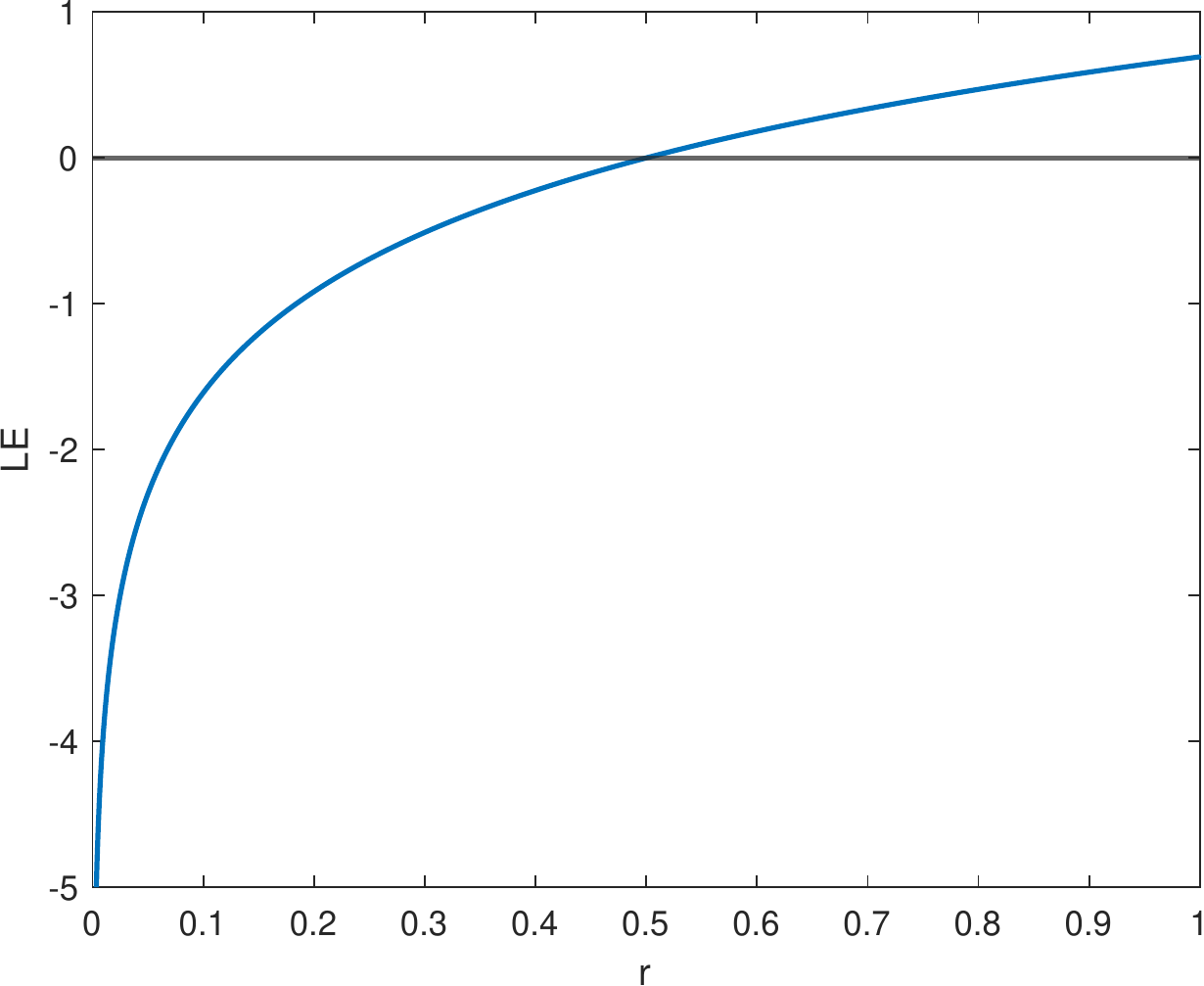}}
	\subfigure[]{
		\label{fig:LE_Sine_map}
		\includegraphics[width=0.3\textwidth]{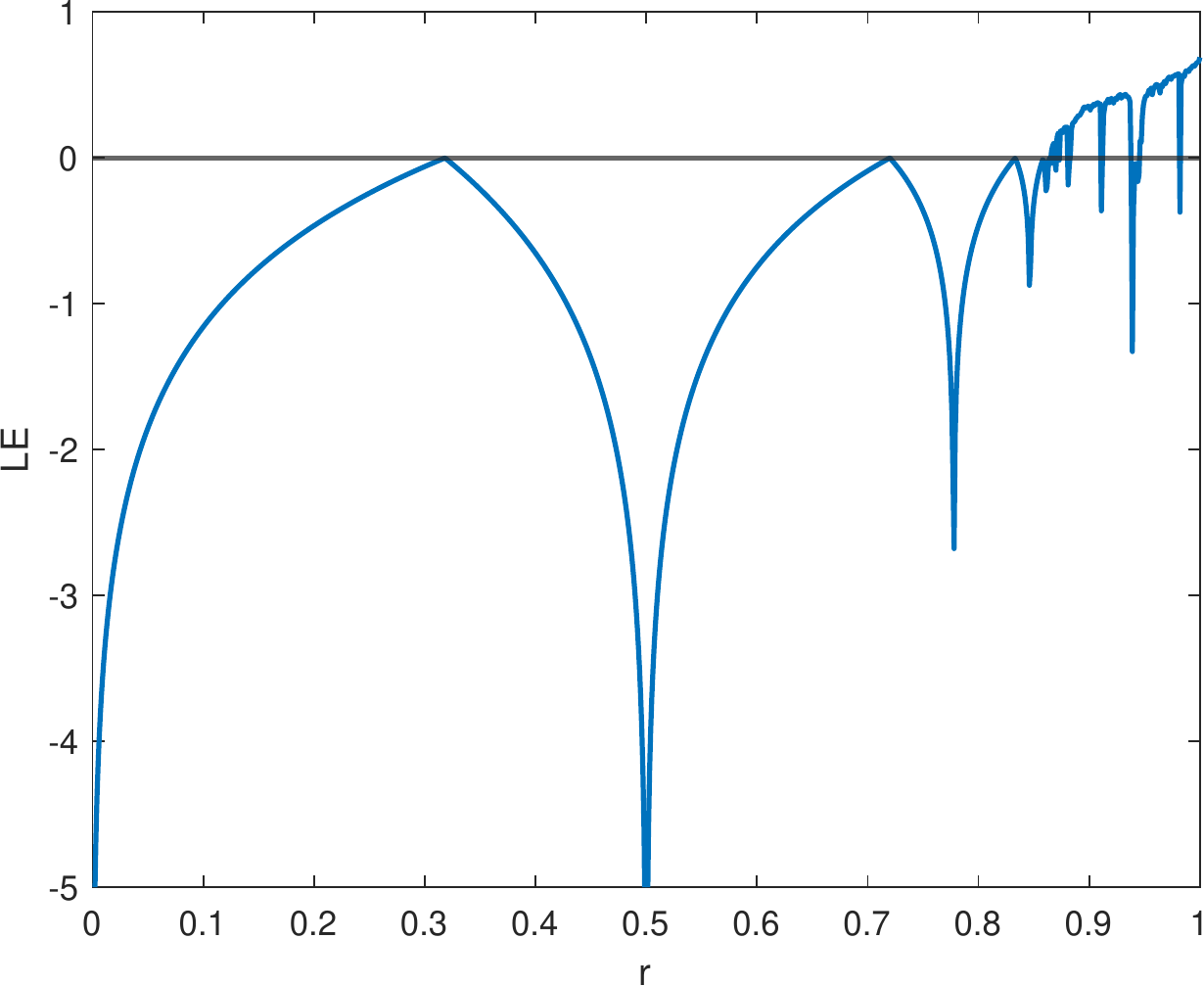}}
	\caption{Lyapunov exponents of three seed maps: (a) Logistic map; (b) Tent map; (c) Sine map.}
	\label{fig:LE_seed_maps}
\end{figure}

\begin{figure}[htbp]
	\centering
	\subfigure[]{
		\includegraphics[width=0.3\textwidth]{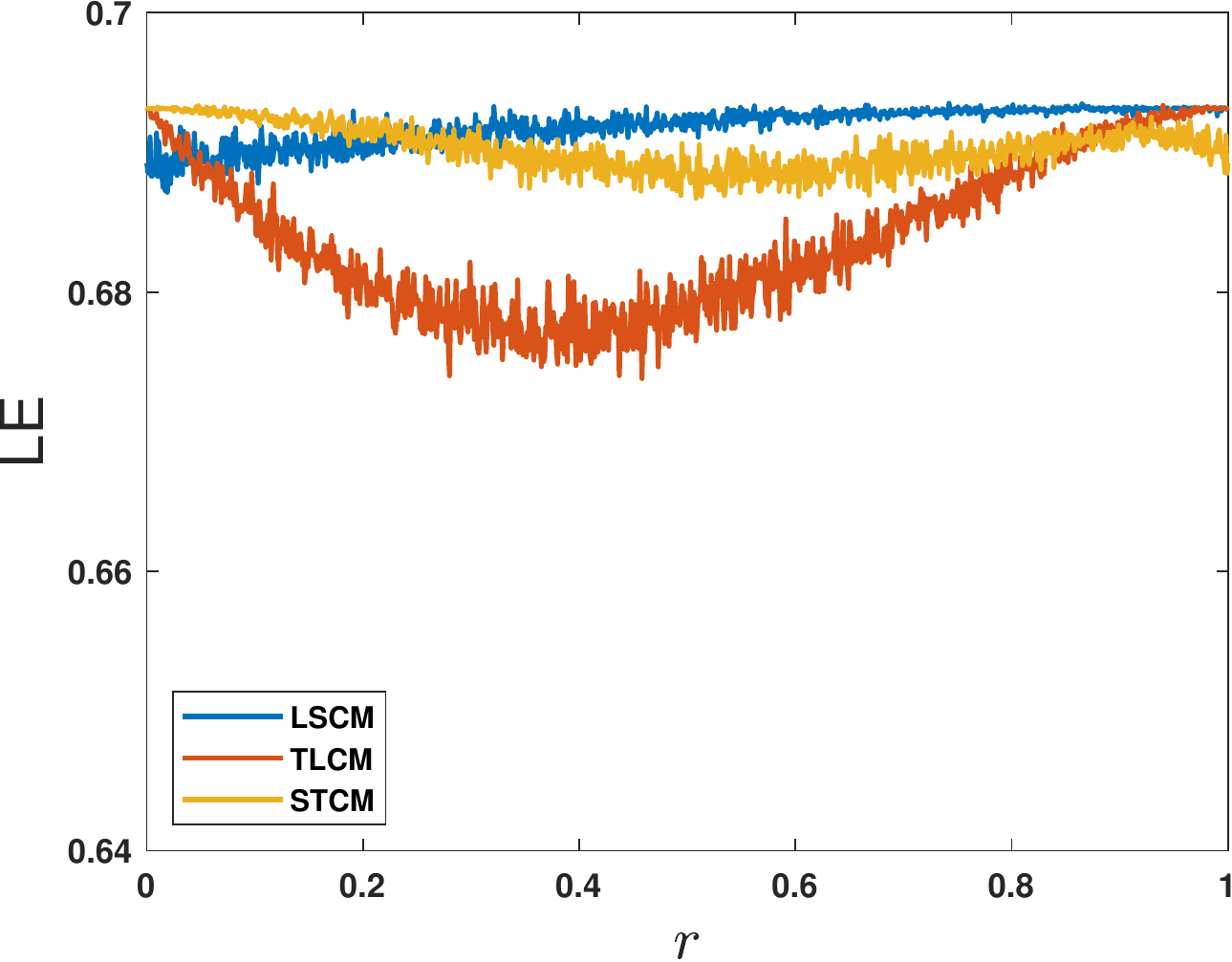}}
	\subfigure[]{
		\includegraphics[width=0.3\textwidth]{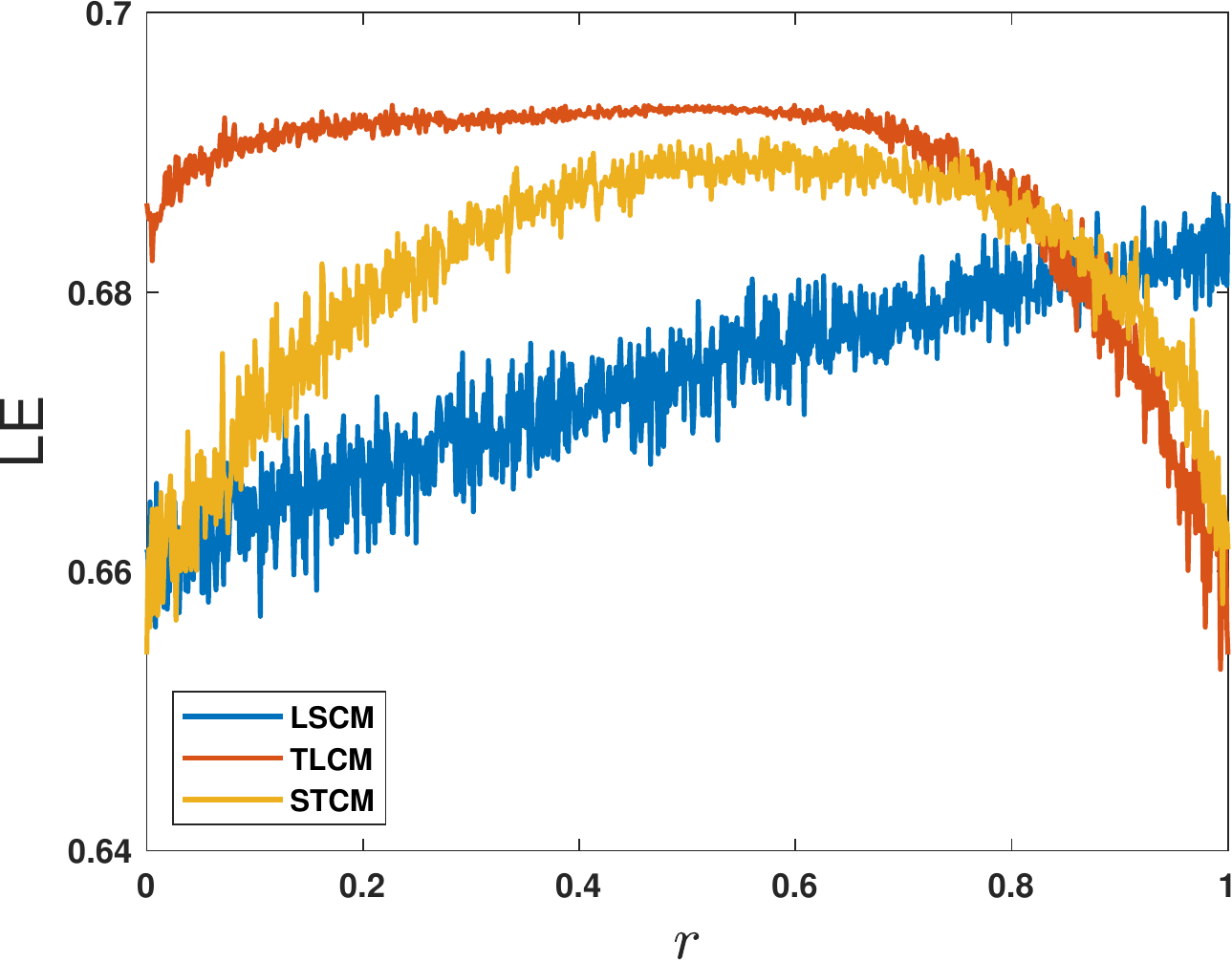}}
	\subfigure[]{
		\includegraphics[width=0.3\textwidth]{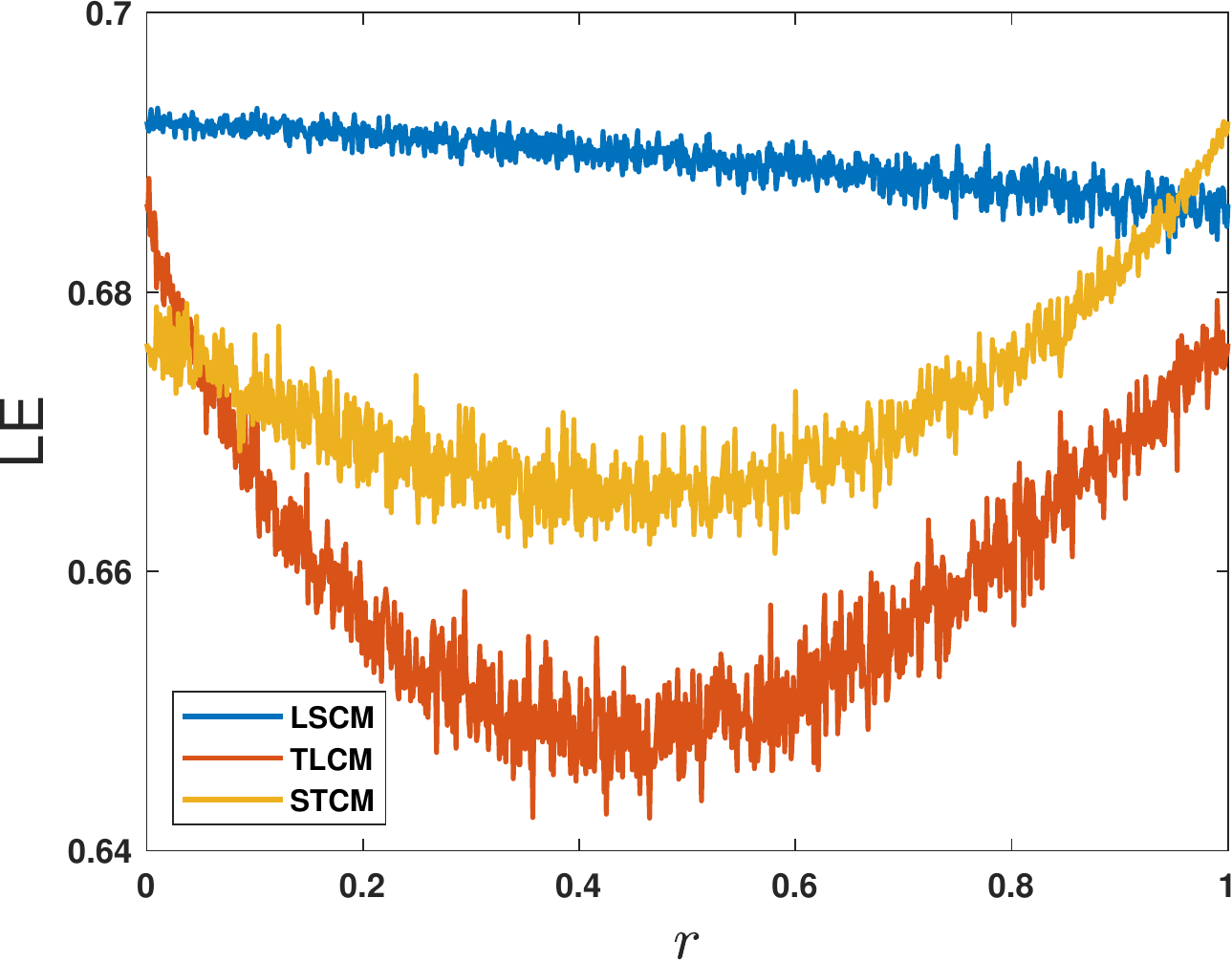}}
	\subfigure[]{
		\includegraphics[width=0.3\textwidth]{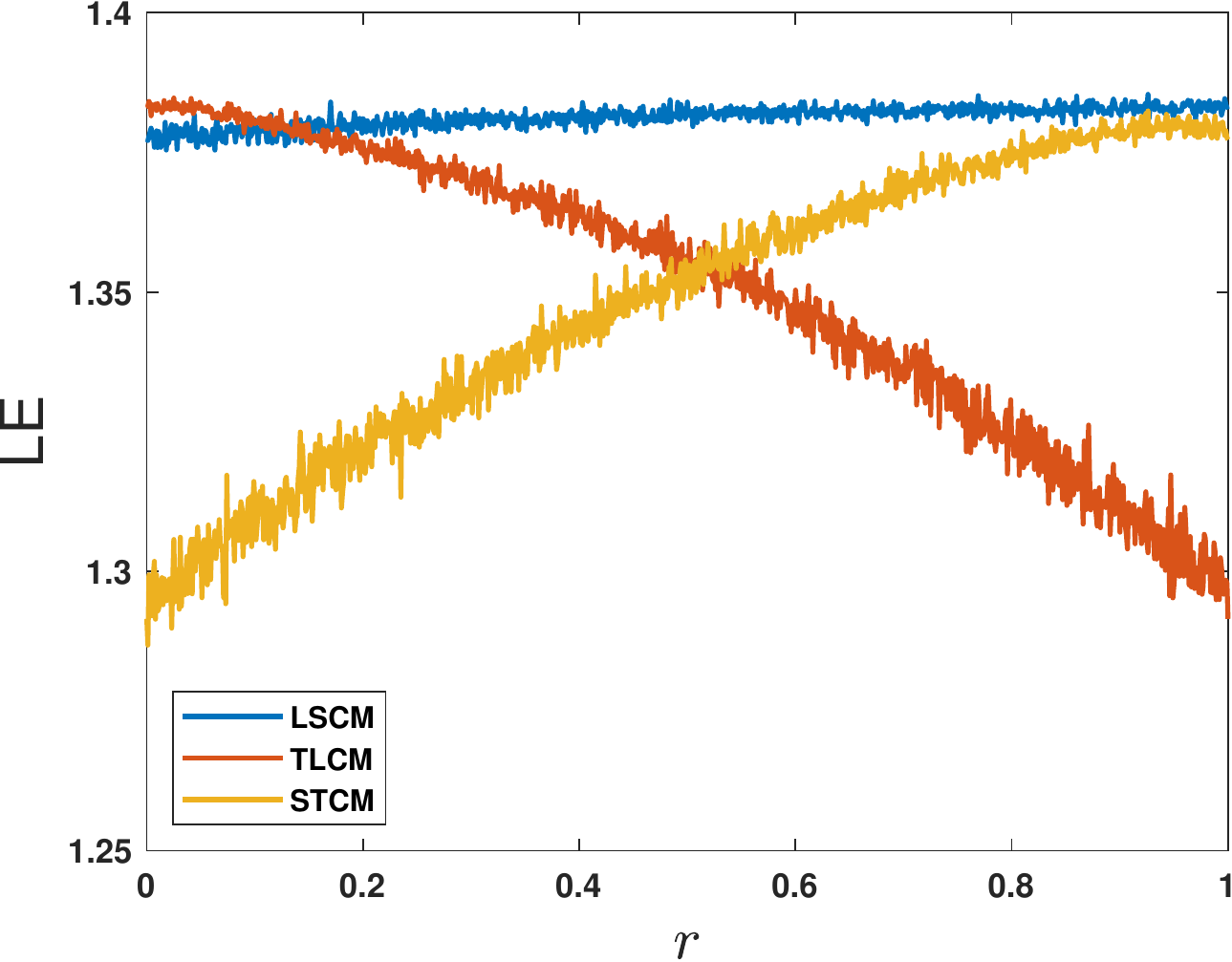}}
	\subfigure[]{
		\includegraphics[width=0.3\textwidth]{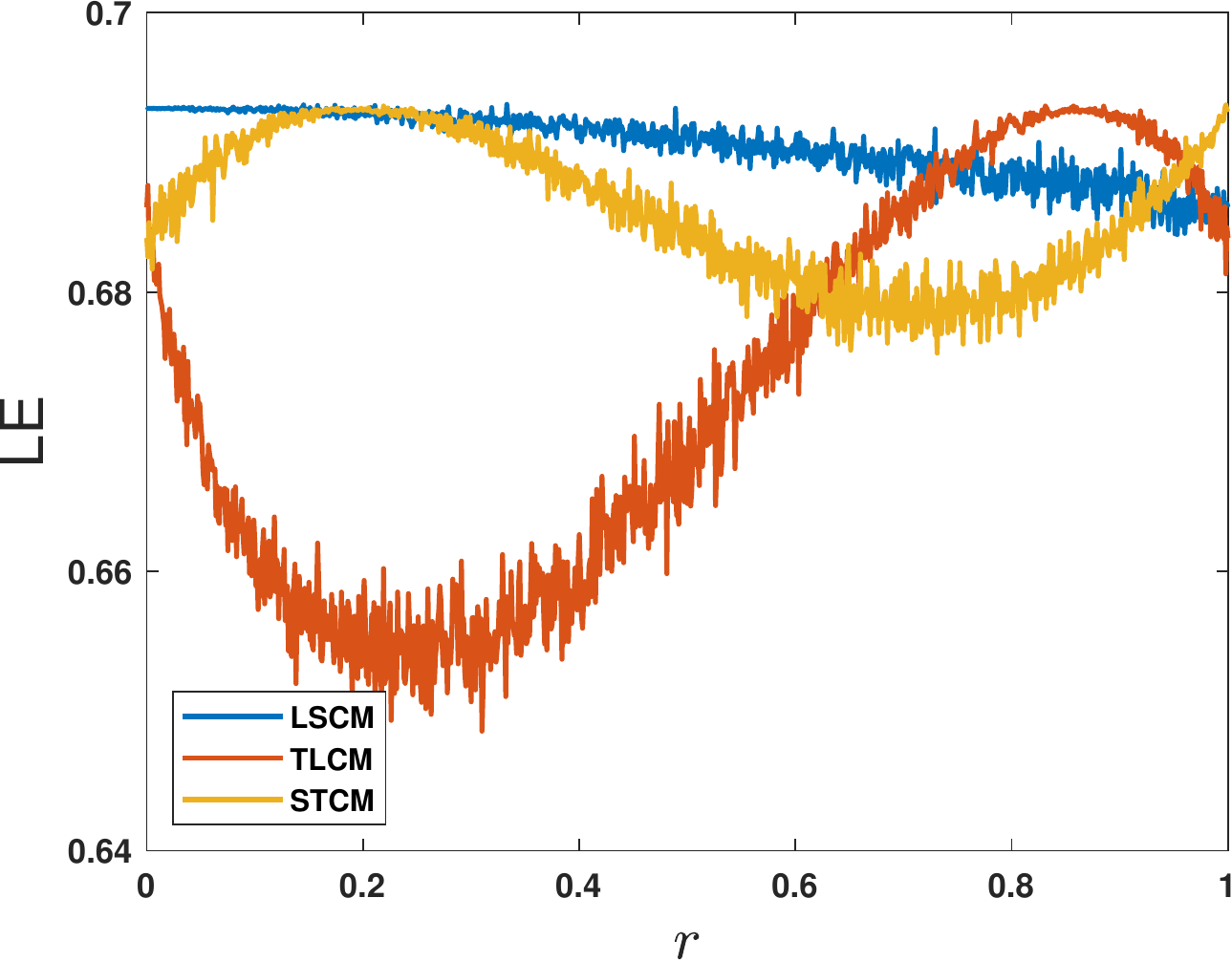}}
	\caption{Lyapunov exponents of new chaotic maps based on five BEFs: (a) Power function; (b) Exponential function; (c) Logarithmic function; (d) Trigonometric function; (e) Inverse trigonometric function.}
	\label{fig:LE_BEF}
\end{figure}

The LEs of new chaotic maps based on Eq.(\ref{eq:fx=2x})-(\ref{eq:fx=8x}) are shown in Figure \ref{fig:LE_COMP}. The results of Zhou's method \cite{zhouNew1DChaotic2014} and Hua's method \cite{huaCosinetransformbasedChaoticSystem2019} using the same seed maps are illustrated as comparison. As we can see, all the generated type-I, type-II, and type-III coupling maps are chaotic for all $r \in [0,1]$ and their LEs are lager than Zhou's method. Also, the type-I with LSCM, TLCM, and STCM have LEs close to Hua's method while the LEs of type-II and type-III coupling maps are far larger than Hua's method. Because larger LE indicates more complex chaotic behavior, our UT-CCS is of the best performance. Most importantly, Figure \ref{fig:LE_COMP} also tells us that all the type-II chaotic maps have larger LEs than type-I and all the type-III chaotic maps have larger LEs than type-II, which demonstrates the correctness of Theorem \ref{th:1}. Thus, we can even get new chaotic maps with further larger LEs by constructing proper $f(x)$.

\begin{figure}[htbp]
	\centering
	\subfigure[]{
		\label{fig:LE_LSCM}
		\includegraphics[width=0.3\textwidth]{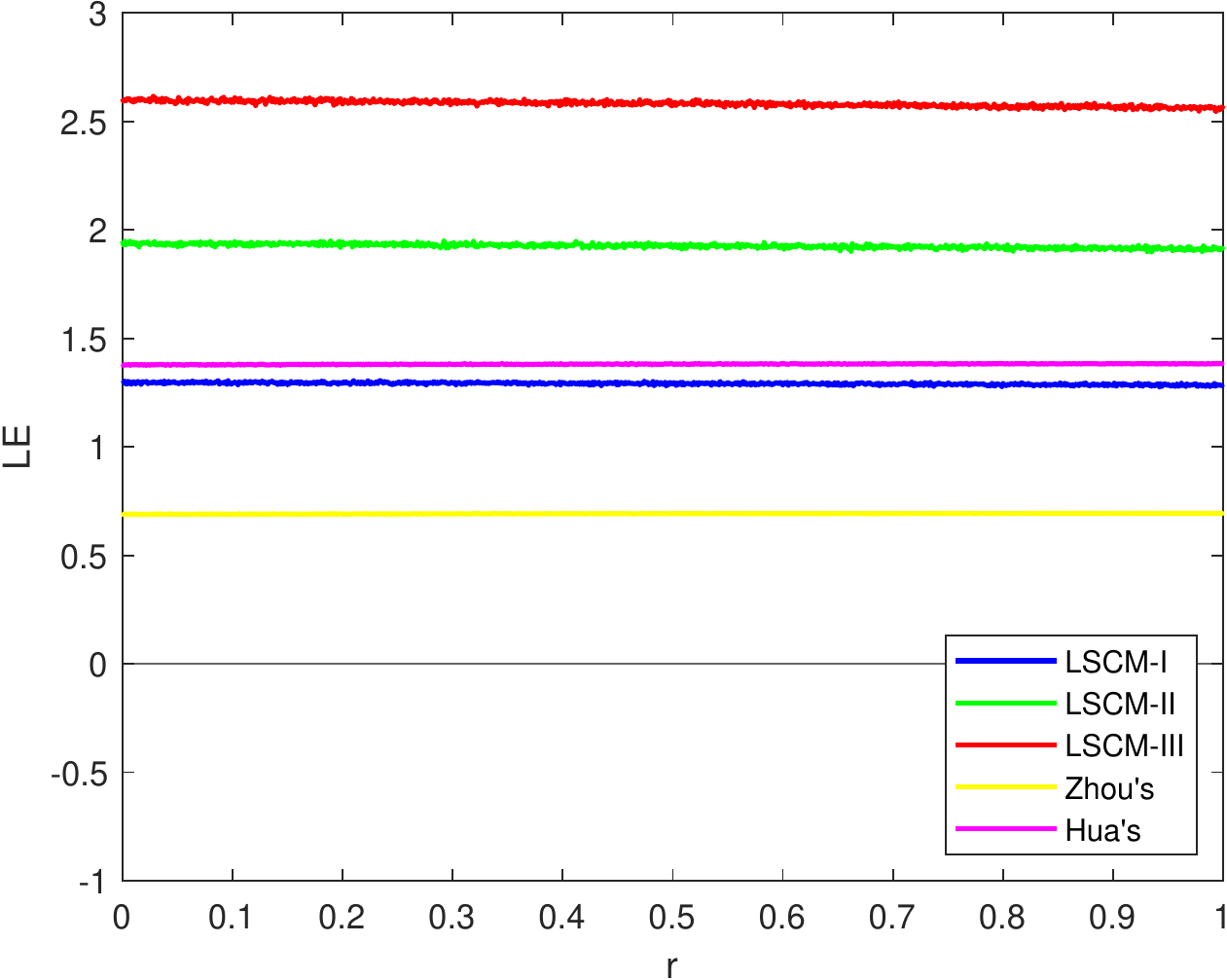}}
	\subfigure[]{
		\label{fig:LE_TLCM}
		\includegraphics[width=0.3\textwidth]{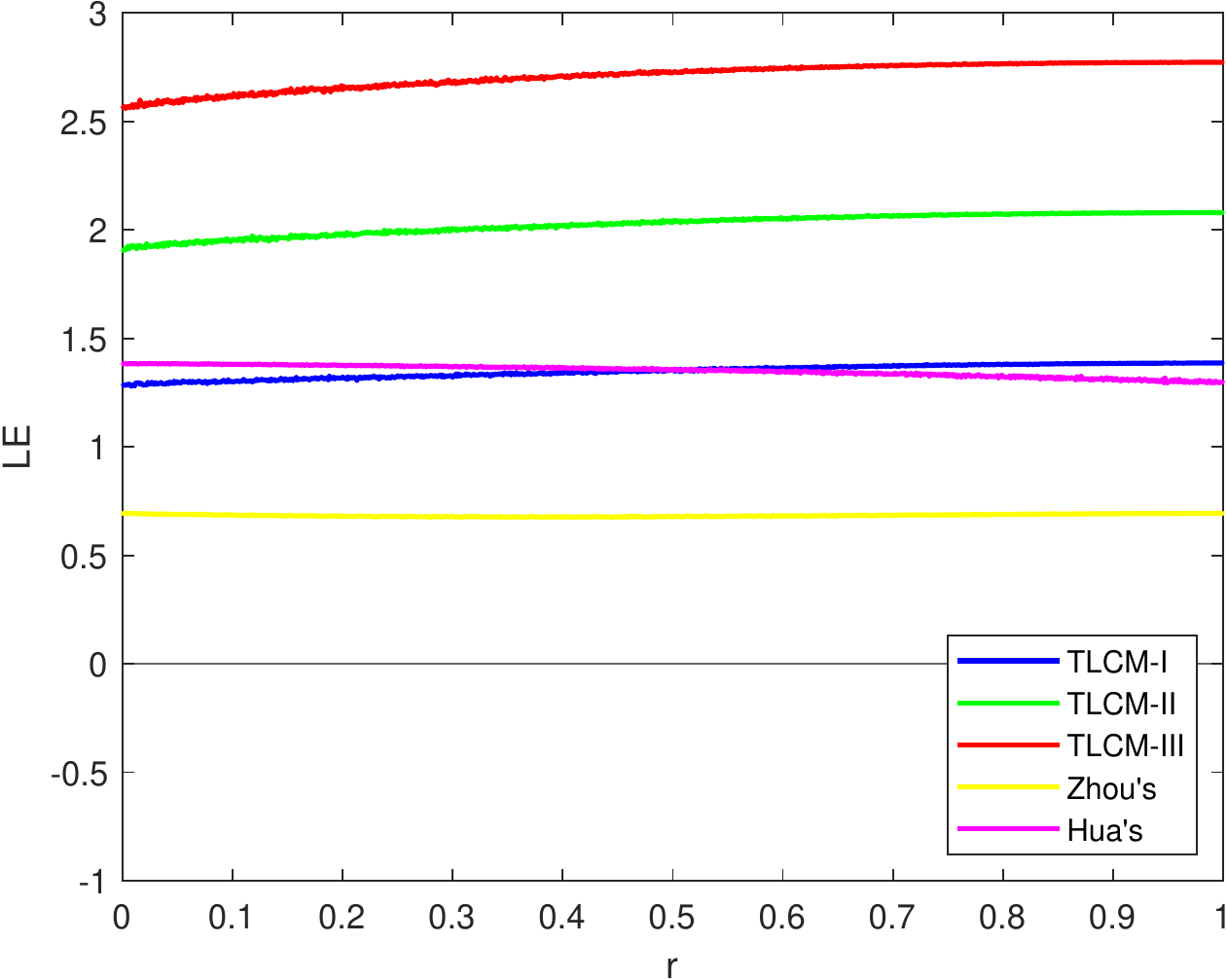}}
	\subfigure[]{
		\label{fig:LE_STCM}
		\includegraphics[width=0.3\textwidth]{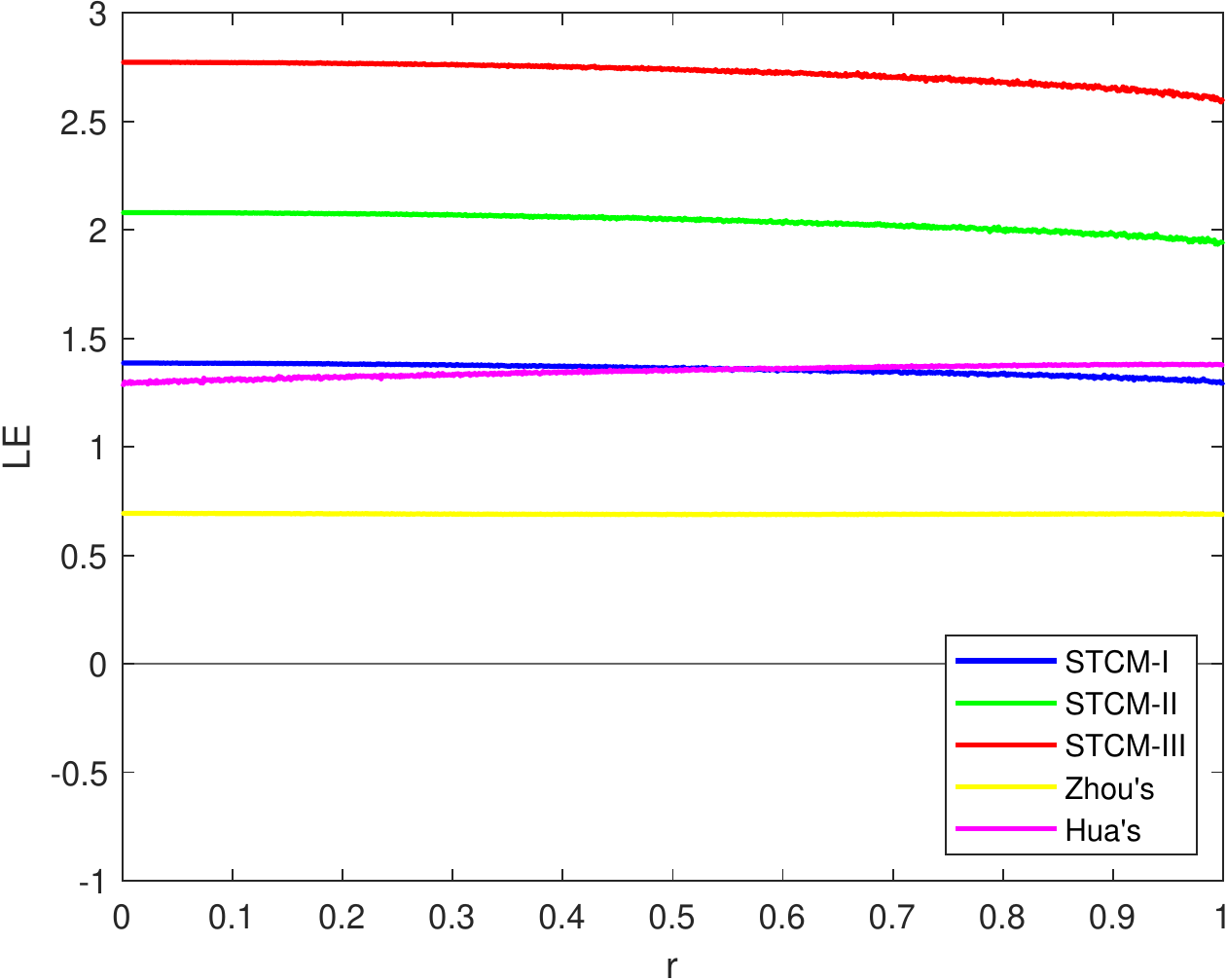}}
	\caption{Comparison of LEs between different generated chaotic maps based on same seed maps: (a) LSCM; (b) TLCM; (c) STCM.}
	\label{fig:LE_COMP}
\end{figure}

\subsection{Bifurcation diagram}
Bifurcation diagram presents the distribution of points in different trajectories of a dynamic system visually with the change of parameter. According to the ergodicity of visited points of an orbit, we will know whether the system is chaotic. The BDs of type-III coupling maps as well as three seed maps are shown in Figure \ref{fig:BD_COMP}, where the initial value is set as $x_0=0.1$. It follows that the seed maps are non-chaotic for most $r \in [0,1]$, since most orbits are fixed or periodic. However, the points of all type-III coupling maps fill the whole phase plane, suggesting the strong chaotic behaviors for all $r \in [0,1]$. 

\begin{figure}[htbp]
	\centering
	\subfigure[]{
		\label{fig:BD_Logistic}
		\includegraphics[width=0.3\textwidth]{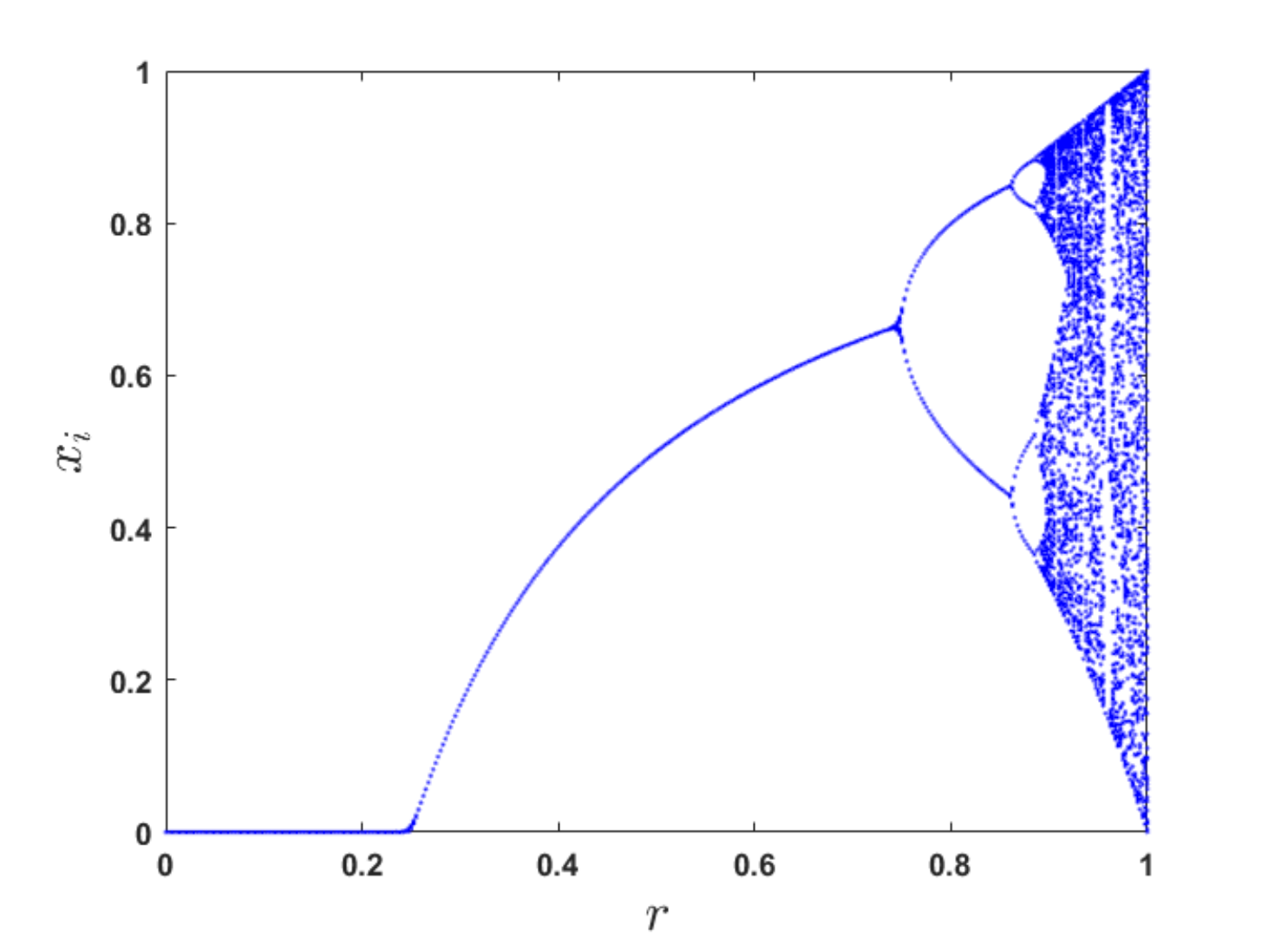}}
	\subfigure[]{
		\label{fig:BD_Tent}
		\includegraphics[width=0.3\textwidth]{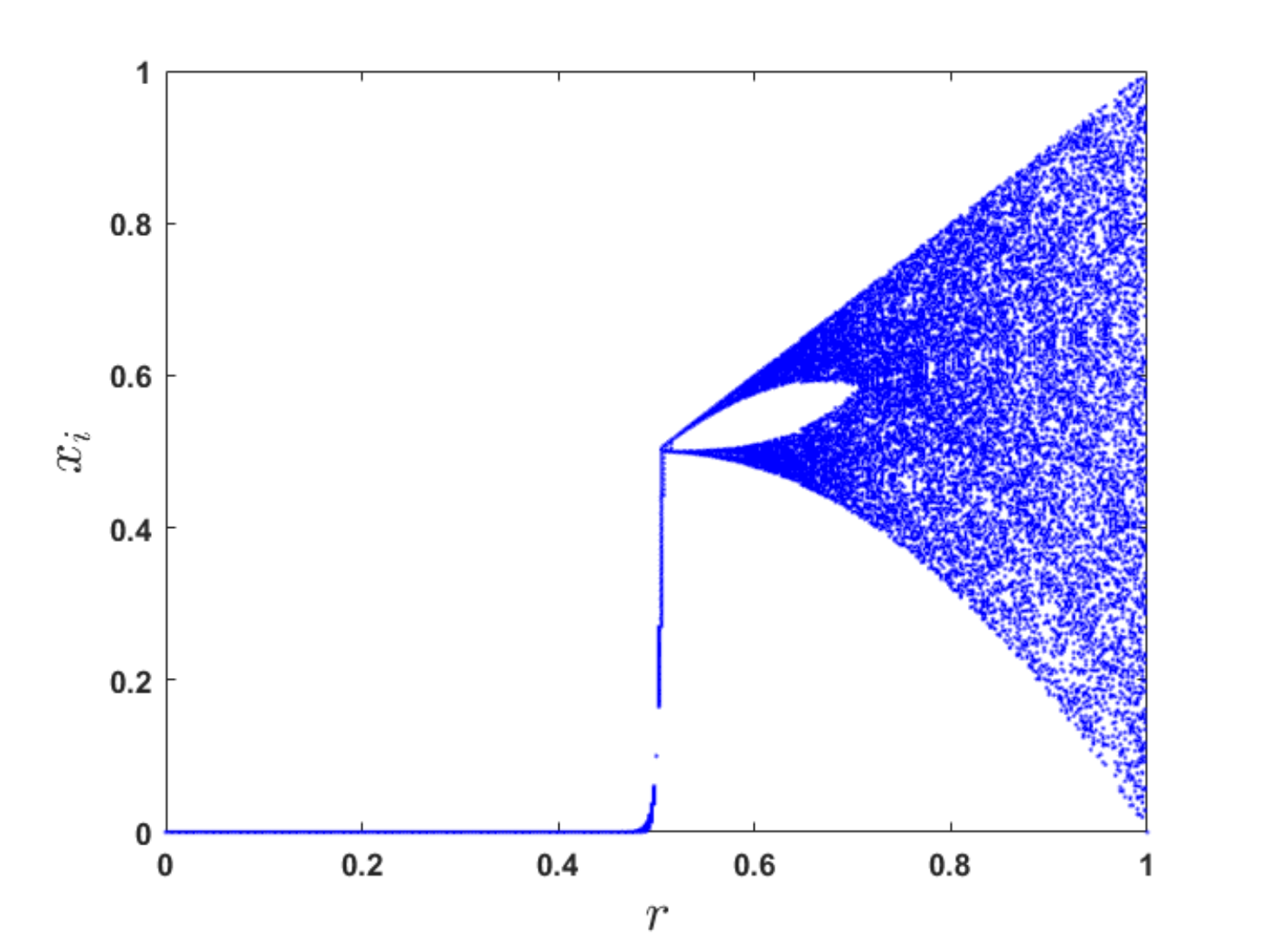}}
	\subfigure[]{
		\label{fig:BD_Sine}
		\includegraphics[width=0.3\textwidth]{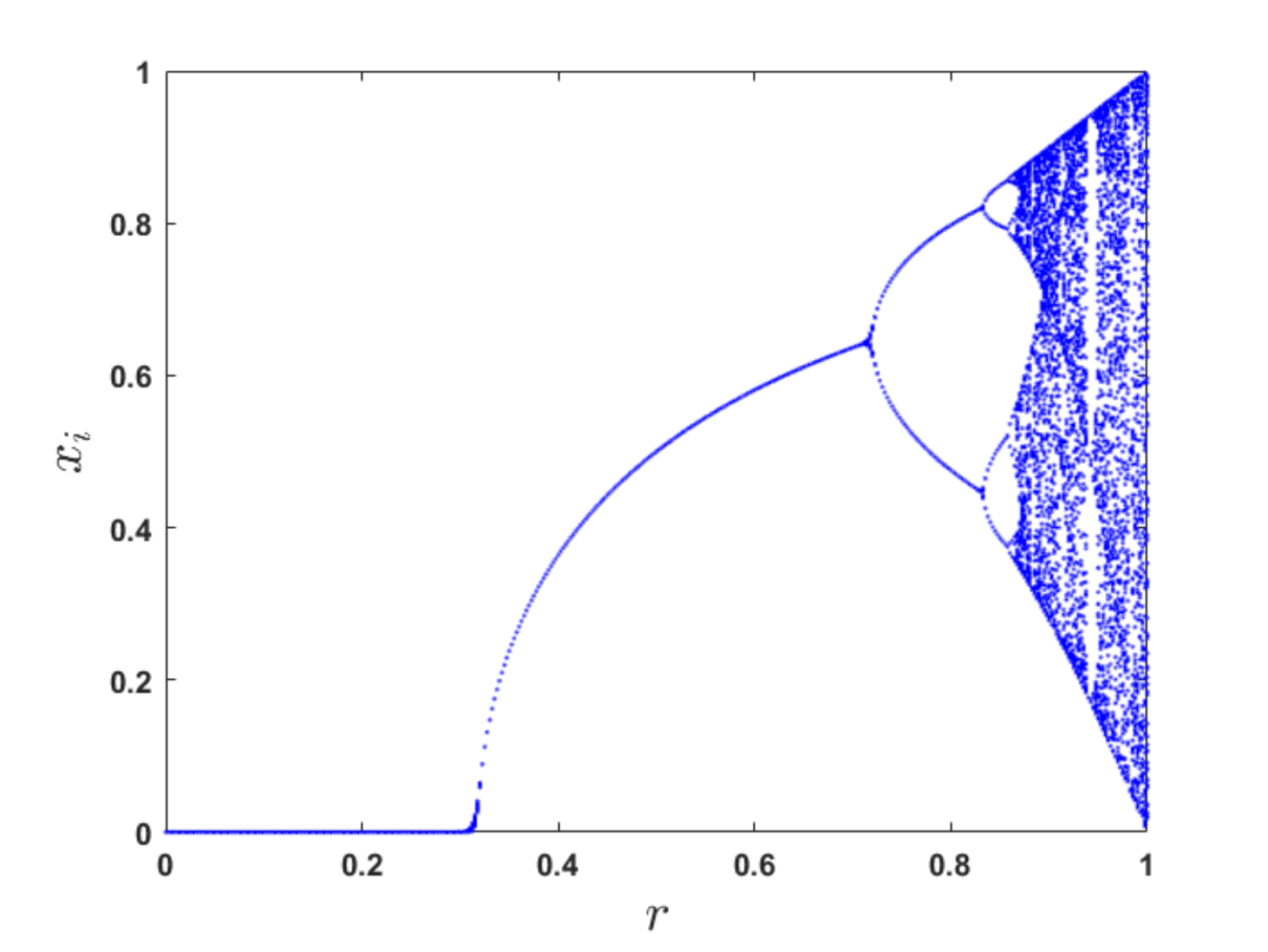}}
	\subfigure[]{
		\label{fig:BD_LSCM}
		\includegraphics[width=0.3\textwidth]{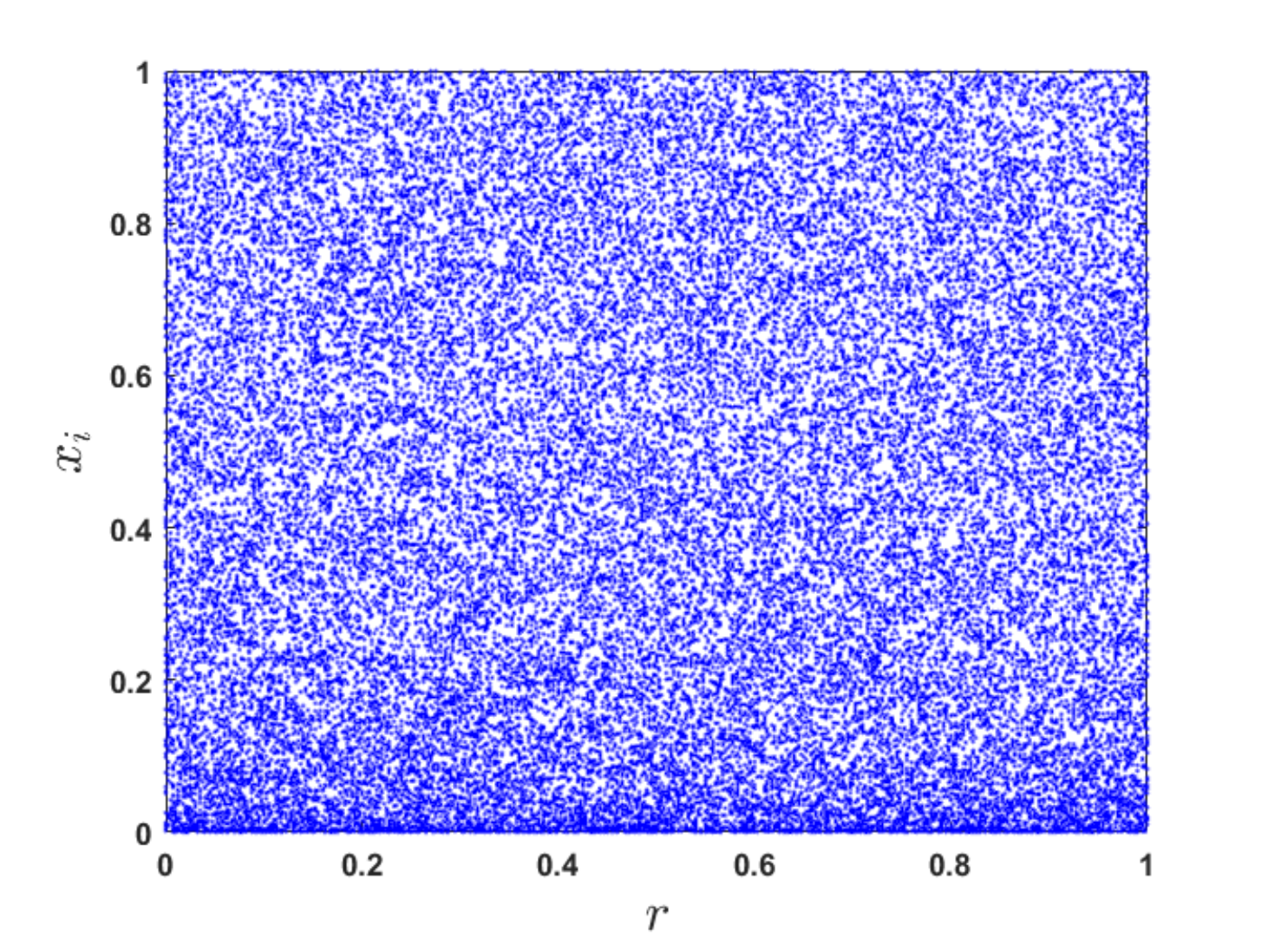}}
	\subfigure[]{
		\label{fig:BD_TLCM}
		\includegraphics[width=0.3\textwidth]{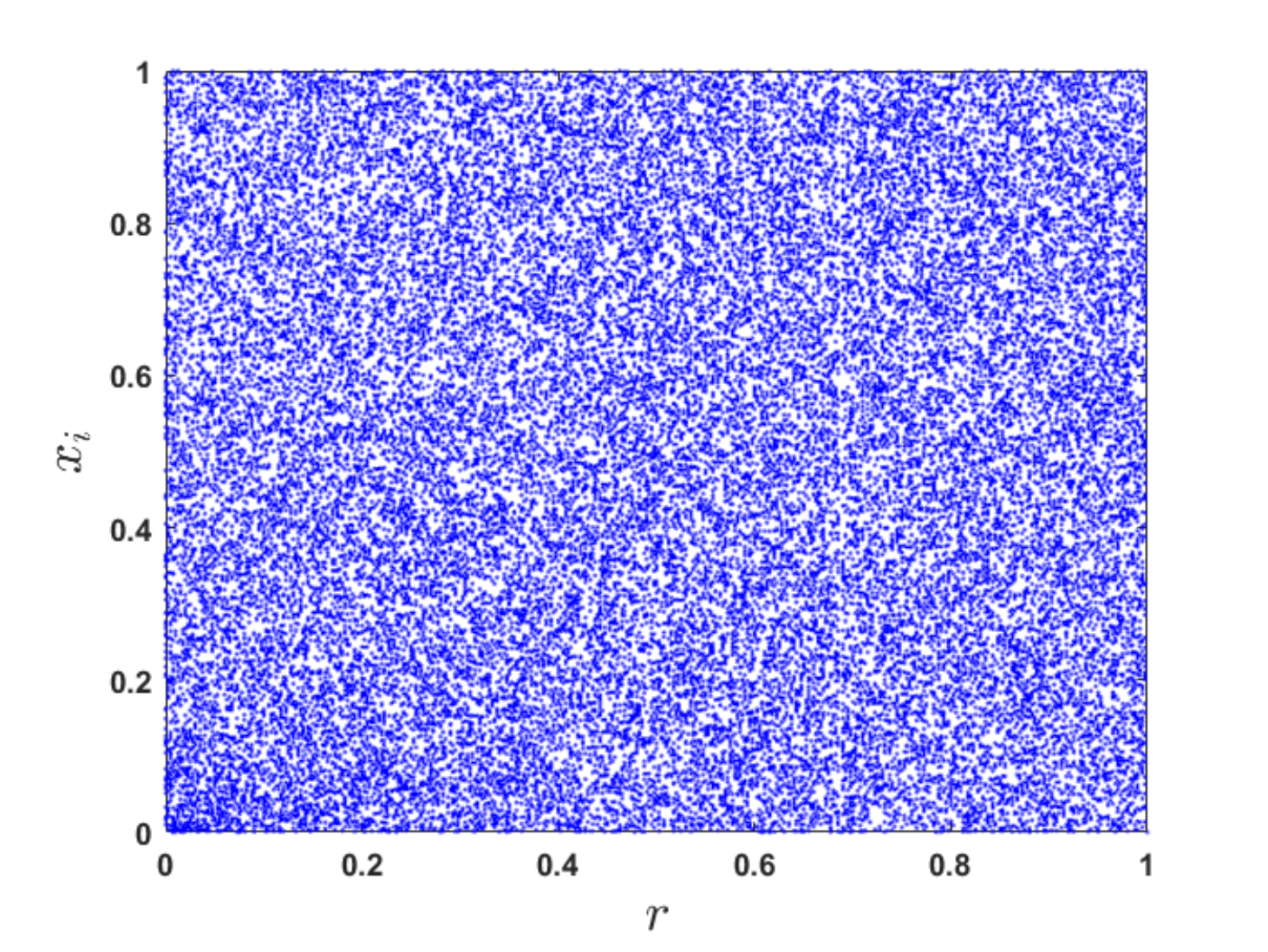}}
	\subfigure[]{
		\label{fig:BD_STCM}
		\includegraphics[width=0.3\textwidth]{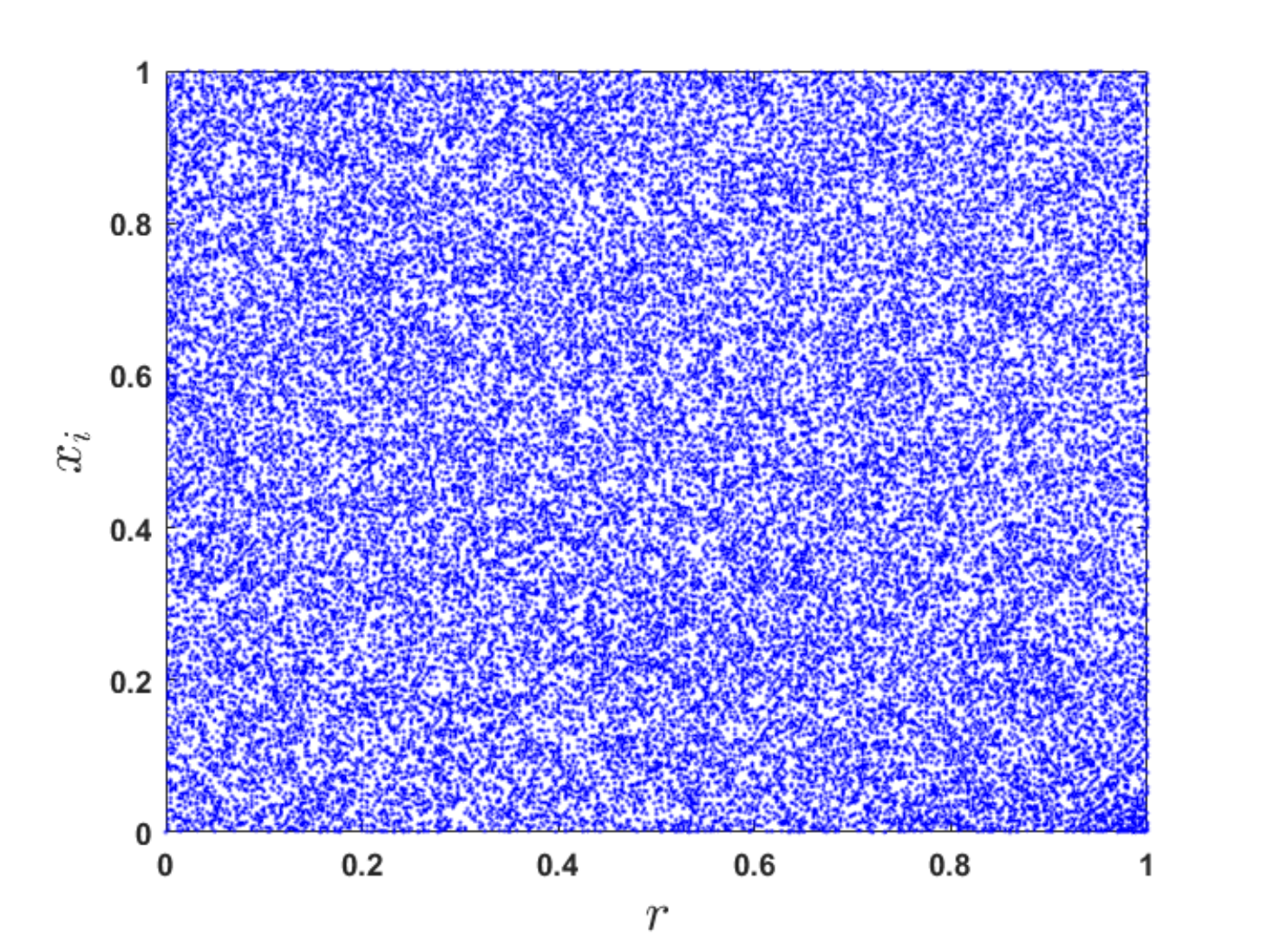}}
	\caption{Bifurcation diagrams of different chaotic maps: (a) Logistic map; (b) Tent map; (c) Sine map; (d) LSCM-III; (e) TLCM-III; (f) STCM-III.}
	\label{fig:BD_COMP}
\end{figure}

\subsection{Cobweb diagram}
Cobweb diagram is a plotting scheme developed to present the behavior of 1D recursions intuitively. Suppose a recursion is defined by $x_{n+1}=F(x_n)$, then the cobweb diagram can be gained by following steps:
\begin{enumerate}[Step 1.]
	\setlength{\itemindent}{1.5em}
	\item Plot the diagonal line $y=x$ and the graph of $y=F(x)$.
	\item Plot the vertical line segment from $(x_i,x_i)$ to $(x_i,x_{i+1})$.
	\item Plot the horizontal line segment from $(x_i,x_{i+1})$ to $(x_{i+1},x_{i+1})$.
	\item Repeat Step 2 and 3 with iteration.
\end{enumerate}
Generally, there are three results for any recursion after long-run iteration: converging to a fixed point, being cyclic, and reaching chaotic. Correspondingly, the cobweb plot will show inward square spiral, fixed track, and largely filled plane, respectively. Figure \ref{fig:CP_COMP} has illustrated the Cobweb plots of three type-III coupling maps and their seed maps with parameter $r=0.8$. It can be seen easily that all the generated coupling maps are chaotic while Logistic map and Sine map degrade into cyclic. The Tent map is chaotic when $r=0.8$ but it can't traverse $[0,1]$ with iteration.
\begin{figure}[htbp]
	\centering
	\subfigure[]{
		\label{fig:CP_Logistic}
		\includegraphics[width=0.3\textwidth]{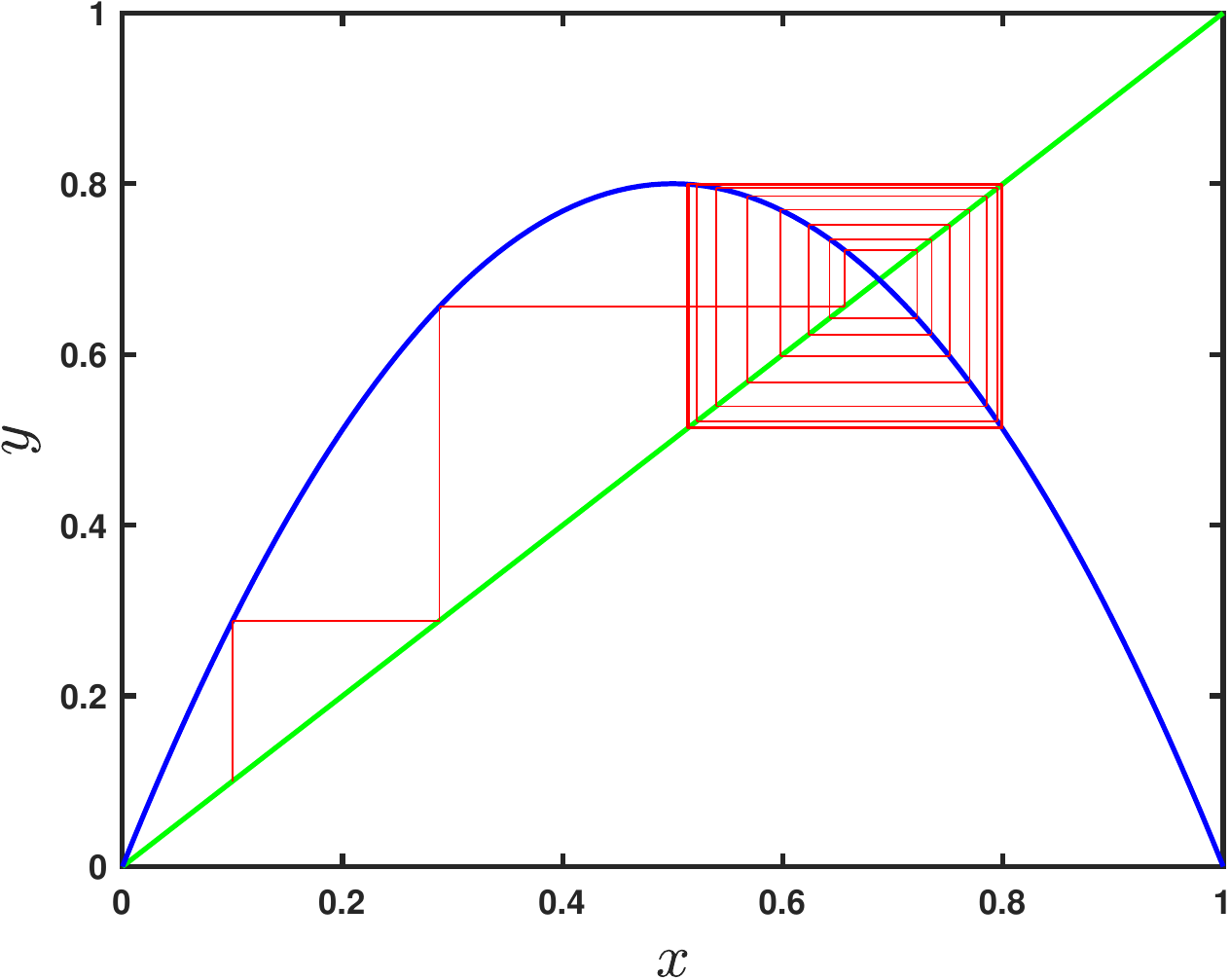}}
	\subfigure[]{
		\label{fig:CP_Tent}
		\includegraphics[width=0.3\textwidth]{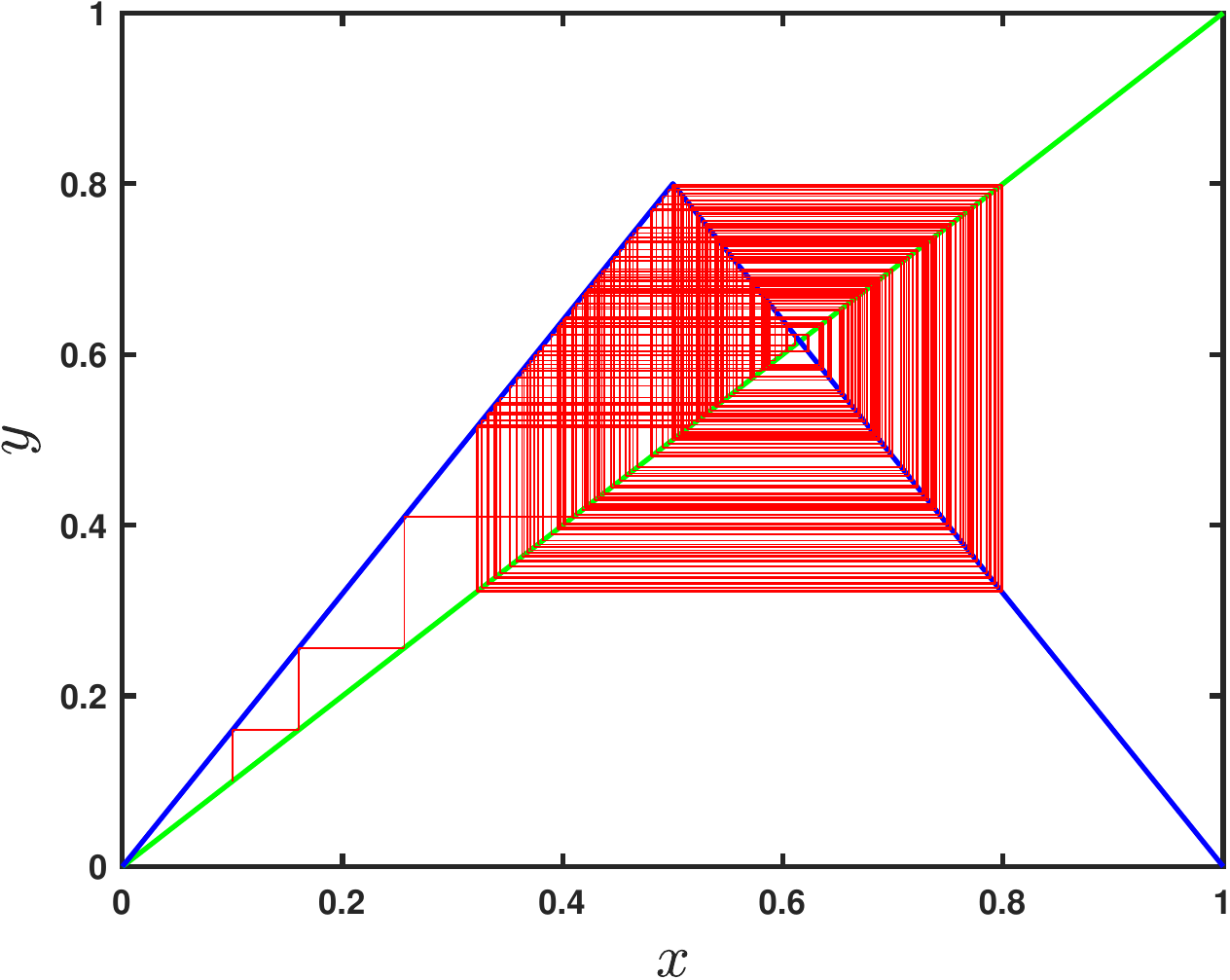}}
	\subfigure[]{
		\label{fig:CP_Sine}
		\includegraphics[width=0.3\textwidth]{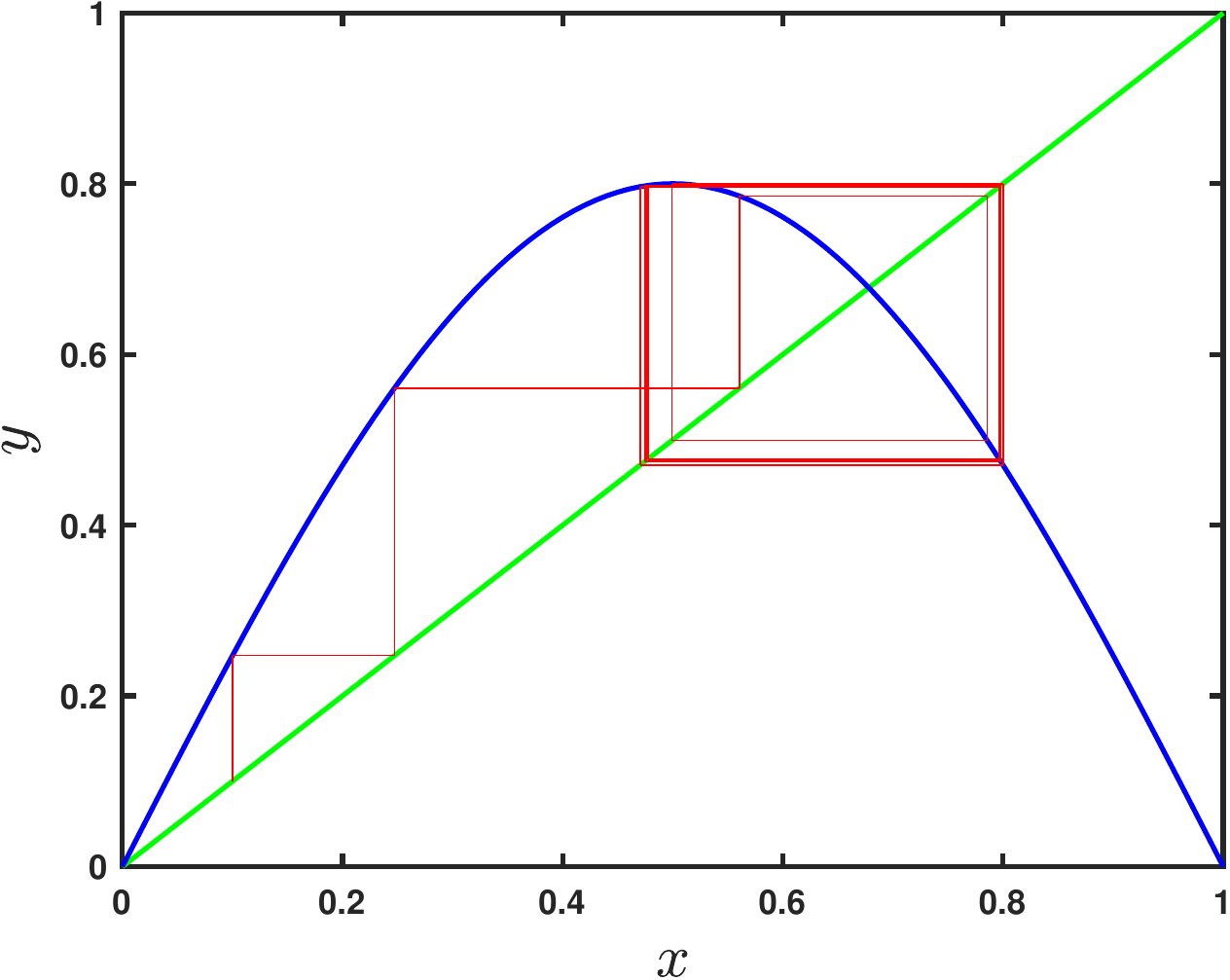}}
	\subfigure[]{
		\label{fig:CP_LSCM}
		\includegraphics[width=0.3\textwidth]{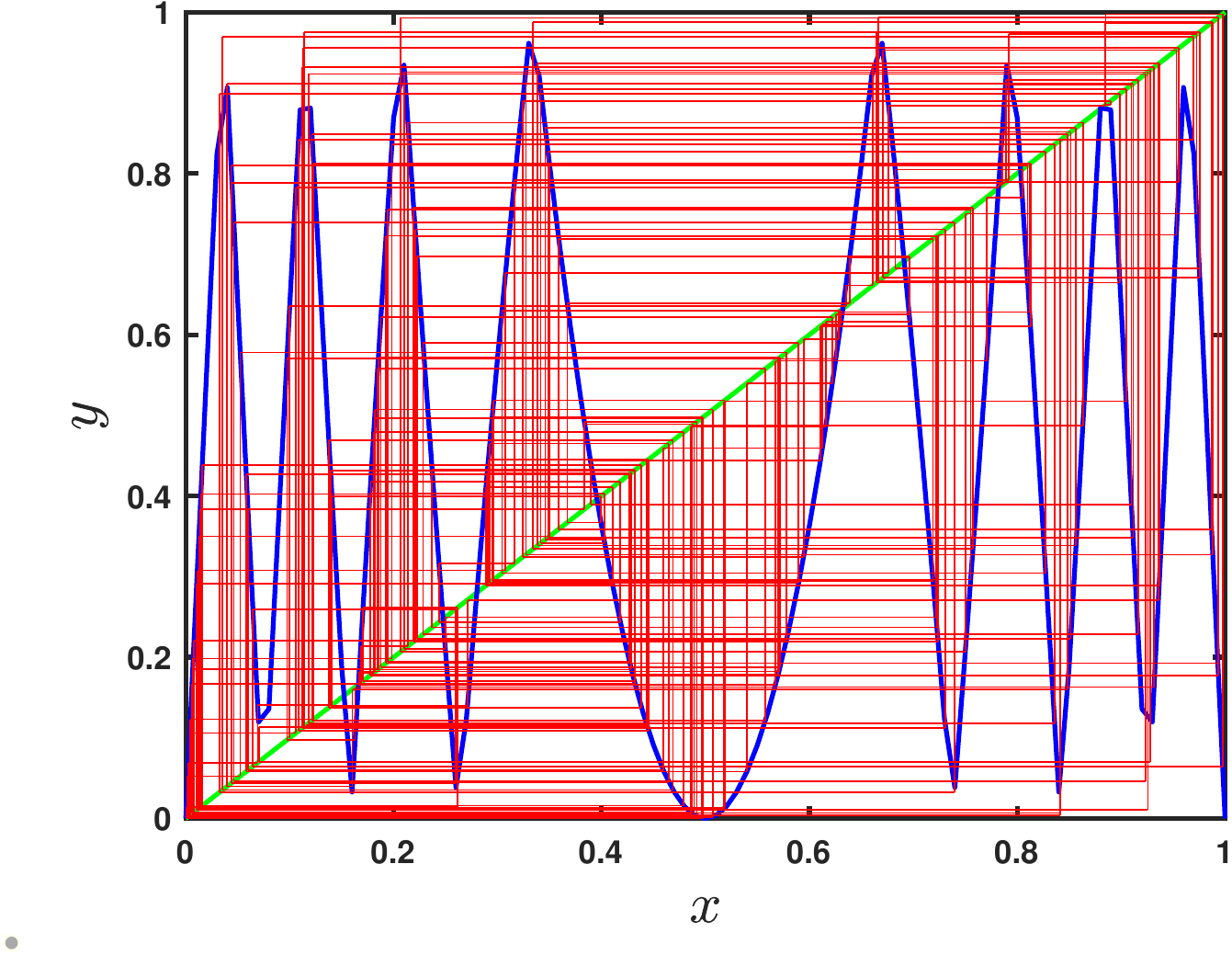}}
	\subfigure[]{
		\label{fig:CP_TLCM}
		\includegraphics[width=0.3\textwidth]{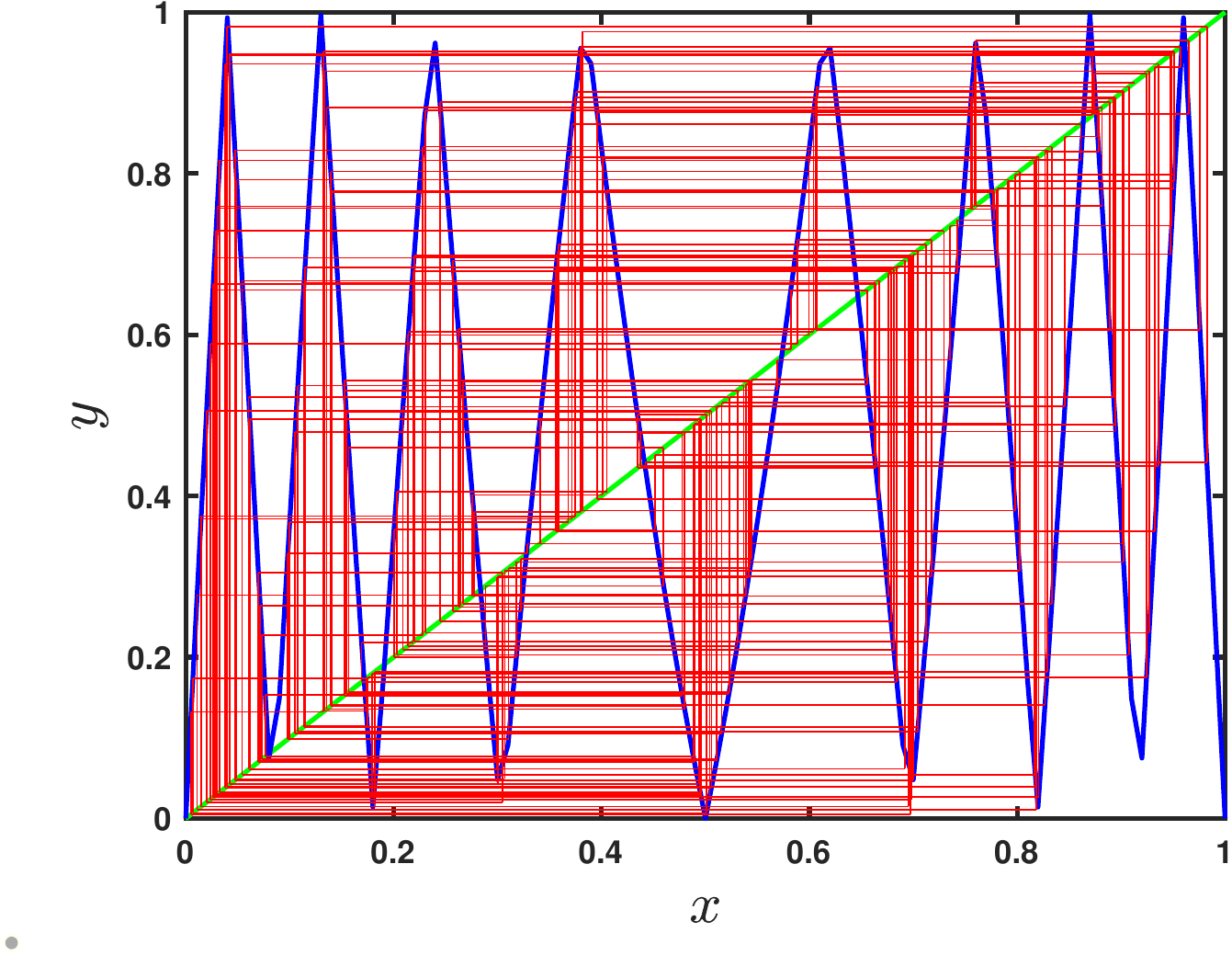}}
	\subfigure[]{
		\label{fig:CP_STCM}
		\includegraphics[width=0.3\textwidth]{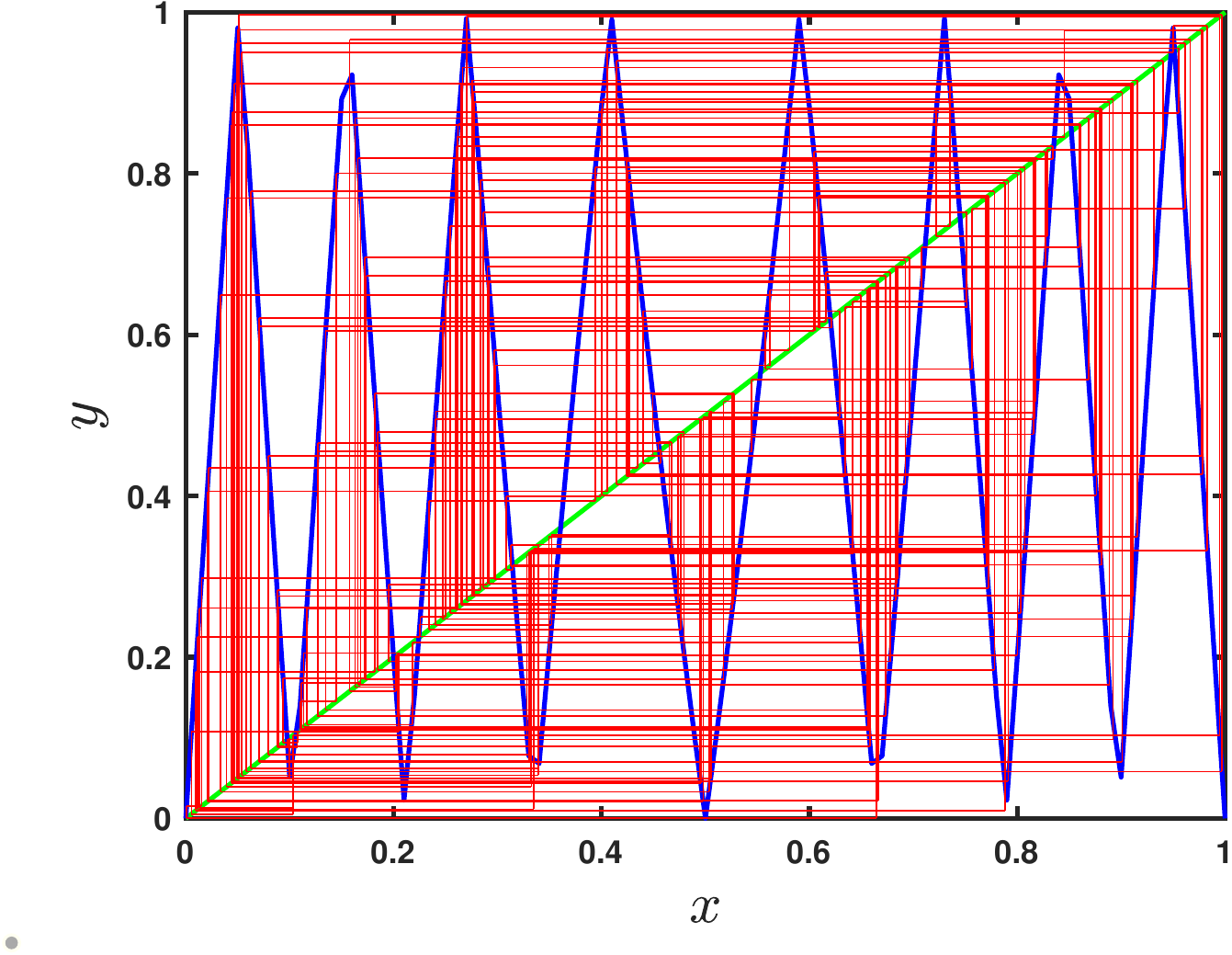}}
	\caption{Cobweb diagrams of different chaotic maps: (a) Logistic map; (b) Tent map; (c) Sine map; (d) LSCM-III; (e) TLCM-III; (f) STCM-III.}
	\label{fig:CP_COMP}
\end{figure}
\begin{figure}[htbp]
	\centering
	\subfigure[LSCM]{
		\label{}
		\includegraphics[width=0.3\textwidth]{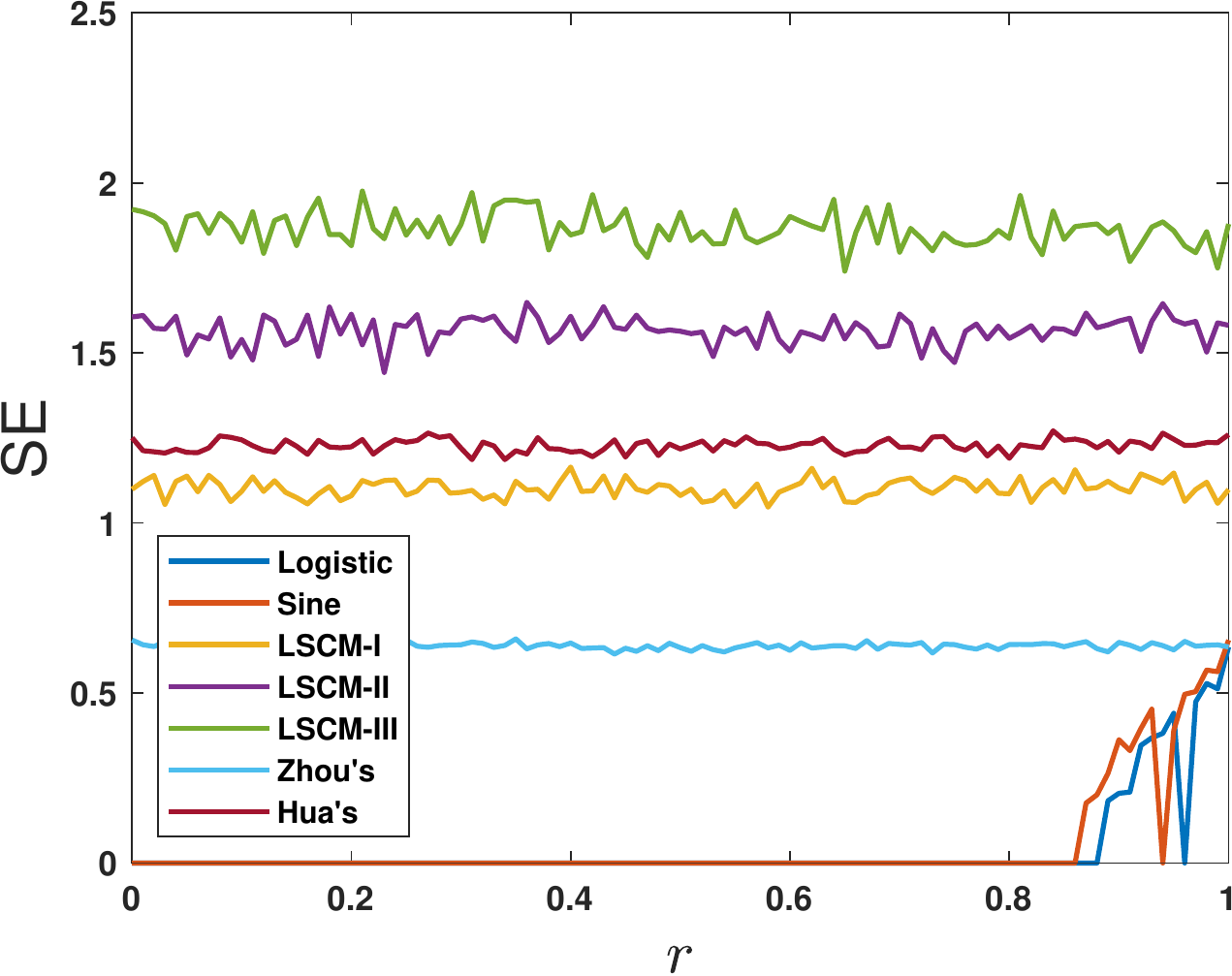}}
	\subfigure[TLCM]{
		\label{}
		\includegraphics[width=0.3\textwidth]{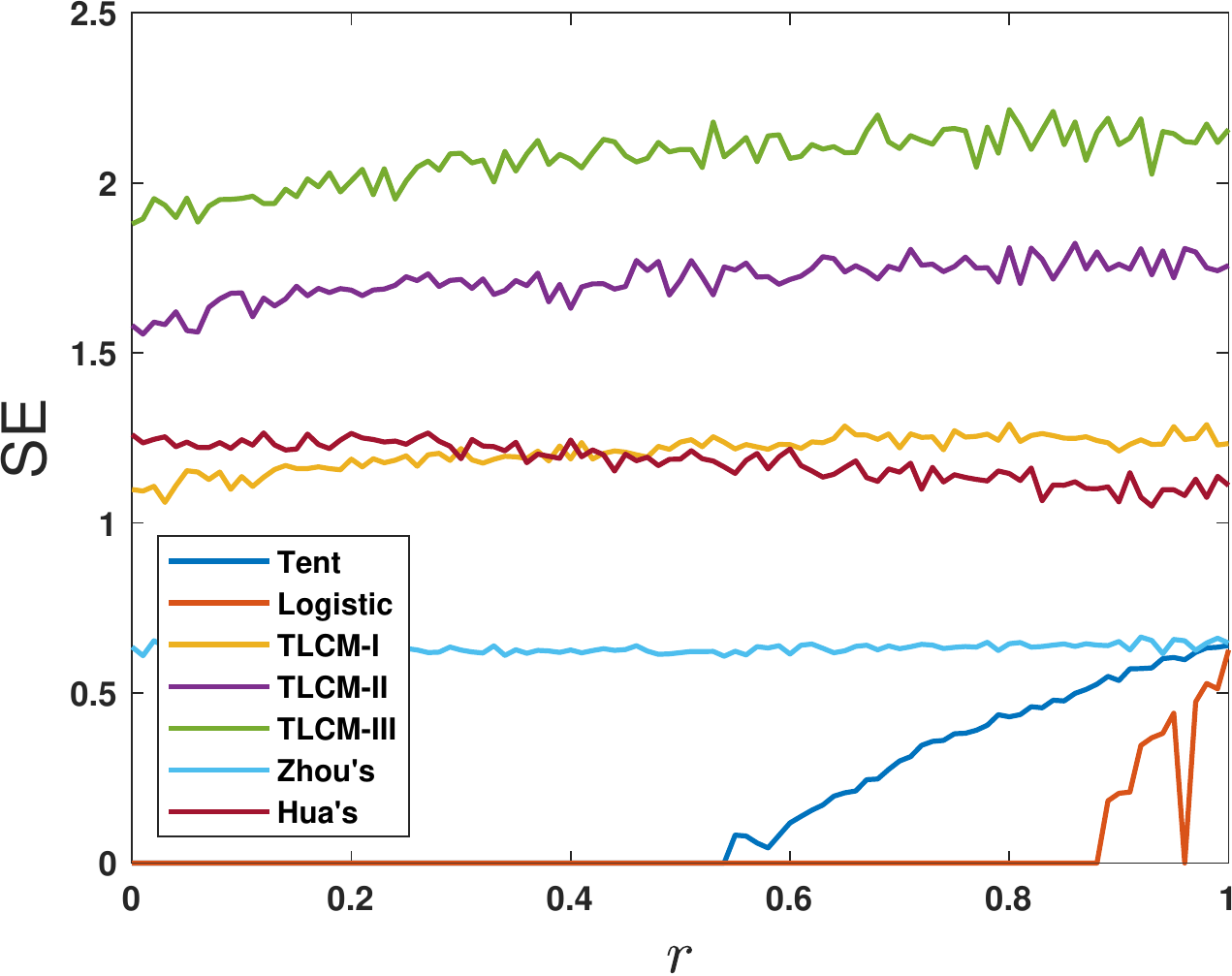}}
	\subfigure[STCM]{
		\label{}
		\includegraphics[width=0.3\textwidth]{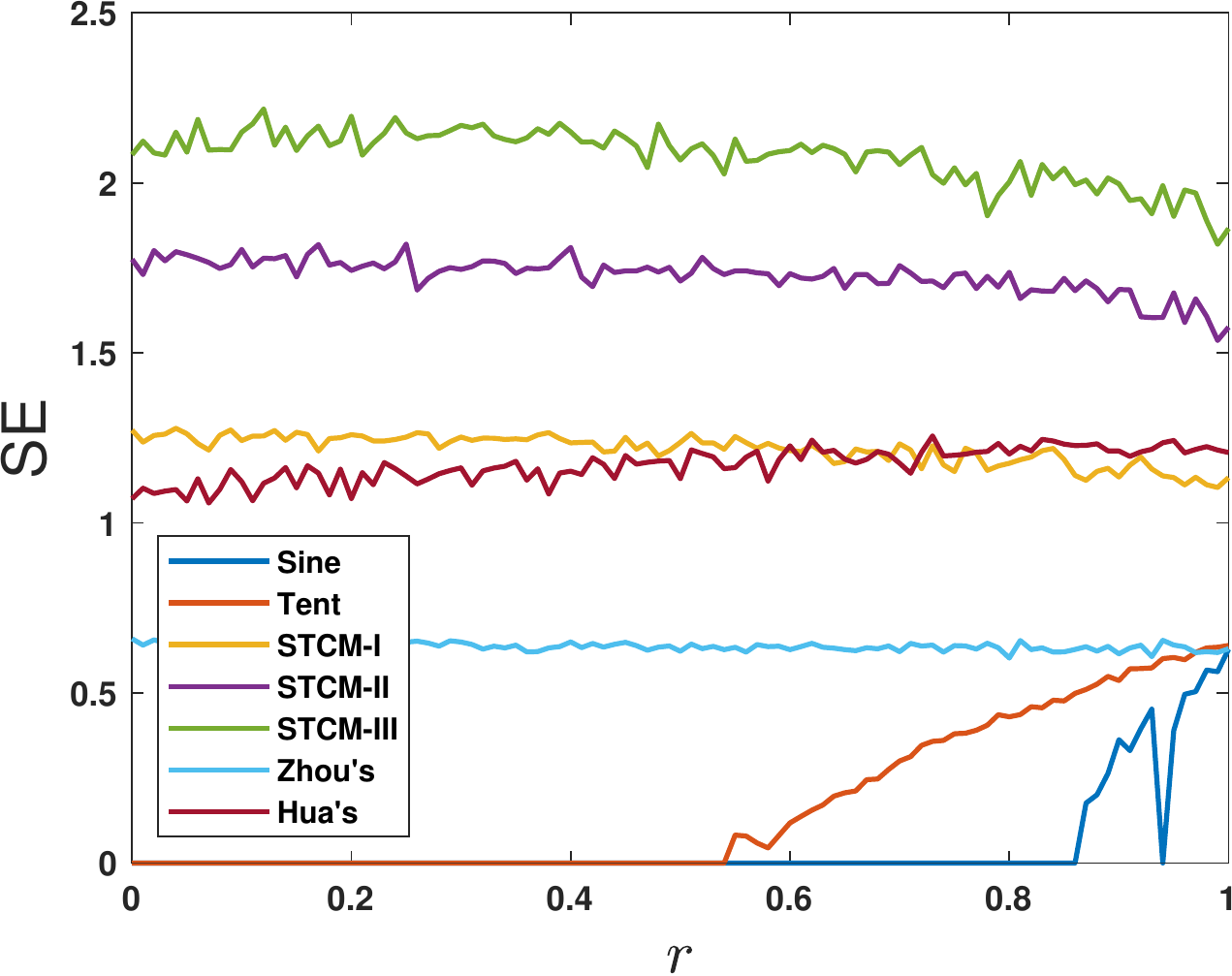}}
	\caption{Comparison of LEs between different generated chaotic maps based on same seed maps: (a) LSCM; (b) TLCM; (c) STCM.}
	\label{fig:SE_COMP}
\end{figure}
\subsection{Sample entropy}
Sample entropy is an improved version of approximate entropy, both of which give a quantitative measurement of the complexity of a time series \cite{richmanPhysiologicalTimeseriesAnalysis2000a}. The definition and computing method of SE can refer to \cite{huaCosinetransformbasedChaoticSystem2019}. Actually, SE describes the probability of emergence of new schema in time series. A series is considered to be more complex if new schema is more likely to appear. Therefore, a larger SE value indicates a more disordered series. We have calculated the SEs of different sequences obtained by iterating various chaotic maps with the initial value $x_0=0.1$. The results are shown in Figure \ref{fig:SE_COMP}. Coinciding with Lyapunov exponents, the type-III chaotic maps get largest SEs among tested chaotic maps including Zhou's method, Hua's CTBCS, type-I and type-II coupling maps, and underlying seed maps. This results demonstrate the good performance of UT-CCS.


\section{CBPRNG for image encryption}
\label{sec:PRNG}
In this section, a new PRNG based on proposed UT-CCS is designed for image encryption. It has been shown in Section \ref{sec:test_of_chaos} that the generated chaotic maps have excellent performance, thus the chaotic sequences gained by iterating chaotic maps already have many properties that an eligible pseudo-random number sequence (PRNS) should have, such as high sensitivity to initial value, aperiodicity, non-convergence, and unpredictability. The only problem left is that the numerical distribution of chaotic sequences are sometimes uneven. As can be seen in Figure \ref{fig:BD_LSCM}-\ref{fig:BD_STCM}, the generated chaotic sequences from type-III coupling maps gather around 0 for some $r$. This non-uniform distribution reduces the randomness greatly, which causes insecurity in encryption. Hence, it is inappropriate to set the coupling chaotic maps as PRNGs directly. To address this problem, we propose the following CBPRNG.
\begin{figure}[htbp]
	\centering
	\includegraphics[scale=0.8]{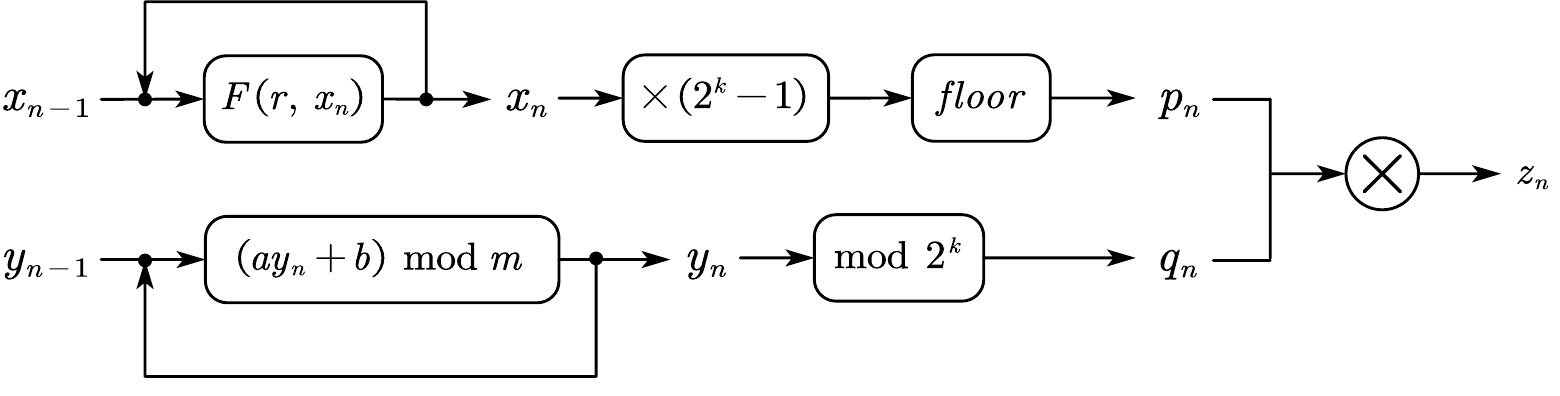}
	\caption{The structure of $k$-bit CBPRNG}
	\label{fig:PRNG_structure}
\end{figure}

\subsection{Structure of proposed CBPRNG}
Figure \ref{fig:PRNG_structure} illustrates the structure of $k$-bit CBPRNG, where $k$ is usually set as 32 or 64 in modern computer system. $x_{n+1}=F(r,x_n)$ is one of the three type-III coupling chaotic maps. $y_{n+1}=(ay_n+b) \ mod \ m$ is the classical linear congruential generator (LCG), where we set $a=1103515245$, $b=12345$, $m=2^{31}$ to get good performance \cite{leebInversiveLinearCongruential1997}. Elements in two sequence $x_n$, $y_n$ are converted into $k$-bit integers $p_n$, $q_n$. Note that $x_n \in [0,1]$ while $y_n \in \mathbb{N}$. The operation \emph{floor} denotes the rounding down function which transforms a floating point number $r$ to the maximum integer no greater than $r$. At last, the corresponding elements in two sequences $p_n$, $q_n$ perform a bitwise XOR operation (denoted by $\otimes$) to get the final PRNS $z_n$. The initial value of chaotic map is $x_0$ while that of LCG is $floor(x_0\cdot m)$, where $x_0 \in (0,1)$.

\subsection{Discussion on the uniformity of CBPRNG}
\label{sec:uniform_proof}
Now we attempt to analyze the numerical distribution of generated PRNS $z_n$ theoretically. To achieve this, we give some simple theorems at first.
\begin{theorem}
	\label{th:2}
	Suppose $X$ is a discrete random variable ranging on $\left\{0,1,2,\cdots\right.$, $\left.2^k-1\right\}$, $X_{k-1}X_{k-2}\cdots X_1X_0$ is the binary representation of $X$. Then $X$ has a uniform distribution if and only if each bit $X_i(i=0,1,\cdots,k-1)$ in the binary form has a 0-1 uniform distribution, i.e.
	\begin{equation}
		\begin{aligned}
		&P\left\{X=0\right\}=P\left\{X=1\right\}=\cdots =P\left\{X=2^k-1\right\}=\frac{1}{2^k} \\
		\Longleftrightarrow &\forall \  i \in \left\{0,1,\cdots,k-1\right\},P\left\{X_i=0\right\}=P\left\{X_i=1\right\}=\frac{1}{2} \\
		\end{aligned}
	\end{equation}
\end{theorem}
\begin{proof}
	1) Sufficiency: $\forall \  x \in \left\{0,1,2,\cdots,2^k-1\right\}$, let $x_{k-1}x_{k-2}\cdots x_1x_0$ represent the binary form of $x$. Then
	\begin{equation*}
		\begin{aligned}
		P\left\{X=x\right\}&=P\left\{X_{k-1}=x_{k-1},X_{k-1}=x_{k-1},\cdots,X_0=x_0\right\} \\
		&=P\left\{X_{k-1}=x_{k-1}\right\} \cdot P\left\{X_{k-2}=x_{k-2}\right\} \cdots P\left\{X_0=x_0\right\} \\
		&=\frac{1}{2} \cdot \frac{1}{2} \cdots \frac{1}{2} \\
		&=\frac{1}{2^k} \\
		\end{aligned}
	\end{equation*}
	2) Necessity: $\forall \  i \in \left\{0,1,2,\cdots,k-1\right\}$, we have
	$$P\left\{X_i=0\right\}=\sum P\left\{X=X_{k-1} \cdots X_{i+1}0X_{i-1} \cdots X_0 \right\}$$
	where $\sum$ denotes summing probabilities of all possible values of 
	$X_{k-1} \cdots X_{i+1}$ $X_{i-1} \cdots X_0$. Apparently there are $2^{k-1}$ combinations of values totally. Considering the condition $P\left\{X=0\right\}=P\left\{X=1\right\}=\cdots =P\left\{X=2^k-1\right\}=1/2^k$, we know that the probability of every combination is $1/2^k$. Therefore we have
	$$P\left\{X_i=0\right\}=2^{k-1} \cdot \frac{1}{2^k}=\frac{1}{2}$$
	The same procedure can be easily adapted to $X_i=1$, thus we finally get
	$$P\left\{X_i=0\right\}=P\left\{X_i=1\right\}=\frac{1}{2}$$
\end{proof}

\begin{theorem}
	\label{th:3}
	Suppose $\xi$, $\eta$ are two independent 0-1 binary random variables. Random variable $\tau = \xi \otimes \eta$, where $\otimes$ denotes exclusive or operation. Then $\tau$ has a 0-1 uniform distribution if either of $\xi$ and $\eta$ has a 0-1 uniform distribution.
\end{theorem}
\begin{proof}
	We may suppose $\xi$ has a 0-1 uniform distribution, i.e. $P\left\{\xi=0\right\}=P\left\{\xi=1\right\}=1/2$. The distributed of $\eta$ is represented as $P\left\{\eta = 0 \right\}=p$, $P\left\{\eta = 1 \right\}=1-p$. Then we have
	\begin{equation*}
		\begin{aligned}
		P\left\{\tau =0 \right\}&=P\left\{\xi=0,\eta=0 \right\} + P\left\{\xi=1,\eta=1 \right\} \\
		&=P\left\{\xi=0 \right\} \cdot P\left\{\eta=0 \right\} + P\left\{\xi=1 \right\} \cdot P\left\{\eta=1 \right\} \\
		&=\frac{1}{2} \cdot p + \frac{1}{2} \cdot (1-p) \\
		&=\frac{1}{2} \\
		\end{aligned}
	\end{equation*}
	Similarly, we can get $P\left\{\tau =1 \right\} = 1/2$. Therefore, $\tau$ has a 0-1 uniform distribution.
\end{proof}

Based on Theorem \ref{th:2} and Theorem \ref{th:3}, it is easy to prove that the output $z_n$ generated by CBPRNG has a uniform distribution. At first, it is well-known that LCG can produce evenly distributed sequence. That is, $q_n$ has a uniform distribution on $\left\{0,1,2,\cdots,2^k-1\right\}$. Then, from Theorem \ref{th:2}, we know that each bit $q_n^i \ (i=0,1,\cdots,k-1)$ in the binary form of $q_n$ has a 0-1 uniform distribution. Because $z_n=p_n \otimes q_n$, where $\otimes$ denotes the bitwise XOR operation, we get each bit $z_n^i$ of $z_n$ has a 0-1 uniform distribution by Theorem \ref{th:3}. At last, we conclude that $z_n$ has a uniform distribution on $\left\{0,1,2,\cdots,2^k-1\right\}$ form the necessity of Theorem \ref{th:2} and the proof is done.

\subsection{Randomness tests of CBPRNG}
PRNG have been studied for many years, along with which a mature evaluation system is established. Here, we employ two widely used randomness test suites to examine the performance of CBPRNG.

\subsubsection{TestU01}
Implemented by C language, TestU01 offers a collection of utilities for the empirical randomness testing of PRNGs \cite{lecuyerTestU01ACLibrary2007}. In our experiment, we use two predefined test batteries, Rabbit and Alphabit to evaluate the randomness of bit sequences generated by CBPRNG. As a comparison, the bit sequences directly gained by iterating chaotic maps also participate in the test. The Rabbit test battery includes 38 subtests for bit sequences of length $2^{20}$ while the Alphabit test battery includes 17 subtests.

Table \ref{tab:TestU01} gives the TestU01 results for bit sequences generated by different PRNGs. It can be seen that the random sequences directly gained from type-III coupling chaotic maps---LSCM-III,TLCM-III, STCM-III, fail many subtests. However, the bit sequences produced by CBPRNG: LSCM-LCG,TLCM-LCG and STCM-LCG, pass all subtests of Rabbit battery and Alphabit battery, which has proved the strong randomness of proposed CB-\\PRNG.

\begin{table}[htbp]
	\centering
	\caption{TestU01 results}
	\begin{tabular}{c|cc}
		\toprule
		PRNGs & Rabbit & Alphabit \\
		\midrule
		LSCM-III   & 11/38 & 4/17 \\
		LSCM-LCG & 38/38 & 17/17 \\
		TLCM-III   & 12/38 & 6/17 \\
		TLCM-LCG & 38/38 & 17/17 \\
		STCM-III   & 15/38 & 7/17 \\
		STCM-LCG & 38/38 & 17/17 \\
		\bottomrule
	\end{tabular}%
	\label{tab:TestU01}%
\end{table}%

\subsubsection{NIST SP800-22 test}
The NIST SP800-22 test standard \cite{basshamSP80022Rev2010} is designed by National Institute of Standards and Technology (NIST) to measure the randomness of binary sequences, which is especially used in the field of cryptography. This test standard has 15 subtests totally and each of them outputs a \emph{P-value}. A random sequence is thought to pass a subtest if the corresponding \emph{P-value} is greater than 0.01. As recommended in \cite{basshamSP80022Rev2010}, at least 55 sequences should be tested and there are two approaches to interpret the empirical results of multiple test sequences. One is examining the proportion of sequences that pass a subtest. The other is checking the uniform distribution of \emph{P-value} for different test sequences to ensure uniformity.

In this experiment, we test 100 binary sequences for every PRNG. The length of each sequence is 1000000 bits. Similarly, we also test the chaotic sequences directly gained by iterating three type-III coupling maps to compare the results with our proposed CBPRNG. According to \cite{basshamSP80022Rev2010}, the minimum pass rate for each subtest with the exception of the random excursion (variant) test is 96 for 100 test sequences. Besides, a $\chi^2$ test is suggested to evaluate the uniformity of \emph{P-value} for 100 test sequences in an arbitrary subtest, which outputs a $\emph{P-value}_T$ as a result. And the sequences can be considered to be uniformly
distributed if $\emph{P-value}_T \geqslant 0.0001$.

Table \ref{tab:NIST_1} lists the proportion of 100 test sequences that pass a subtest for different PRNGs. The pass rate of proposed CBPRNG in all subtests are high enough to meet minimum requirements. However, the pass rate of type-III coupling chaotic maps are relatively low, and even none of 100 test sequences can pass in some subtests. Therefore, the CBPRNG can generate sequences of better randomness. When applying to image encryption, the CBPRNG will enhance security. Besides, Table \ref{tab:NIST_2} lists the $\chi^2$ test results for the 100 \emph{P-value} of each subtest. The three CBPRNGs have $\emph{P-value}_T$ greater than 0.0001 for all subtests while the results of three coupling maps are not good because there are many $\emph{P-value}_T$ smaller than 0.0001. This also indicates that the sequences generated by CBPRNG are really uniformly distributed, which confirms our analysis in Section \ref{sec:uniform_proof}.

\begin{table}[htbp]
	\centering
	\scriptsize
	\caption{Proportion of NIST SP800-22 test results of various PRNGs \tnote{*}}
	\begin{threeparttable}
		\begin{tabular}{c|cc|cc|cc}
			\toprule
			\multirow{3}[4]{*}{Subtest} & \multicolumn{6}{c}{Proportion($\geqslant 96$) \tnote{*}} \\
			\cmidrule{2-7}          & LSCM-LCG & LSCM-III  & TLCM-LCG & TLCM-III  & STCM-LCG & STCM-III \\
			& \multicolumn{2}{c|}{$(x_0=0.4584, r=0.6541)$} & \multicolumn{2}{c|}{$(x_0=0.4584, r=0.0257)$} & \multicolumn{2}{c}{$(x_0=0.4584, r=0.9335)$} \\
			\midrule
			Frequency & 99    & \textbf{0} & 96    & \textbf{0} & 96    & \textbf{0} \\
			Block Frequency & \multirow{2}[0]{*}{98} & \multirow{2}[0]{*}{\textbf{0}} & \multirow{2}[0]{*}{99} & \multirow{2}[0]{*}{\textbf{1}} & \multirow{2}[0]{*}{99} & \multirow{2}[0]{*}{\textbf{56}} \\
			$(m=128)$ &       &       &       &       &       &  \\
			Cusum-Forward & 99    & \textbf{0} & 96    & \textbf{0} & 95    & \textbf{0} \\
			Cusum-Reverse & 100   & \textbf{0} & 96    & \textbf{0} & 96    & \textbf{0} \\
			Runs  & 98    & \textbf{0} & 100   & \textbf{0} & 100   & \textbf{0} \\
			LongestRun & 97    & 98    & 99    & 97    & 99    & 99 \\
			Rank  & 99    & 99    & 99    & 100   & 99    & 100 \\
			FFT   & 99    & \textbf{95}    & 100   & 98    & 100   & 100 \\
			Non Over. Temp. & \multirow{2}[0]{*}{100} & \multirow{2}[0]{*}{\textbf{55}} & \multirow{2}[0]{*}{98} & \multirow{2}[0]{*}{\textbf{61}} & \multirow{2}[0]{*}{99} & \multirow{2}[0]{*}{\textbf{83}} \\
			$(m=9,B=000000001)$ &       &       &       &       &       &  \\
			Over. Temp.$(m=9)$ & 100   & 97    & 98    & 98    & 98    & 97 \\
			Universal & 98    & \textbf{82} & 99    & 99    & 99    & 100 \\
			Appr. Entropy & \multirow{2}[0]{*}{99} & \multirow{2}[0]{*}{\textbf{0}} & \multirow{2}[0]{*}{98} & \multirow{2}[0]{*}{\textbf{1}} & \multirow{2}[0]{*}{98} & \multirow{2}[0]{*}{\textbf{29}} \\
			$(m=10)$ &       &       &       &       &       &  \\
			Ran. Exc.$(x=+1)$ \tnote{**} & 61/62 & -     & 61/61 & -     & 61/61 & - \\
			Ran. Exc. Var. \tnote{**} & \multirow{2}[0]{*}{62/62} & \multirow{2}[0]{*}{-} & \multirow{2}[0]{*}{60/61} & \multirow{2}[0]{*}{-} & \multirow{2}[0]{*}{61/61} & \multirow{2}[0]{*}{-} \\
			$(x=-1)$ &       &       &       &       &       &  \\
			Serial$(m=16,\nabla \Psi^2_m)$ & 99    & \textbf{0} & 98    & \textbf{88} & 98    & 96 \\
			Linear Complexity & \multirow{2}[1]{*}{99} & \multirow{2}[1]{*}{100} & \multirow{2}[1]{*}{100} & \multirow{2}[1]{*}{98} & \multirow{2}[1]{*}{100} & \multirow{2}[1]{*}{99} \\
			$(M=500)$ &       &       &       &       &       &  \\
			\midrule
			Success Count & 15/15 & 4/15  & 15/15 & 6/15  & 15/15 & 7/15 \\
			\bottomrule
		\end{tabular}%
		\begin{tablenotes}
			\footnotesize
			\item[*] Bold font denotes that the PRNG fail the subtest.
			\item[**] The random excursion (variant) test may be not applicable when there are an insufficient number of cycles in the test sequence, see \cite{basshamSP80022Rev2010}. The $\emph{P-value}_T$ cannot be calculated when the number of tested sequences are less than 55, denoted by "-".
		\end{tablenotes}
	\end{threeparttable}
	
	\label{tab:NIST_1}%
\end{table}%

\begin{table}[htbp]
	\centering
	\scriptsize
	\caption{$\emph{P-value}_T$ of NIST SP800-22 test results of various PRNGs}
	\begin{tabular}{c|cc|cc|cc}
		\toprule
		\multirow{3}[3]{*}{Subtest} & \multicolumn{6}{c}{$\emph{P-value}_T(\geqslant 0.0001)$} \\
		\cmidrule{2-7} & LSCM-LCG & LSCM-III  & TLCM-LCG & TLCM-III  & STCM-LCG & STCM-III \\
		& \multicolumn{2}{c|}{$(x_0=0.4584, r=0.6541)$} & \multicolumn{2}{c|}{$(x_0=0.4584, r=0.0257)$} & \multicolumn{2}{c}{$(x_0=0.4584, r=0.9335)$} \\
		\midrule
		Frequency & 0.935716  & \textbf{0.000000 } & 0.383827  & \textbf{0.000000 } & 0.383827  & \textbf{0.000000 } \\
		Block Frequency & \multirow{2}[0]{*}{0.236810 } & \multirow{2}[0]{*}{\textbf{0.000000 }} & \multirow{2}[0]{*}{0.834308 } & \multirow{2}[0]{*}{\textbf{0.000000 }} & \multirow{2}[0]{*}{0.834308 } & \multirow{2}[0]{*}{\textbf{0.000000 }} \\
		$(m=128)$ &       &       &       &       &       &  \\
		Cusum-Forward & 0.171867  & \textbf{0.000000 } & 0.955835  & \textbf{0.000000 } & 0.955835  & \textbf{0.000000 } \\
		Cusum-Reverse & 0.437274  & \textbf{0.000000 } & 0.494392  & \textbf{0.000000 } & 0.494392  & \textbf{0.000000 } \\
		Runs  & 0.759756  & \textbf{0.000000 } & 0.574903  & \textbf{0.000000 } & 0.574903  & \textbf{0.000000 } \\
		LongestRun & 0.419021  & 0.012650  & 0.719747  & 0.494392  & 0.719747  & 0.534146  \\
		Rank  & 0.006196  & 0.202268  & 0.419021  & 0.935716  & 0.419021  & 0.699313  \\
		FFT   & 0.171867  & 0.010988  & 0.595549  & 0.350485  & 0.595549  & 0.102526  \\
		Non Over. Temp. & \multirow{2}[0]{*}{0.816537 } & \multirow{2}[0]{*}{\textbf{0.000000 }} & \multirow{2}[0]{*}{0.924076 } & \multirow{2}[0]{*}{\textbf{0.000000 }} & \multirow{2}[0]{*}{0.924076 } & \multirow{2}[0]{*}{\textbf{0.000000 }} \\
		$(m=9,B=000000001)$ &       &       &       &       &       &  \\
		Over. Temp.$(m=9)$ & 0.066882  & \textbf{0.000000 } & 0.554420  & 0.030806  & 0.554420  & 0.191687  \\
		Universal & 0.779188  & \textbf{0.000000 } & 0.137282  & 0.249284  & 0.137282  & 0.911413  \\
		Appr. Entropy & \multirow{2}[0]{*}{0.574903 } & \multirow{2}[0]{*}{\textbf{0.000000 }} & \multirow{2}[0]{*}{0.514124 } & \multirow{2}[0]{*}{\textbf{0.000000 }} & \multirow{2}[0]{*}{0.514124 } & \multirow{2}[0]{*}{\textbf{0.000000 }} \\
		$(m=10)$ &       &       &       &       &       &  \\
		Ran. Exc.$(x=+1)$ & 0.739918  & -     & 0.086458  & -     & 0.086458  & - \\
		Ran. Exc. Var. & \multirow{2}[0]{*}{0.862344 } & \multirow{2}[0]{*}{-} & \multirow{2}[0]{*}{0.993837 } & \multirow{2}[0]{*}{-} & \multirow{2}[0]{*}{0.988549 } & \multirow{2}[0]{*}{-} \\
		$(x=-1)$ &       &       &       &       &       &  \\
		Serial$(m=16,\nabla \Psi^2_m)$ & 0.023545  & \textbf{0.000000 } & 0.554420  & \textbf{0.000000 } & 0.554420  & \textbf{0.000000 } \\
		Linear Complexity & \multirow{2}[1]{*}{0.090936 } & \multirow{2}[1]{*}{0.262249 } & \multirow{2}[1]{*}{0.657933 } & \multirow{2}[1]{*}{0.350485 } & \multirow{2}[1]{*}{0.657933 } & \multirow{2}[1]{*}{0.534146 } \\
		$(M=500)$ &       &       &       &       &       &  \\
		\midrule
		Success Count & 15/15 & 4/15  & 15/15 & 6/15  & 15/15 & 6/15 \\
		\bottomrule
	\end{tabular}%
	\label{tab:NIST_2}%
\end{table}%

\section{Proposed image encryption algorithm}
\label{sec:IE}
Based on CBPRNG, a novel image encryption algorithm  is proposed in this section.  We adopt the classical "confusion-diffusion" framework \cite{kocarevChaosbasedCryptographyBrief2001} to encrypt digital images. Also, a bit plane flip operation is designed to enhance security. The confusion and diffusion are based on random sequences generated by three types of CBPRNG---LSCM-LCG,TLCM-LCG and STCM-LCG. 
Figure \ref{fig:en_framework} depicts the process of encryption and structure of keys, in which LSCM-LCG and TLCM-LCG are used in the process of confusion, while STCM-LCG is applied in the diffusing stage.
 The following subsections illusrate the encryption process.

\begin{figure}[htbp]
	\centering
	\includegraphics[scale=0.8]{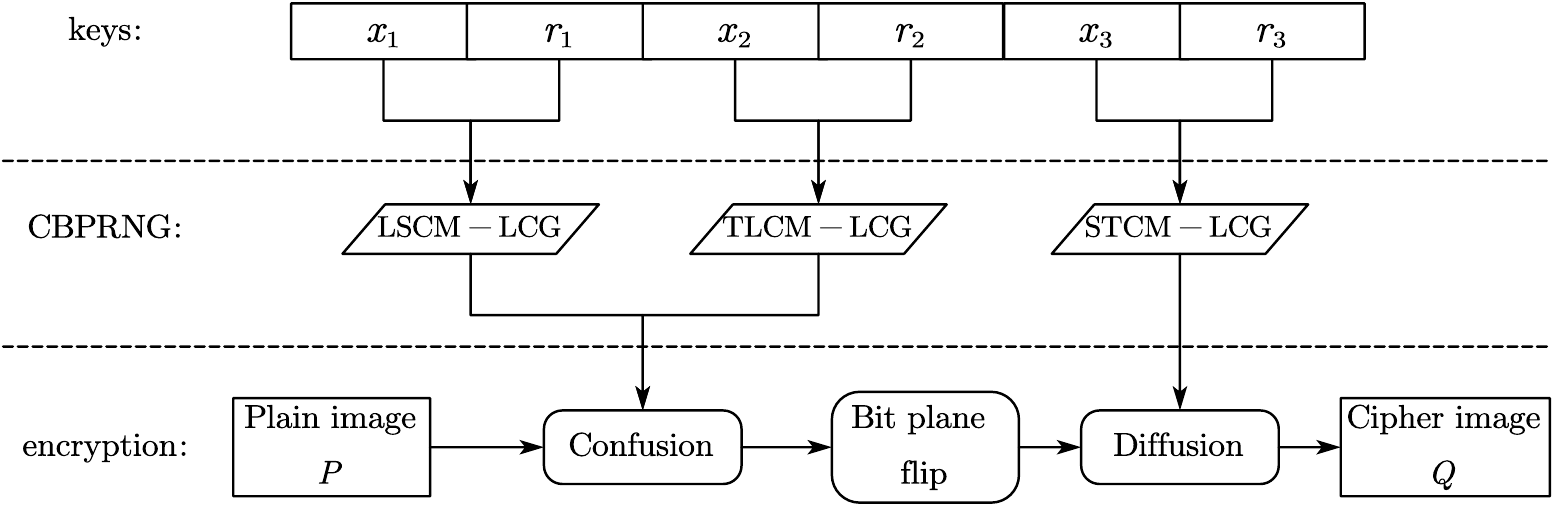}
	\caption{The structure of proposed image encryption scheme.}
	\label{fig:en_framework}
\end{figure}

\subsection{Confusion phase}
The confusion operation scrambles the original image to broke the high correlation among adjacent pixels. An effective approach that many existing algorithms deployed to shuffle the pixels is sorting an unordered random sequence to an ordered one, then pixels are rearranged according to the location maps of elements in two sequences \cite{niuSplicingModelHyper2016,muhammadSecureSurveillanceFramework2018,laiphrakpamCryptanalysisSymmetricKey2017,jainRobustImageEncryption2016}. However, the sorting process is relatively time consuming. In this program, the confusion is completed by exchanging the position of every pixel with the other one sequentially, which is inspired by the famous Knuth shuffle algorithm \cite{durstenfeldAlgorithm235Random1964}. And the pixel to be exchanged is decided by two random matrix generated by CBPRNG.

Let $P$ be a digital image matrix with size $M\times N$. Note that if $P$ is a multichannel color image, just convert it to a gray image by concating different channels. Also, we may assume the color depth of images is 8 bits in the following text. And the other cases have similar encrypting process. For each pixel $P(i,j)$ in $P$, its location is exchanged with pixel $P(H_1(i,j),H_2(i,j))$, where $H_1$, $H_2$ are two $M\times N$ random matrix obtained from LSCM-LCG and TLCM-LCG, respectively. 
The exchanging order  is up to down and left to right, i.e. with the growth of $i$, $j$. The confusion operation is described as Algorithm \ref{algo:confusion}.

\begin{algorithm} 
	\caption{Confusion algorithm} 
	\label{algo:confusion} 
	\begin{algorithmic}[1] 
		\REQUIRE $P$: Original image matrix of size $M\times N$ \\
		\qquad $x_1,r_1,x_2,r_2$: Secrete keys
		\ENSURE $T$: The scrambled image matrix
		
		\STATE $H_1={\rm LSCM-LCG}(x_1,r_1) \  mod \  M$
		\STATE $H_2={\rm TLCM-LCG}(x_2,r_2) \  mod \  N$
		\STATE $T=P$
		\FOR{$i=1$ to $M$}
			\FOR{$j=1$ to $N$}
				\STATE Exchange$(T(i,j),T(H_1(i,j),H_2(i,j)))$
			\ENDFOR	
		\ENDFOR
		
	\end{algorithmic} 
\end{algorithm}

\subsection{Bit plane flip}
The bit plane flip operation (BPF) increases the complexity of encryption and gets attackers into trouble. Suppose $T$ is the scrambled image matrix, we split $T$ into 8 binary planes $T_7T_6\cdots T_0$ by regarding the same bits of each pixel as a bit plane. Then, every bit plane is flipped according to the parity of $i$. When $i$ is even, the plane $T_i$ is flipped up and down; when $i$ is odd, the plane $T_i$ is flipped left and right. After all the 8 planes are flipped, they are merged into a new image matrix $S$. Figure \ref{fig:BPF} exhibits the process of BPF operation.

\begin{figure}[htbp]
	\centering
	\includegraphics[scale=0.8]{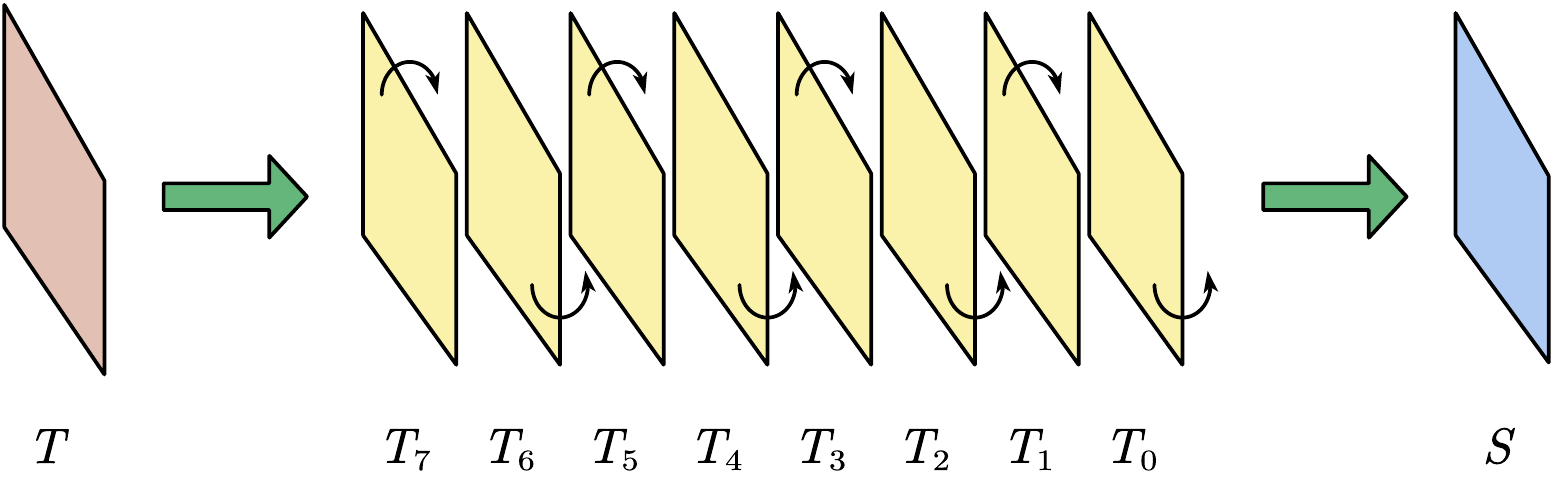}
	\caption{Process of BPF operation}
	\label{fig:BPF}
\end{figure}

\subsection{Diffusion phase}
In image encryption, it makes sense that a plain image with little change results a totally different ciphertext image. That is, the encrypted image should be extremely sensitive to original image, which is crucial in resisting known-plaintext attack \cite{zhangImprovedKnownplaintextAttack2018} and chosen-plaintext attack \cite{liuChosenplaintextAttackImage2016}. The diffusion operation is exactly designed to make small changes in plain image spread through the whole image matrix.

Our diffusion algorithm contains two main steps: forward diffusion (FD) and reverse diffusion (RD). The reason why we perform two diffusion in opposite directions is to increase the stability. That is, wherever the change is happened in original image, it will affect all other pixels. Besides, a random sequence participants in the whole process for improving numerical distribution of cipher image.

In FD phase, we make a pixel influenced by the bit-reverse order (e.g. the bit-reverse order of $6=(110)_2$ is $(011)_2=3$) of its previous pixel in cipher. This process is based on bitwise XOR operation and can be formally defined as Eq.(\ref{eq:diffusion_1}):
\begin{equation}
	\label{eq:diffusion_1}
	\hat{G}_i=\hat{S}_i \otimes \hat{U}_i \otimes {\rm BitRev}(\hat{G}_{i-1}),
\end{equation}
where $\hat{G}$ is the diffused pixel sequence, $\hat{S}$ is the bit-plane-flipped image pixel sequence, $\hat{U}$ is a random sequence generated by STCM-LCG, and BitRev$(x)$ denotes the bit-reverse order of $x$. 

In RD phase, we employ the modulo operation to make each pixel affected by its succeeding pixel in cipher. The random sequence is still joined by bitwise XOR operation. The RD procedure is defined by Eq.(\ref{eq:diffusion_2}):
\begin{equation}
	\label{eq:diffusion_2}
	\hat{Q}_i=((\hat{G}_i+\hat{Q}_{i+1}) \  mod \  256) \otimes \hat{U}_{i+MN},
\end{equation}
where $\hat{Q}$ is the pixel sequence of the final cipher image and we set $\hat{Q}_{MN+1}=\hat{G}_1$ initially. Algorithm \ref{algo:diffusion} describes the diffusion operation.

\begin{algorithm} 
	\caption{Diffusion algorithm} 
	\label{algo:diffusion} 
	\begin{algorithmic}[1] 
		\REQUIRE $S$: The bit-plane-flipped image matrix of size $M\times N$ \\
		\qquad $x_3,r_3$: Secrete keys
		\ENSURE $Q$: The finally Encrypted image matrix
		
		\STATE $\hat{S} = {\rm reshape}(S,MN,1)$
		\STATE $U={\rm STCM-LCG}(x_3,r_3,2MN,1) \  mod \  256$
		\STATE $\hat{G}_1=\hat{S}_1 \otimes \hat{U}_1$
		\FOR{$i=2$ to $MN$}
		\STATE $\hat{G}_i=\hat{S}_i \otimes \hat{U}_i \otimes {\rm BitRev}(\hat{G}_{i-1})$
		\ENDFOR
		
		\STATE $\hat{Q}_{MN}=((\hat{G}_{MN}+\hat{G}_1) \  mod \  256) \otimes \hat{U}_{2MN}$
		\FOR{$i=MN-1$ to $1$}
		\STATE $\hat{Q}_i=((\hat{G}_i+\hat{Q}_{i+1}) \  mod \  256) \otimes \hat{U}_{i+MN}$
		\ENDFOR
		
		\STATE $ Q = {\rm reshape}(\hat{Q},M,N)$
	\end{algorithmic} 
\end{algorithm}

\section{Security analysis of proposed image encryption algorithm}
\label{sec:test_IE}
To test the performance of our proposed encryption scheme, we implement the encryption-decryption process in MATLAB R2018b, where the hardware parameters of the computer are Intel Core i7-4710MQ CPU and 16GB RAM. The plaintext images that we choose for test are six images of size $256 \times 256$. The four grayscale images are Lena, Cameraman, Peppers, and Starfish, while the two color images are Airplane and Baboon. We evaluate the security of our proposed algorithm from multiple aspects and compare the results with existing methods, which demonstrates the effectiveness and the progress of the proposed image encryption scheme.

\subsection{Histogram analysis}
Histogram of a digital image shows the distribution of gray value of all pixels. A meaningful image will show obvious characteristic in its histogram that the number of pixels with the same gray value takes a specific percentage. Thus, an ideal cipher image should have an uniform distributed histogram that will never leak any meaningful information to attackers who use statistical attacks. 

The histograms of four grayscales and their corresponding ciphers are shown as Figure \ref{fig:hist_gray}. The histograms of plaintext images present remarkable features as seen in Figure \ref{fig:hist_gray_a}-\ref{fig:hist_gray_d}. After encryption, all the images become indistinguishable noise-like ciphers. We can never tell the difference among Figure \ref{fig:hist_gray_i}-\ref{fig:hist_gray_l} visually. And the histograms of four ciphertext images become identical. All 256 gray levels appear with the same probability in encrypted images. This will invalidate the statistical attacks based on pixel value frequency.
\begin{figure}[htbp]
	\centering
	\subfigure[]{\label{fig:hist_gray_a}
		\includegraphics[width=0.22\textwidth]{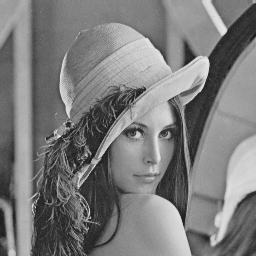}}
	\subfigure[]{\label{fig:hist_gray_b}
		\includegraphics[width=0.22\textwidth]{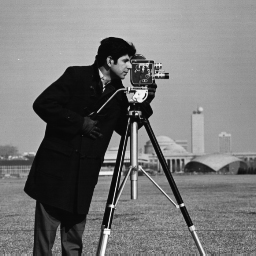}}
	\subfigure[]{\label{fig:hist_gray_c}
		\includegraphics[width=0.22\textwidth]{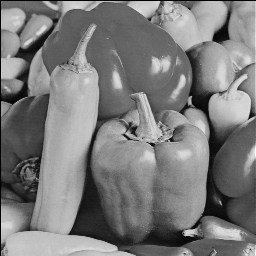}}
	\subfigure[]{\label{fig:hist_gray_d}
		\includegraphics[width=0.22\textwidth]{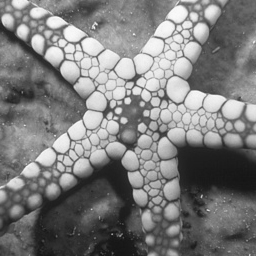}}
	\subfigure[]{\label{fig:hist_gray_e}
		\includegraphics[width=0.22\textwidth]{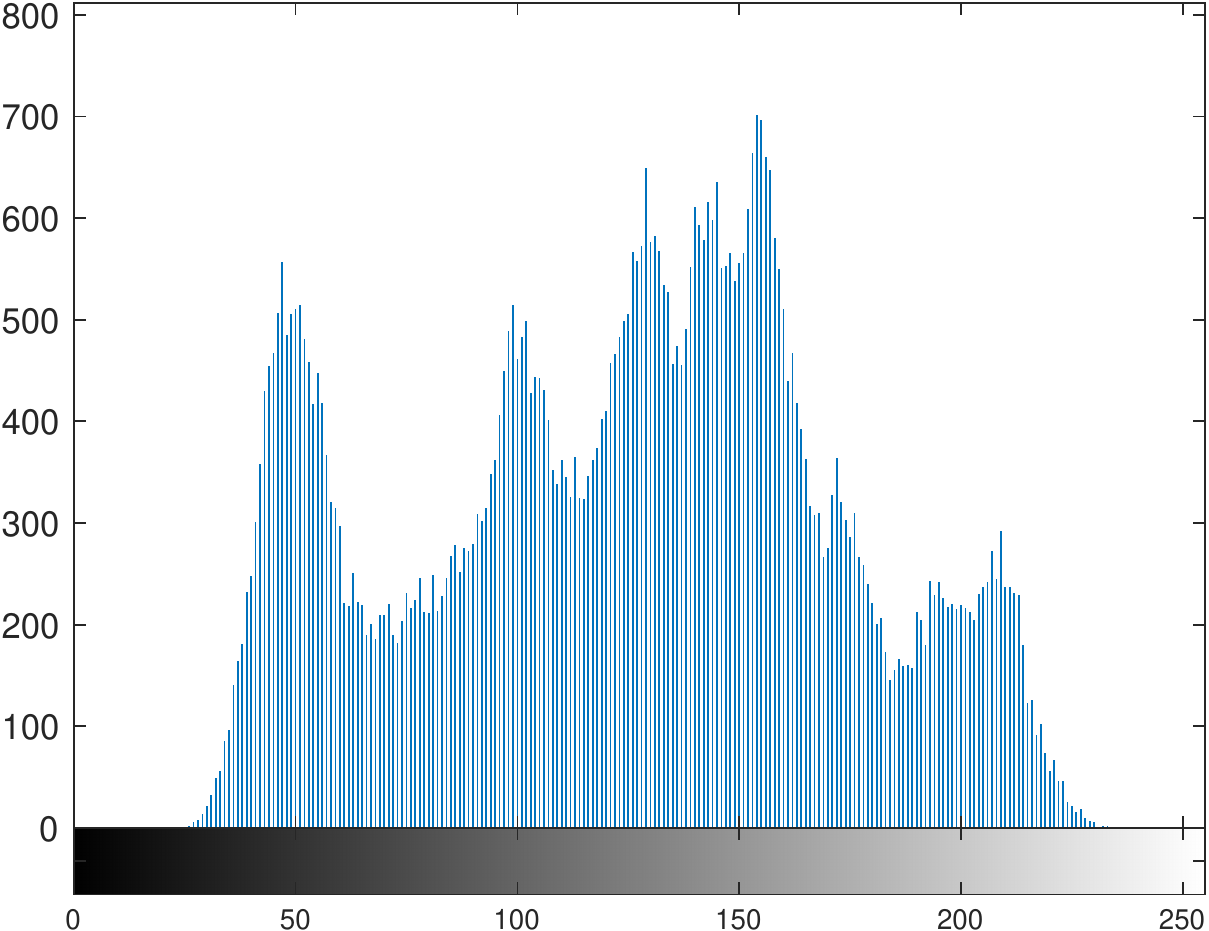}}
	\subfigure[]{\label{fig:hist_gray_f}
		\includegraphics[width=0.22\textwidth]{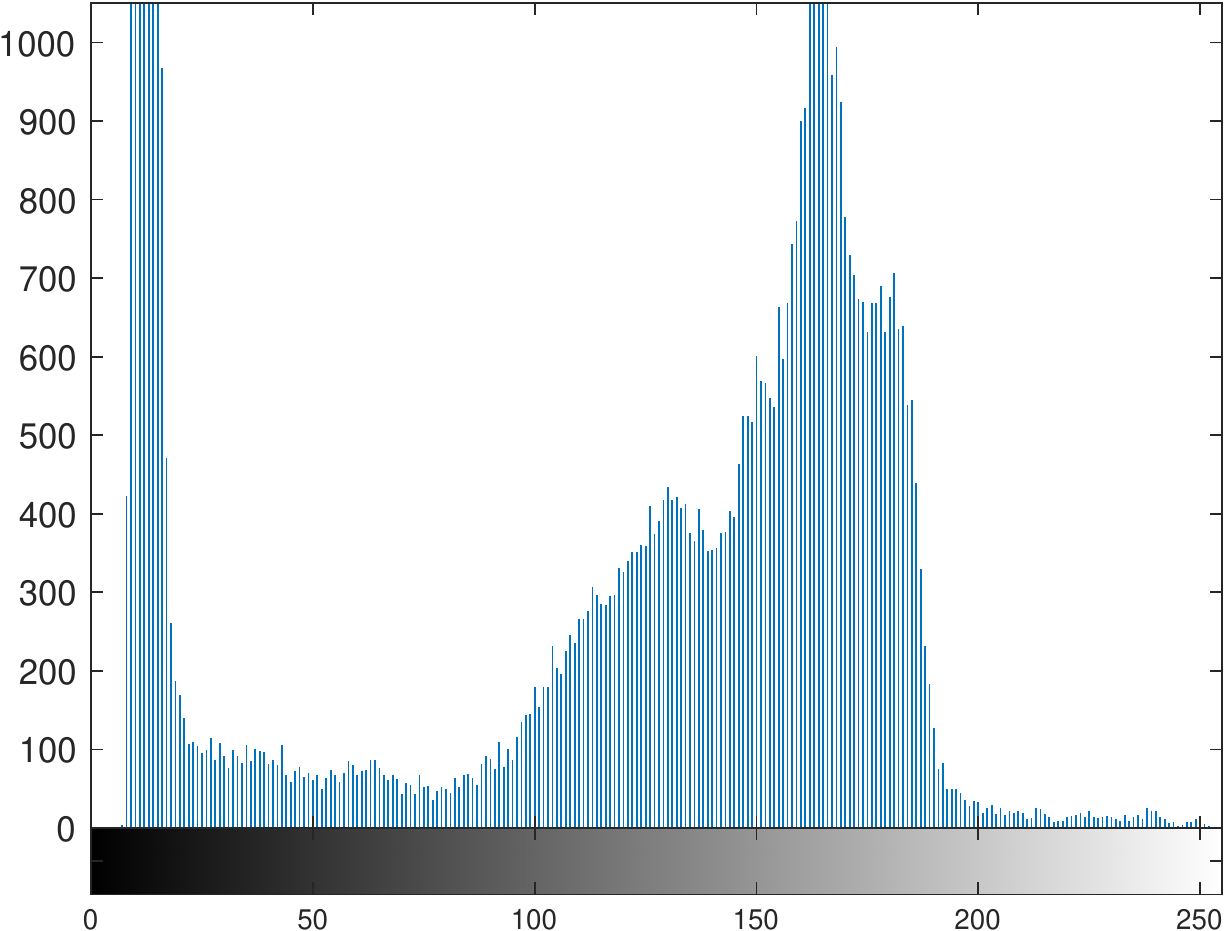}}
	\subfigure[]{\label{fig:hist_gray_g}
		\includegraphics[width=0.22\textwidth]{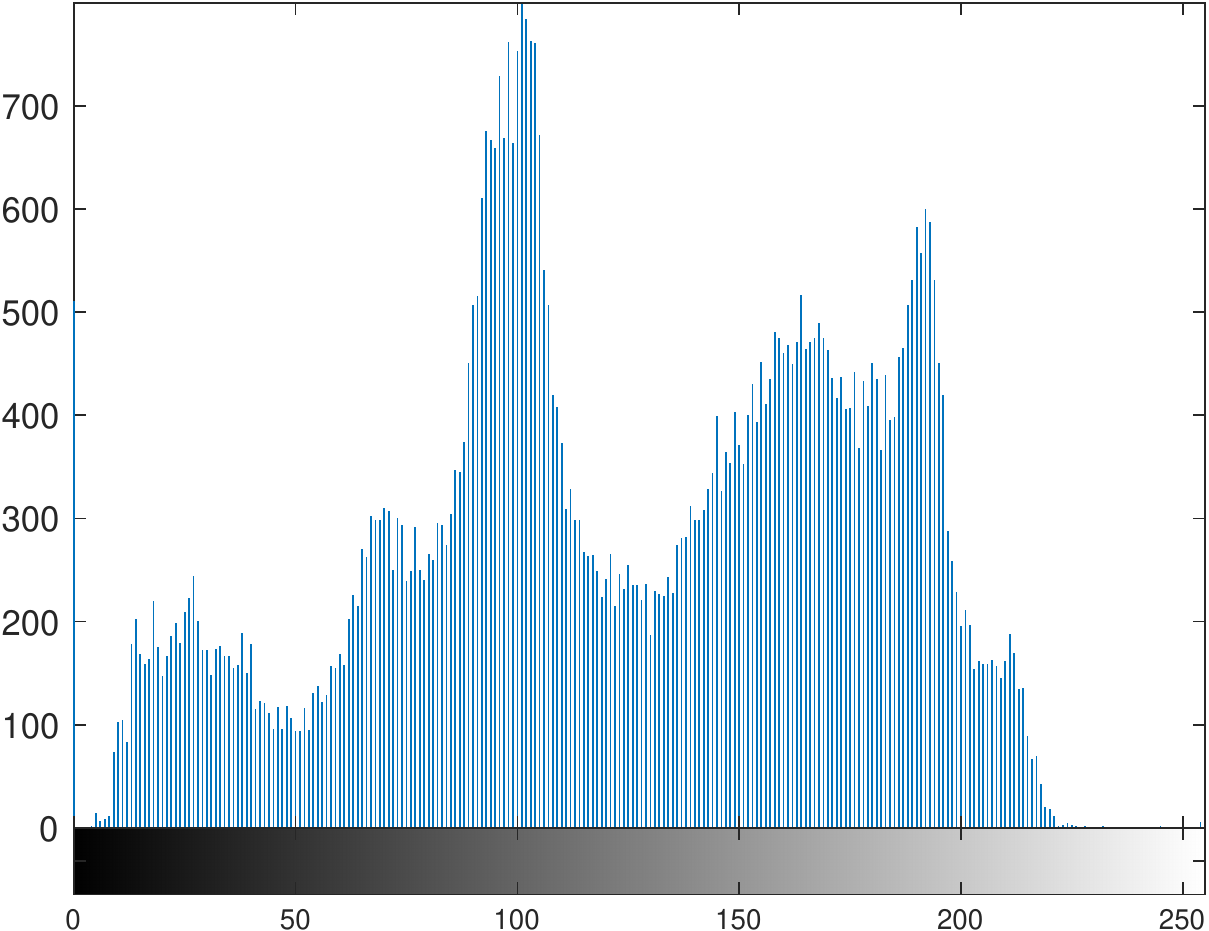}}
	\subfigure[]{\label{fig:hist_gray_h}
		\includegraphics[width=0.22\textwidth]{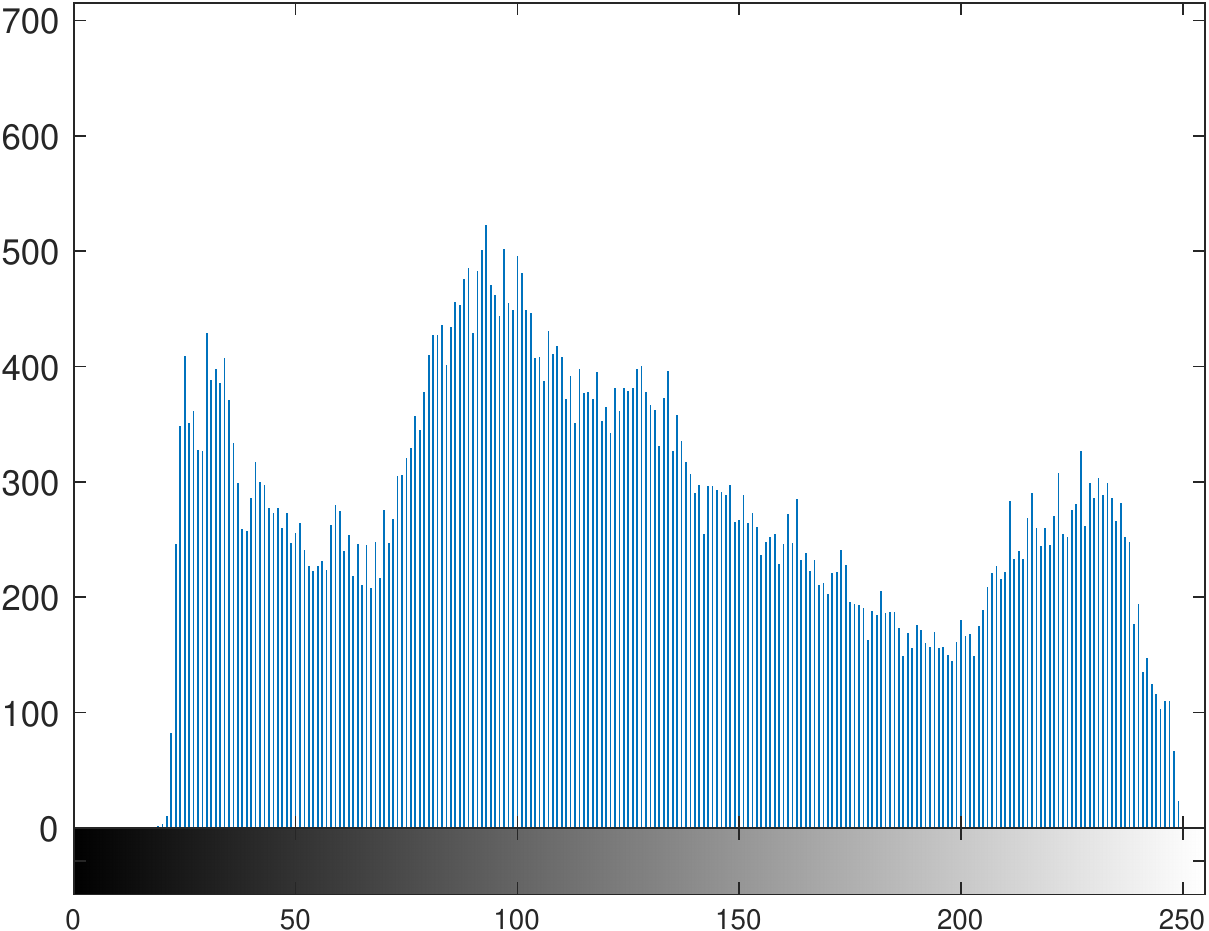}}
	\subfigure[]{\label{fig:hist_gray_i}
		\includegraphics[width=0.22\textwidth]{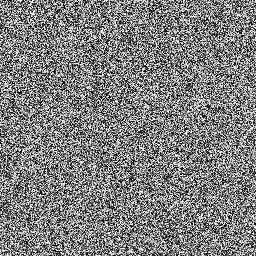}}
	\subfigure[]{\label{fig:hist_gray_j}
		\includegraphics[width=0.22\textwidth]{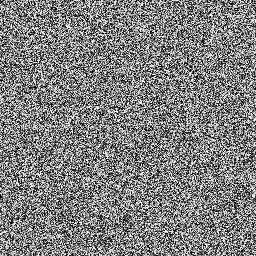}}
	\subfigure[]{\label{fig:hist_gray_k}
		\includegraphics[width=0.22\textwidth]{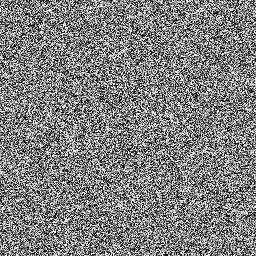}}
	\subfigure[]{\label{fig:hist_gray_l}
		\includegraphics[width=0.22\textwidth]{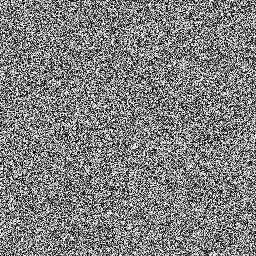}}
	\subfigure[]{\label{fig:hist_gray_m}
		\includegraphics[width=0.22\textwidth]{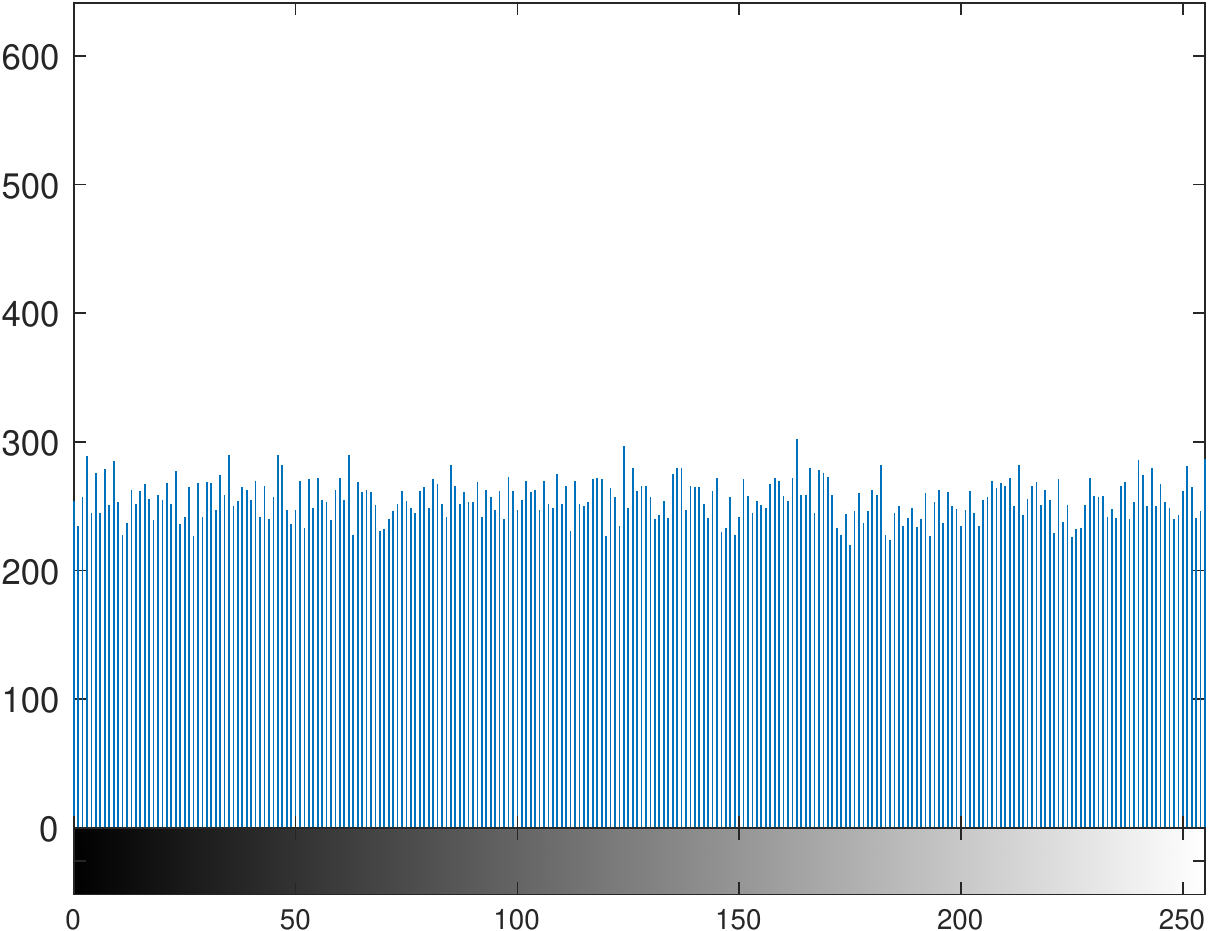}}
	\subfigure[]{\label{fig:hist_gray_n}
		\includegraphics[width=0.22\textwidth]{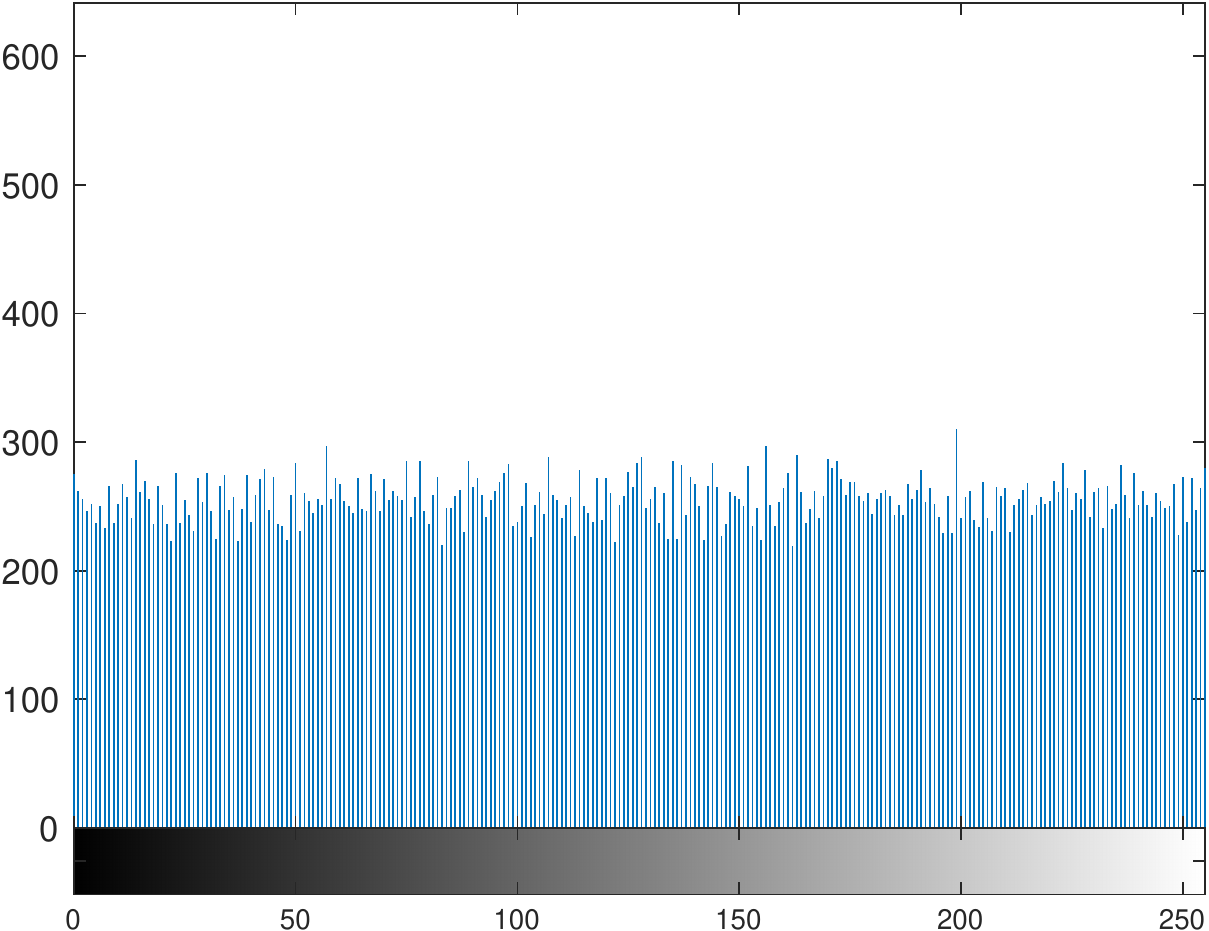}}
	\subfigure[]{\label{fig:hist_gray_o}
		\includegraphics[width=0.22\textwidth]{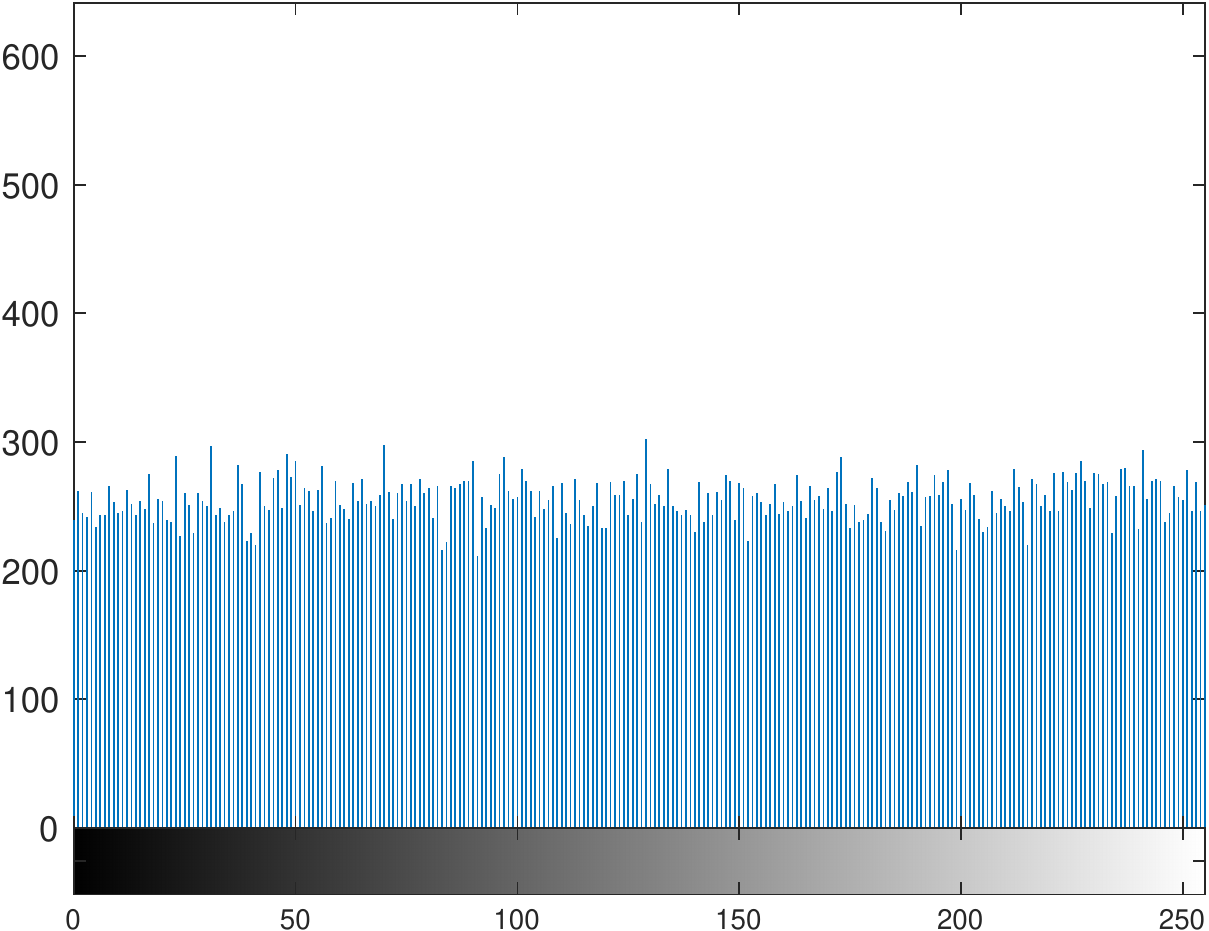}}
	\subfigure[]{\label{fig:hist_gray_p}
		\includegraphics[width=0.22\textwidth]{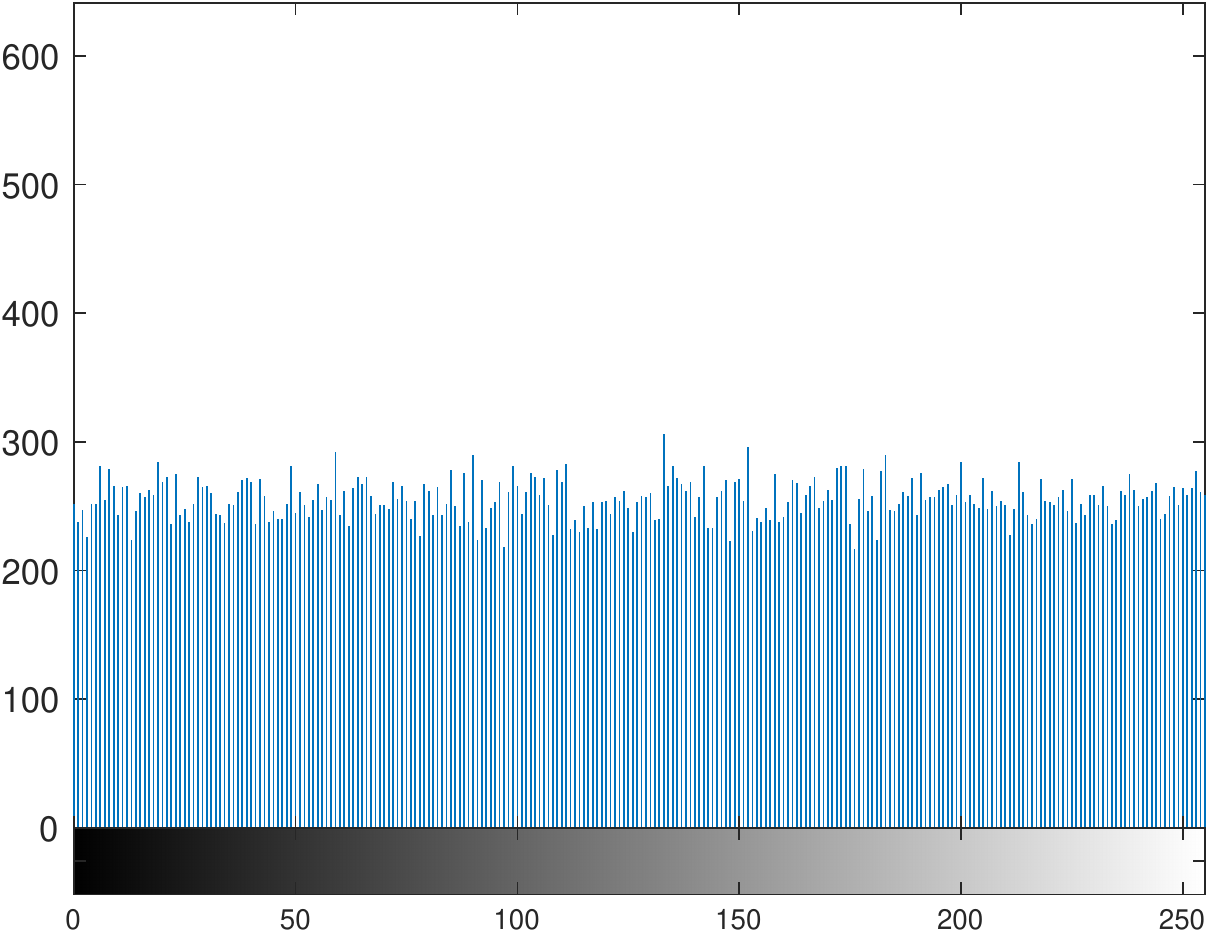}}
	\caption{Histogram analysis: (a) Lena; (b) Cameraman; (c) Peppers; (d) Starfish; (e) Histogram of (a); (f) Histogram of (b); (g) Histogram of (c); (h) Histogram of (d); (i) Encrypted Lena; (j) Encrypted Cameraman; (k) Encrypted Peppers; (l) Encrypted Starfish; (m) Histogram of (i); (n) Histogram of (j); (o) Histogram of (k); (p) Histogram of (l);}
	\label{fig:hist_gray}
\end{figure}

Figure \ref{fig:hist_color} illustrates the histogram of two color images as well as their ciphers. As seen in Figure \ref{fig:hist_color}, all the three channels of color images exhibit characteristics in their histograms. After encryption, the pixels in all three channels are evenly distributed, as shown in Figure \ref{fig:hist_color_f}-\ref{fig:hist_color_h} and  \ref{fig:hist_color_n}-\ref{fig:hist_color_p}. Therefore, The proposed algorithm is also applicable for color images.
\begin{figure}[htbp]
	\centering
	\subfigure[]{
		\includegraphics[width=0.22\textwidth]{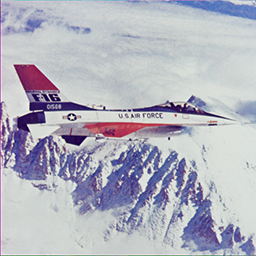}\label{fig:hist_color_a}}
	\subfigure[]{
		\includegraphics[width=0.22\textwidth]{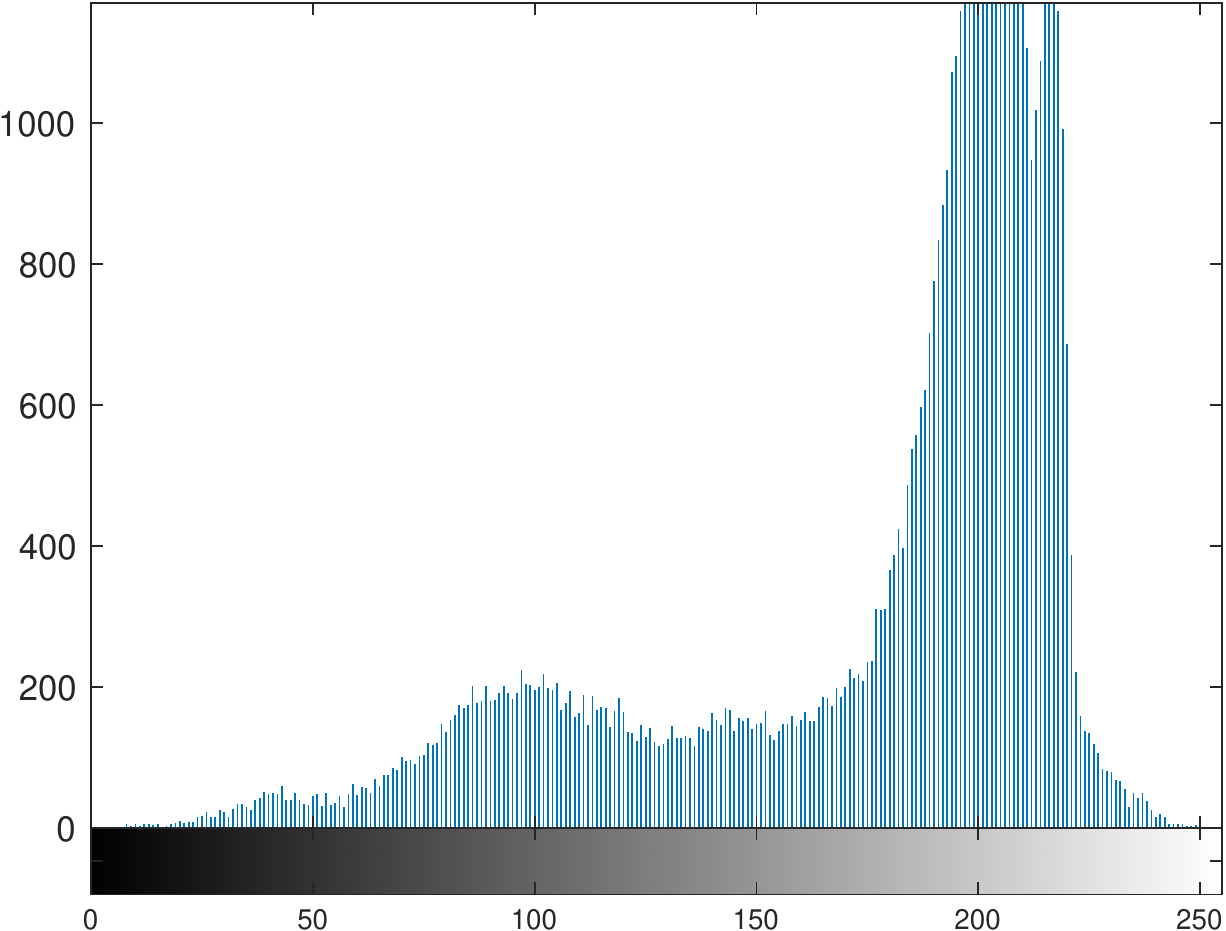}\label{fig:hist_color_b}}
	\subfigure[]{
		\includegraphics[width=0.22\textwidth]{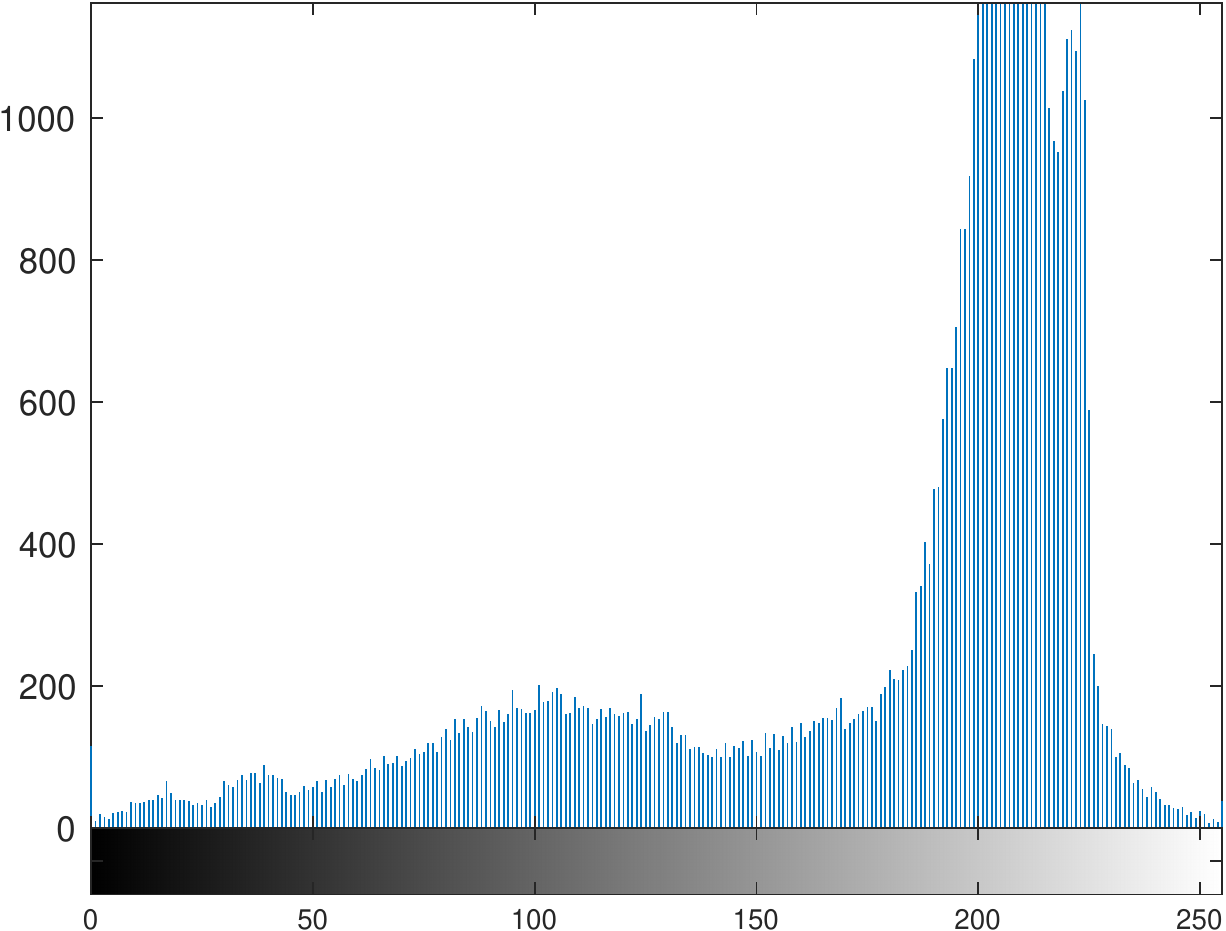}\label{fig:hist_color_c}}
	\subfigure[]{
		\includegraphics[width=0.22\textwidth]{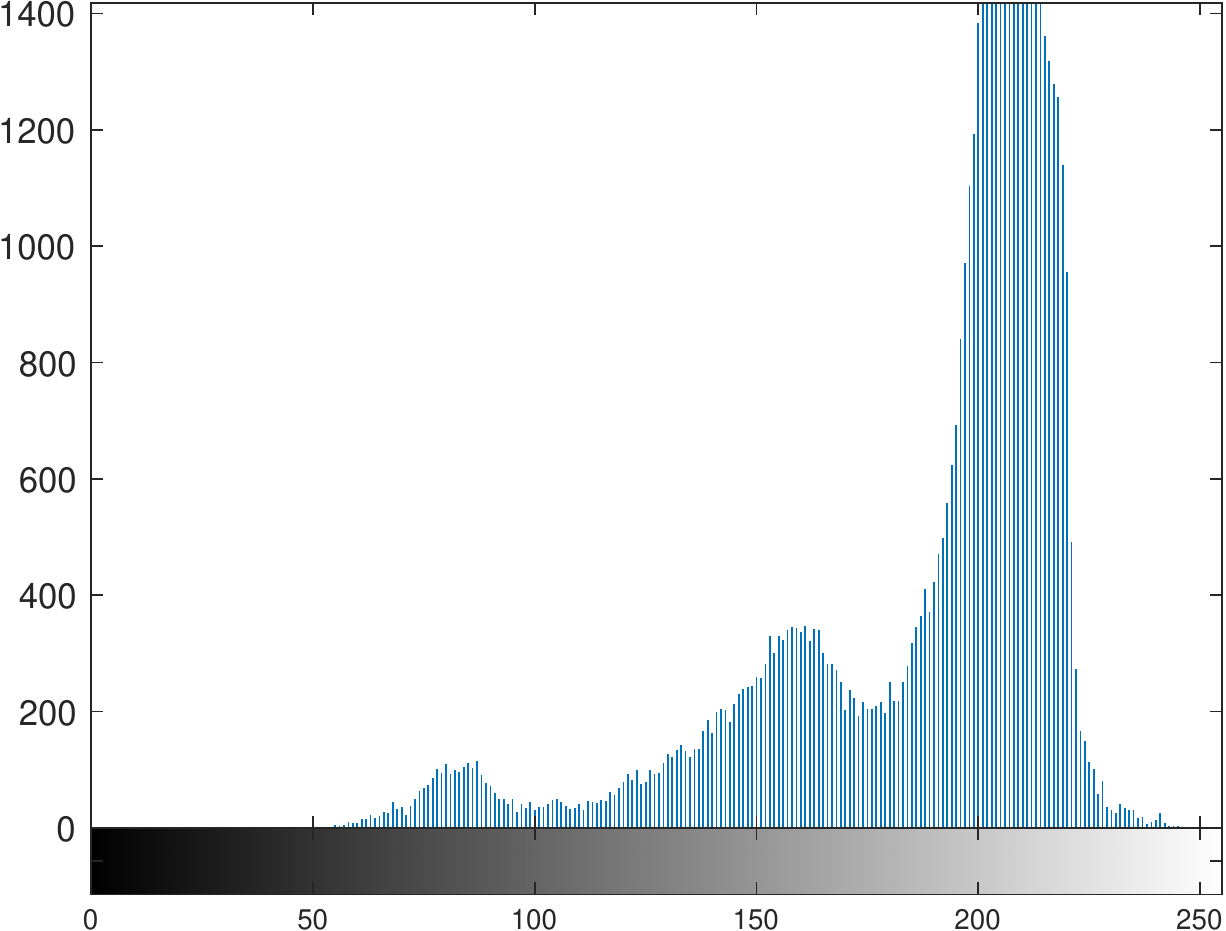}\label{fig:hist_color_d}}
	\subfigure[]{
		\includegraphics[width=0.22\textwidth]{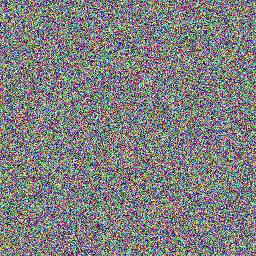}\label{fig:hist_color_e}}
	\subfigure[]{
		\includegraphics[width=0.22\textwidth]{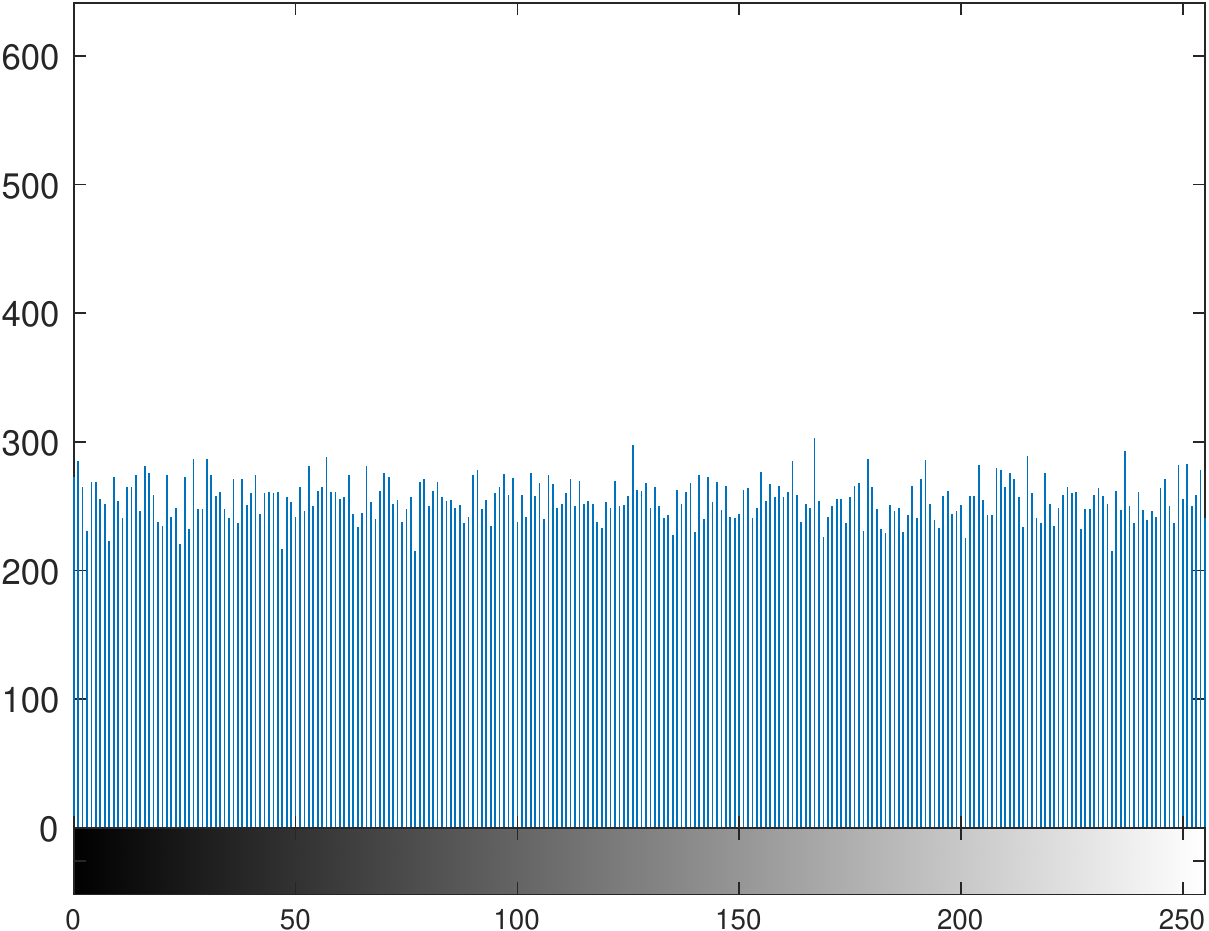}\label{fig:hist_color_f}}
	\subfigure[]{
		\includegraphics[width=0.22\textwidth]{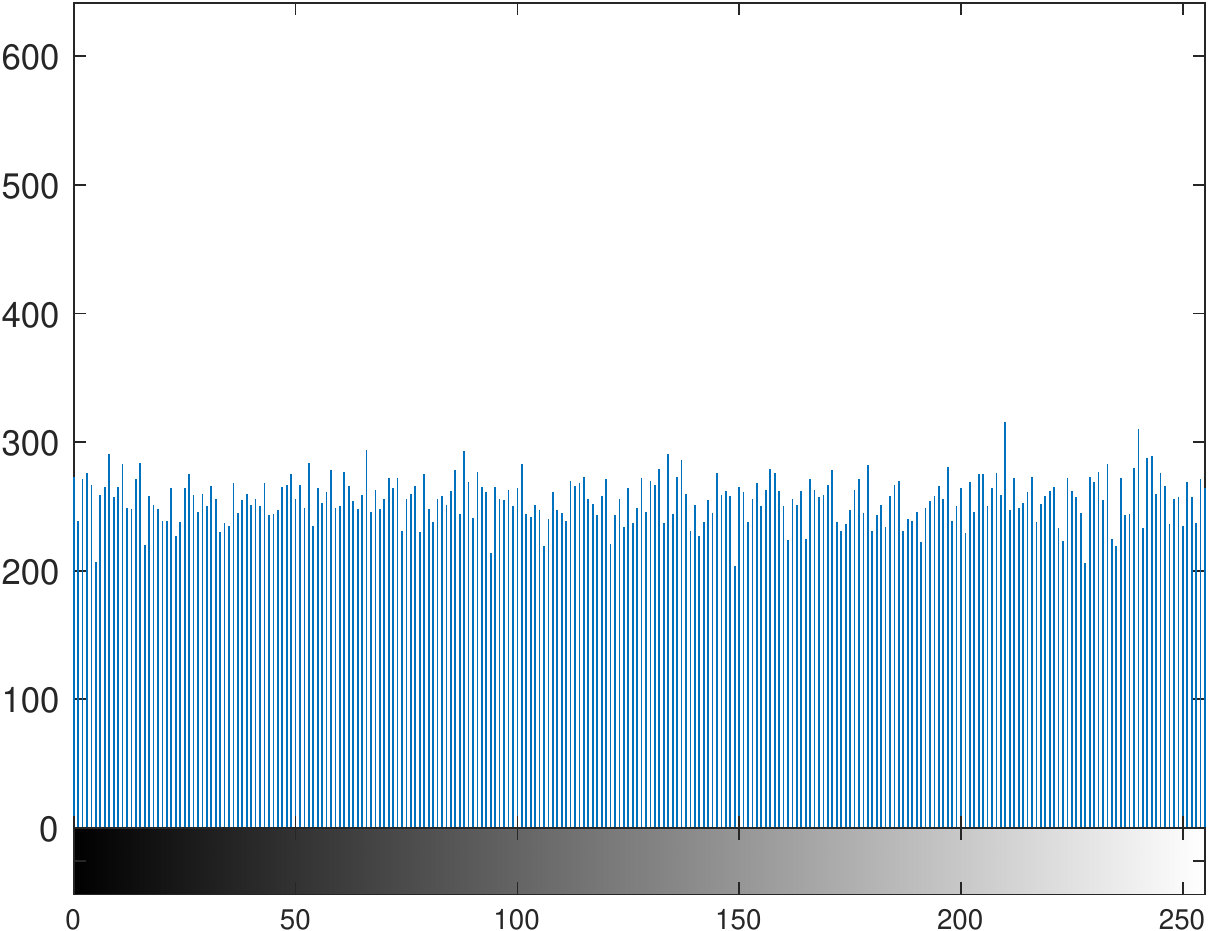}\label{fig:hist_color_g}}
	\subfigure[]{
		\includegraphics[width=0.22\textwidth]{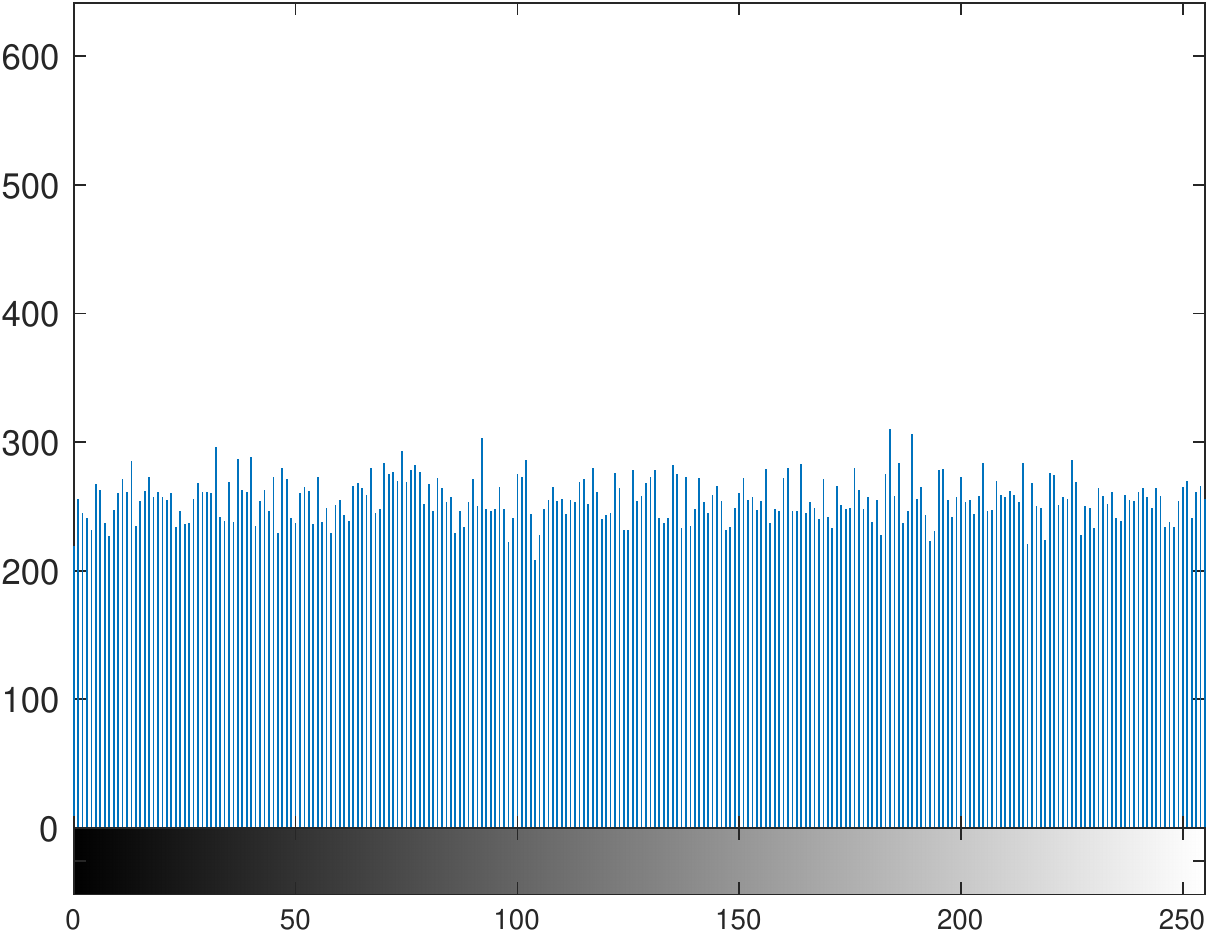}\label{fig:hist_color_h}}
	\subfigure[]{
		\includegraphics[width=0.22\textwidth]{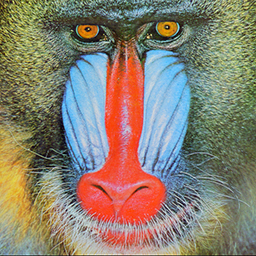}\label{fig:hist_color_i}}
	\subfigure[]{
		\includegraphics[width=0.22\textwidth]{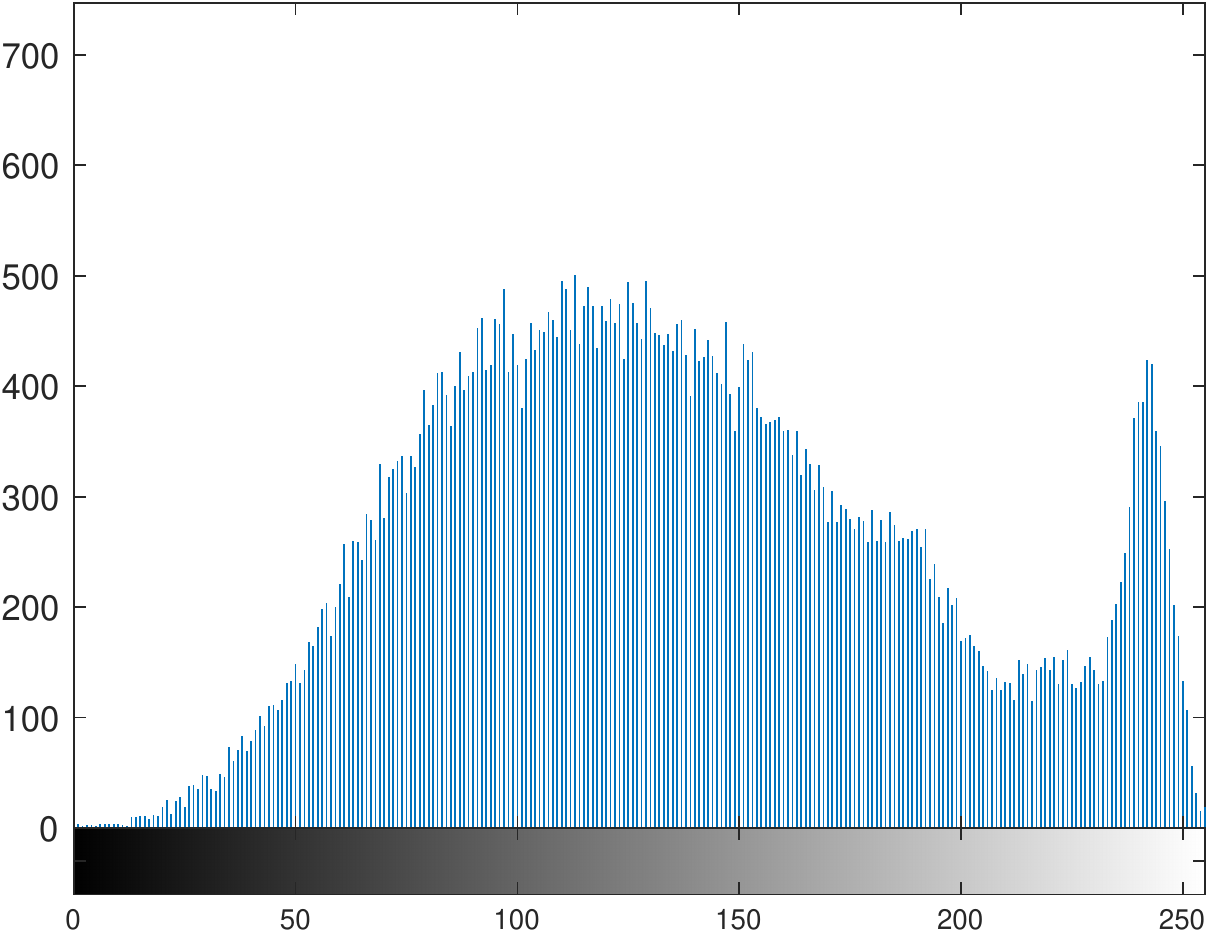}\label{fig:hist_color_j}}
	\subfigure[]{
		\includegraphics[width=0.22\textwidth]{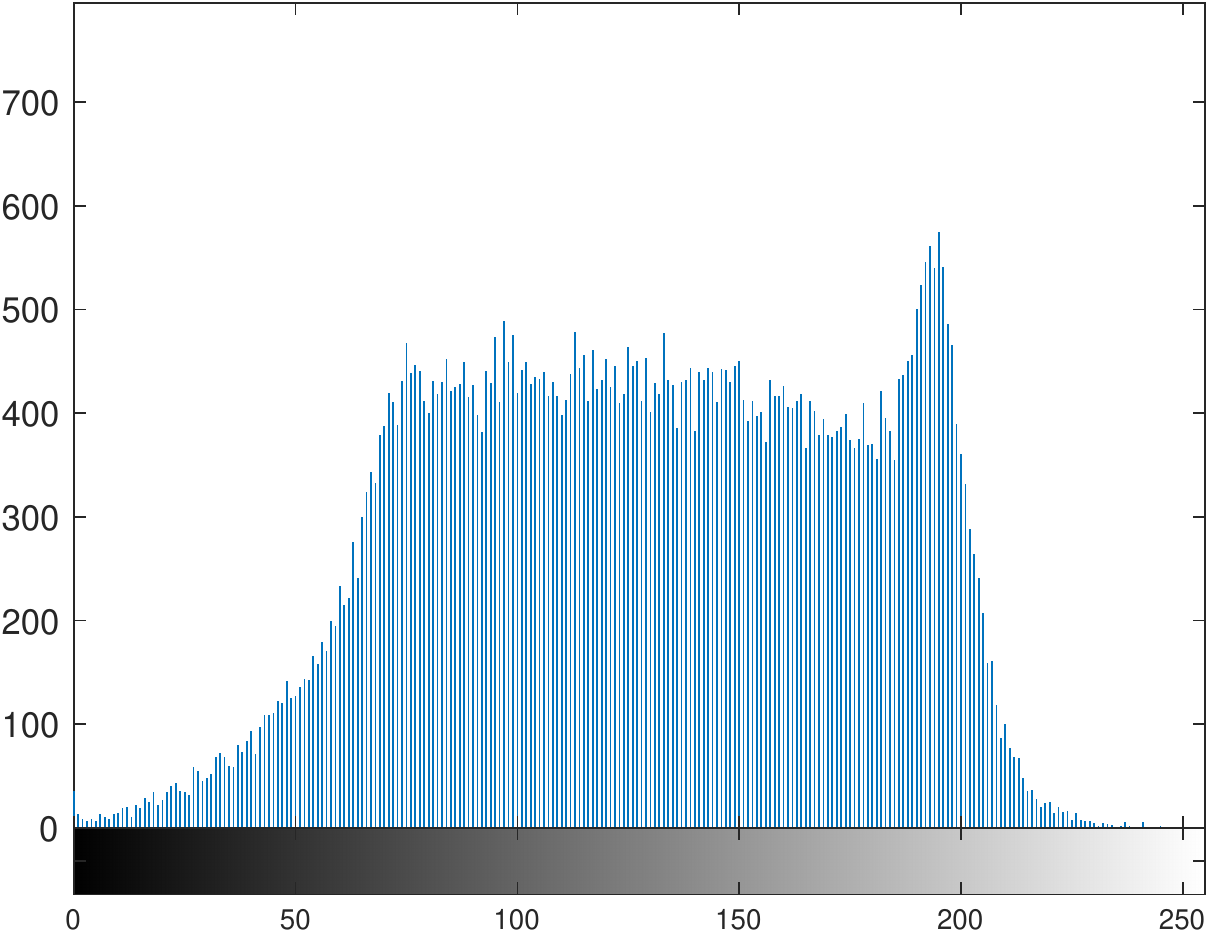}\label{fig:hist_color_k}}
	\subfigure[]{
		\includegraphics[width=0.22\textwidth]{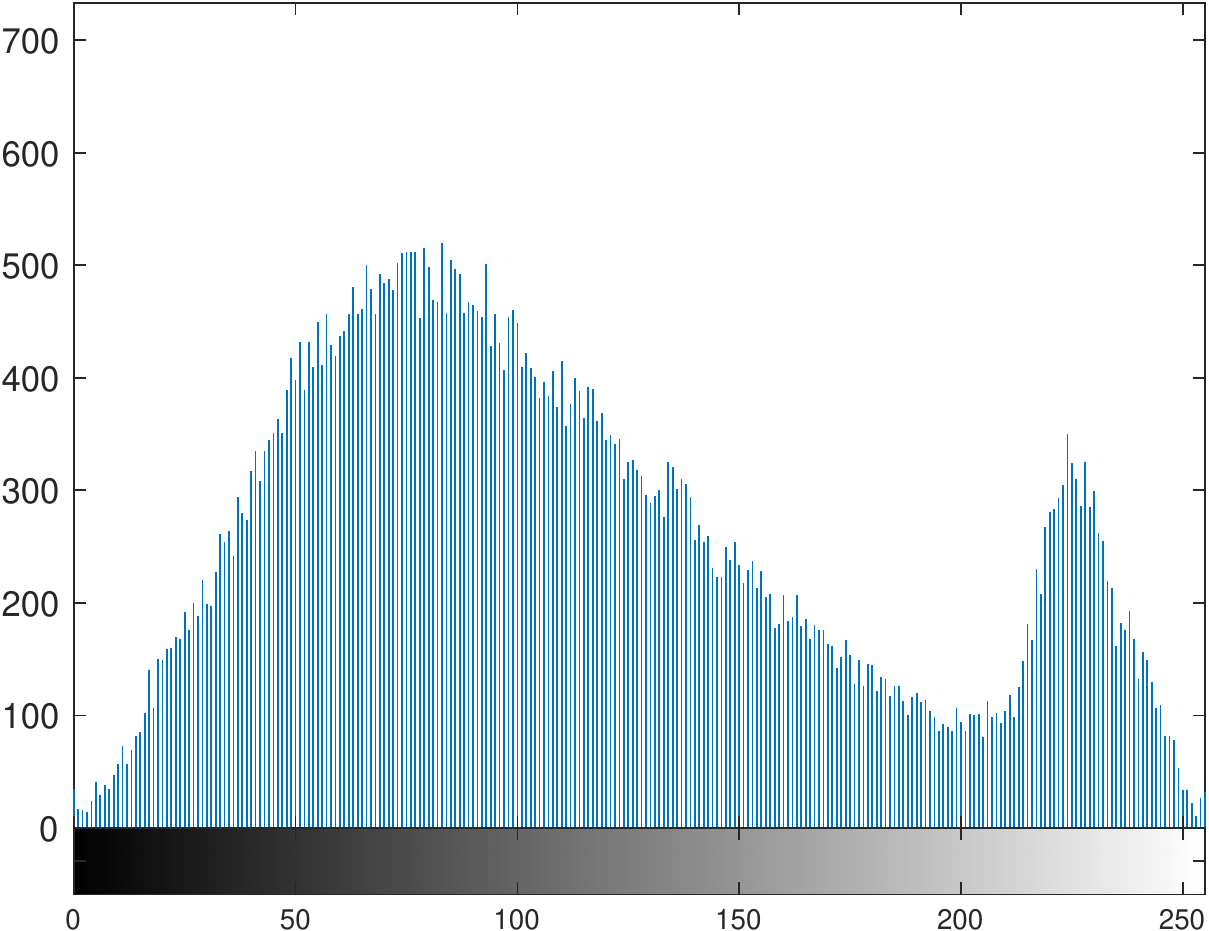}\label{fig:hist_color_l}}
	\subfigure[]{
		\includegraphics[width=0.22\textwidth]{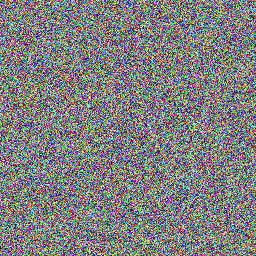}\label{fig:hist_color_m}}
	\subfigure[]{
		\includegraphics[width=0.22\textwidth]{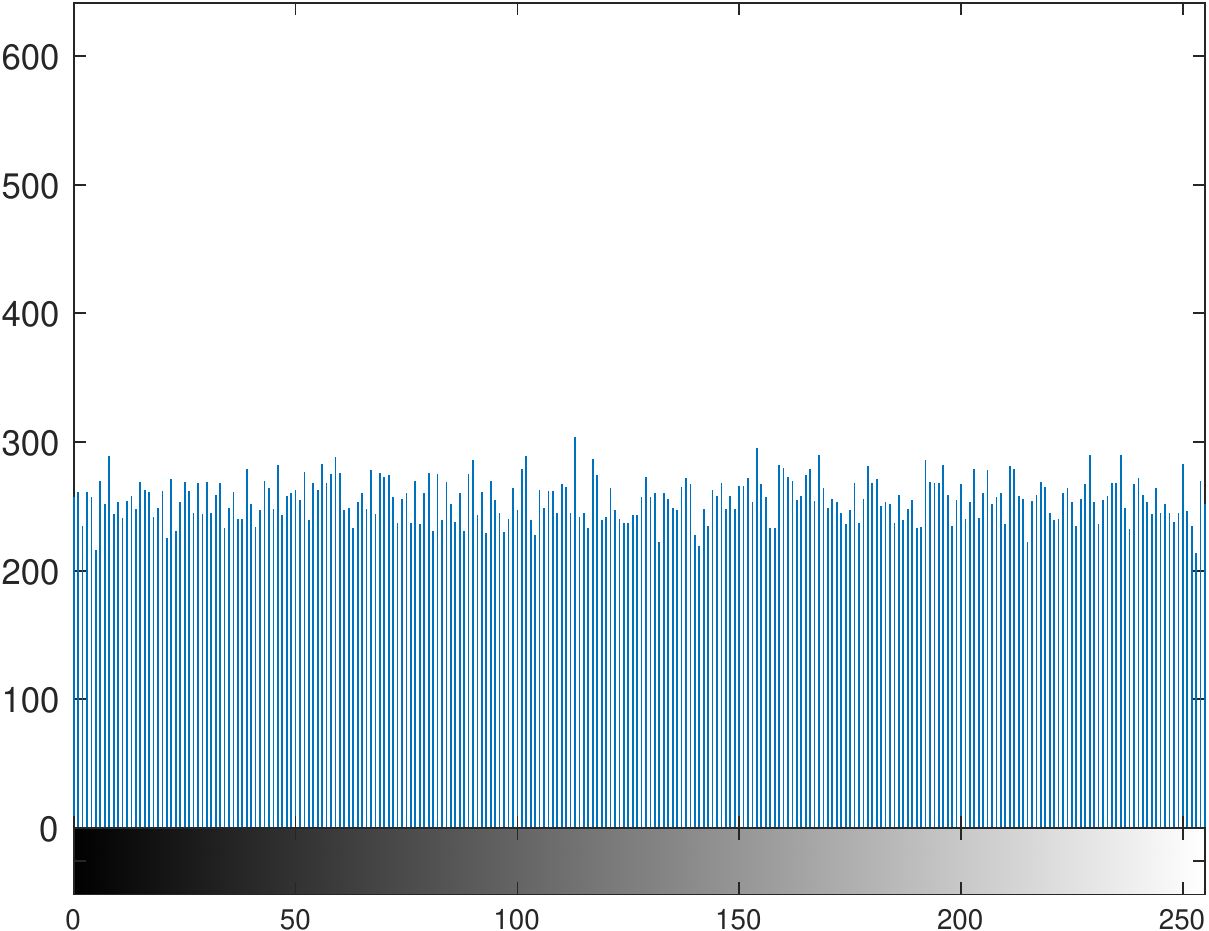}\label{fig:hist_color_n}}
	\subfigure[]{
		\includegraphics[width=0.22\textwidth]{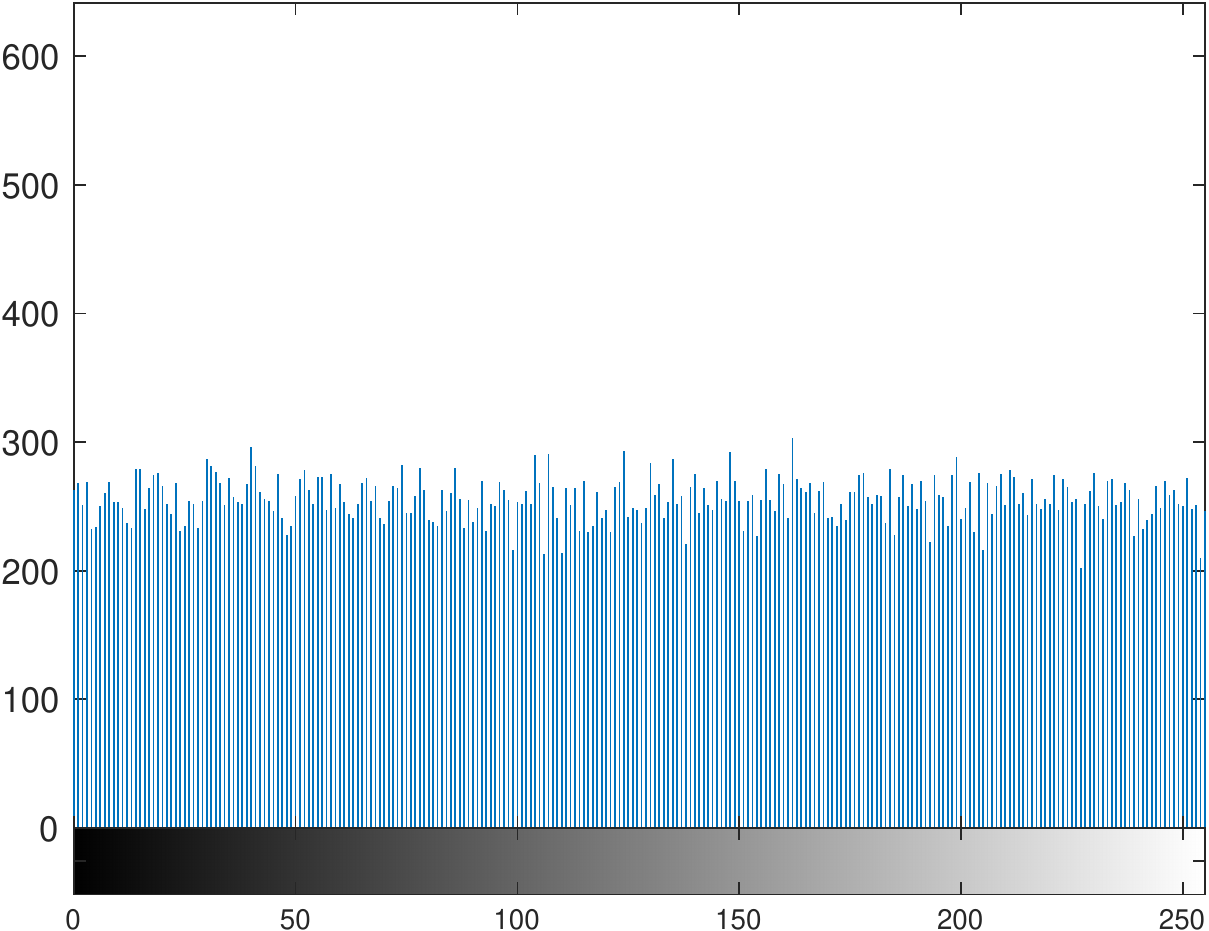}\label{fig:hist_color_0}}
	\subfigure[]{
		\includegraphics[width=0.22\textwidth]{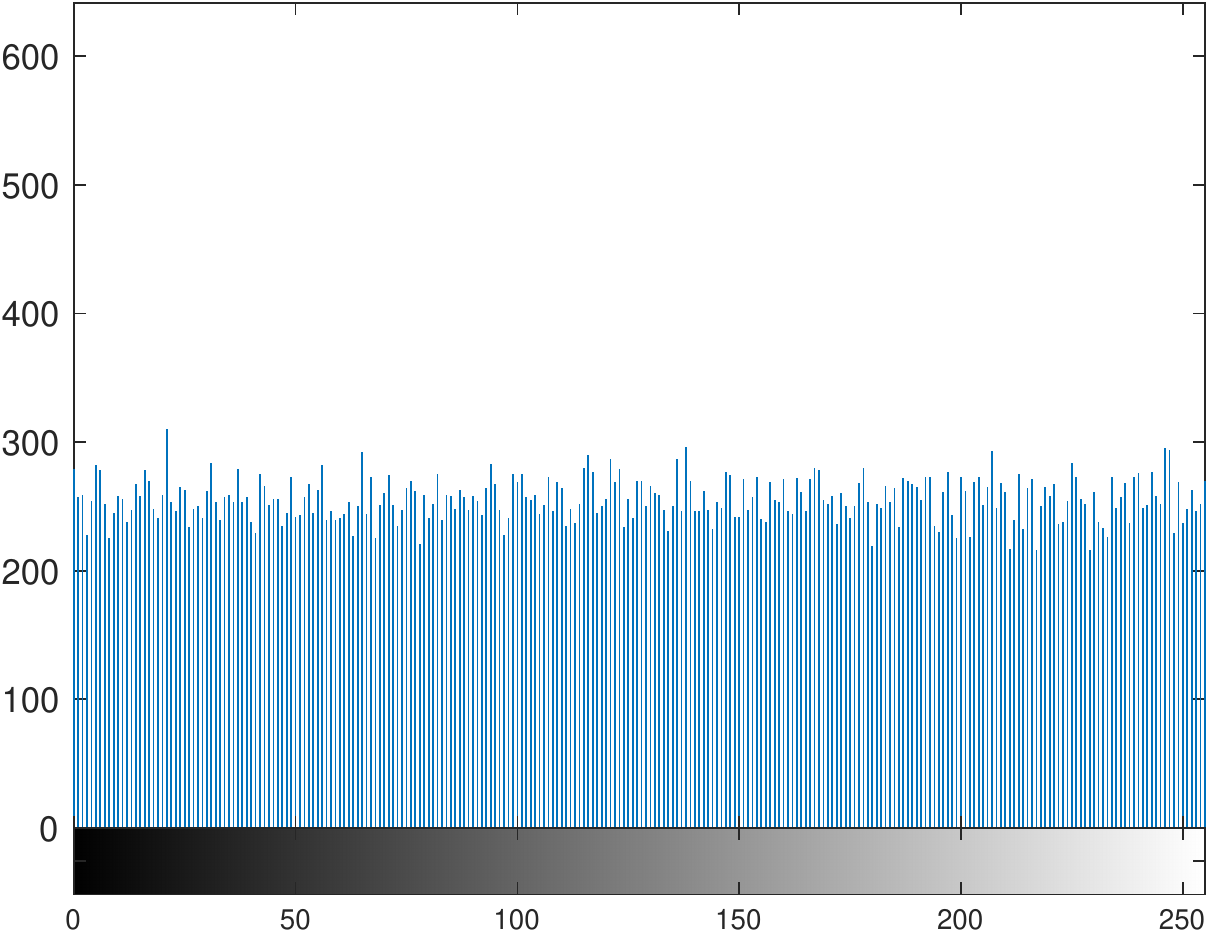}\label{fig:hist_color_p}}
	\caption{Histogram analysis: (a) Airplane; (b) Histogram of R channel of (a); (c) Histogram of G channel of (a); (d) Histogram of B channel of (a); (e) Encrypted Airplane; (f) Histogram of R channel of (e); (g) Histogram of G channel of (e); (h) Histogram of B channel of (e); (i) Baboo; (j) Histogram of R channel of (i); (k) Histogram of G channel of (i); (l) Histogram of B channel of (i); (m) Encrypted Baboo; (n) Histogram of R channel of (m); (o) Histogram of G channel of (m); (p) Histogram of B channel of (m);}
	\label{fig:hist_color}
\end{figure}

The histograms presents the distribution of pixels intuitively. In order to quantitatively examine whether the distribution of pixels in cipher is exact uniform, a $\chi^2$-test is performed on encrypted images \cite{zhangUnifiedImageEncryption2018}. For a 8-bit $M\times N$ grayscale with the hypothesis of uniform distribution, the test statistic is calculated as:
\begin{equation}
	\chi^2=\sum_{i=0}^{255}\frac{(f_i-f_e)^2}{f_e},
\end{equation}
where $f_e=MN/256$, $f_i$ is the occurrence frequency of gray level $i$ in image. $\chi^2$ has a $\chi^2$ distribution of 255 degrees of freedom.

Set the significant level as $\alpha=0.05$, then we may accept the hypothesis if $\chi^2 < \chi_{0.05}^2(255)=293.2478$. Otherwise we reject the hypothesis. The $\chi^2$-test values of six images as well as their ciphers are listed as Table \ref{tab:chi2_test}. As we can see, The plaintext images have $\chi^2$ values far larger than $\chi_{0.05}^2(255)$ while all the ciphers have $\chi^2$ values less than $293.2478$. That is, all ciphertext images pass the test and we can confidently conclude the pixels in encrypted images are evenly distributed.

\begin{table}[htbp]
	\centering
	\caption{$\chi^2$ test for original images and encrypted images}
	\begin{tabular}{c|cccccccccccc}
		\toprule
		Image & \multicolumn{3}{c}{Lena} & \multicolumn{3}{c}{Cameraman} & \multicolumn{3}{c}{Peppers} & \multicolumn{3}{c}{Starfish} \\
		\midrule
		Original image & \multicolumn{3}{c}{11097} & \multicolumn{3}{c}{29979} & \multicolumn{3}{c}{36778} & \multicolumn{3}{c}{16292} \\
		Encrypted image & \multicolumn{3}{c}{267.57} & \multicolumn{3}{c}{250.41} & \multicolumn{3}{c}{210.09} & \multicolumn{3}{c}{256.22} \\
		\midrule
		Image & \multicolumn{6}{c}{Airplane}                  & \multicolumn{6}{c}{Baboon} \\
		Channels & \multicolumn{2}{c}{R} & \multicolumn{2}{c}{G} & \multicolumn{2}{c}{B} & \multicolumn{2}{c}{R} & \multicolumn{2}{c}{G} & \multicolumn{2}{c}{B} \\
		\midrule
		Original image & \multicolumn{2}{c}{15362} & \multicolumn{2}{c}{15059} & \multicolumn{2}{c}{25618} & \multicolumn{2}{c}{23763} & \multicolumn{2}{c}{35738} & \multicolumn{2}{c}{20450} \\
		Encrypted image & \multicolumn{2}{c}{280.55} & \multicolumn{2}{c}{233.14} & \multicolumn{2}{c}{258.70} & \multicolumn{2}{c}{279.74} & \multicolumn{2}{c}{253.49} & \multicolumn{2}{c}{212.15} \\
		\bottomrule
	\end{tabular}%
	\label{tab:chi2_test}%
\end{table}%

\begin{figure}[htbp]
	\centering
	\subfigure[$x_1$]{
		\includegraphics[width=0.3\textwidth]{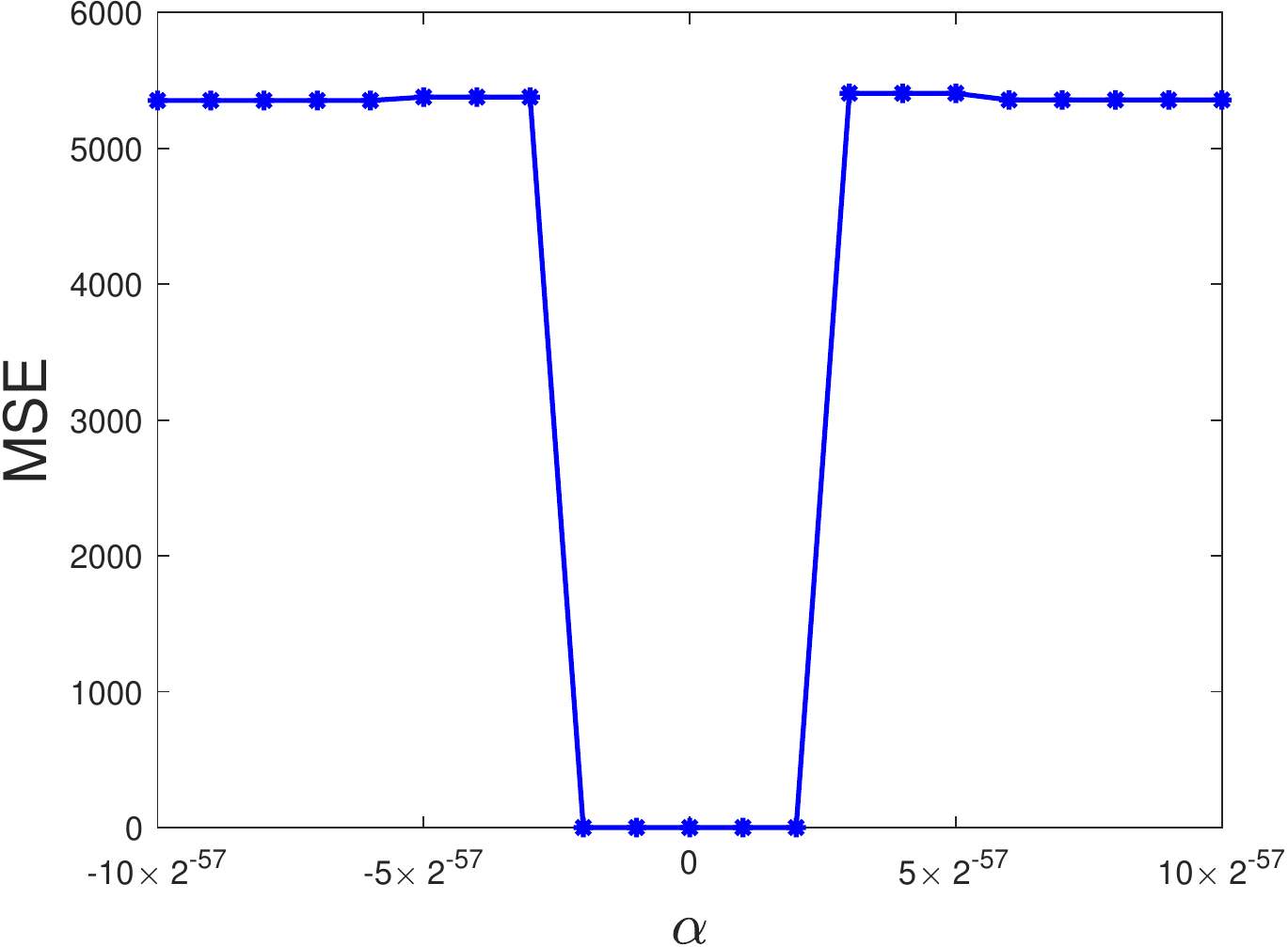}}
	\subfigure[$r_1$]{
		\includegraphics[width=0.3\textwidth]{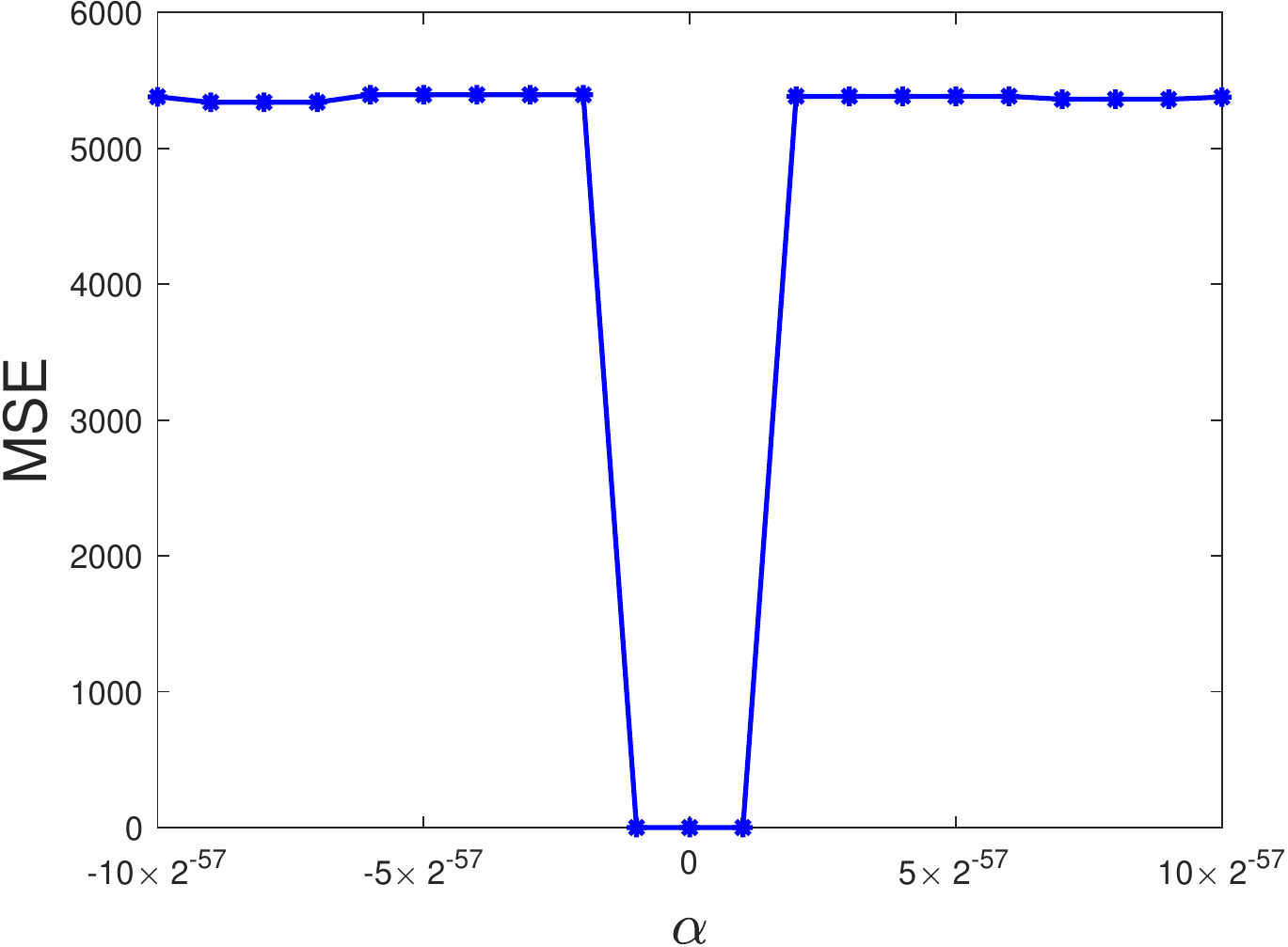}}
	\subfigure[$x_2$]{
		\includegraphics[width=0.3\textwidth]{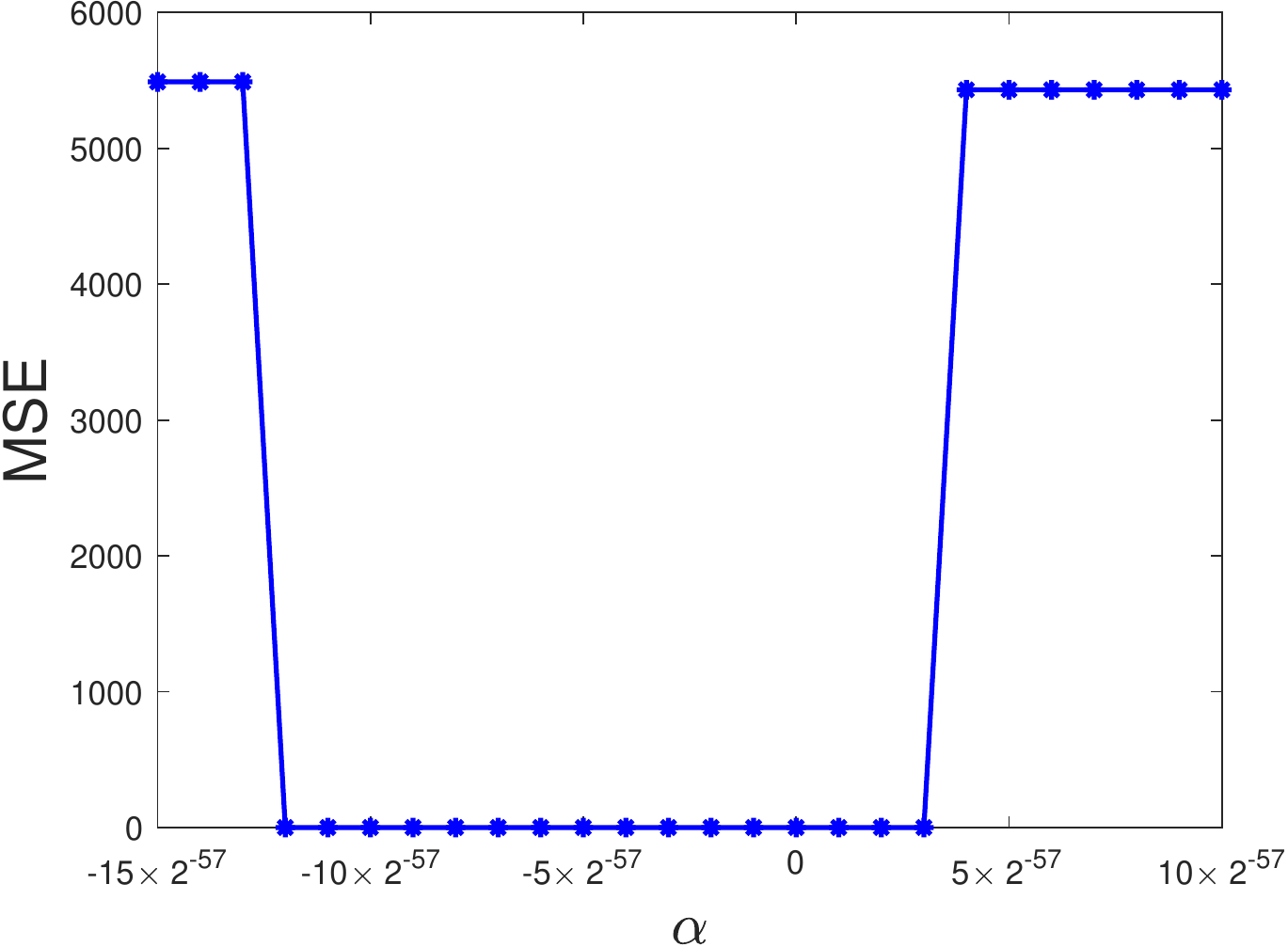}}
	\subfigure[$r_2$]{
		\includegraphics[width=0.3\textwidth]{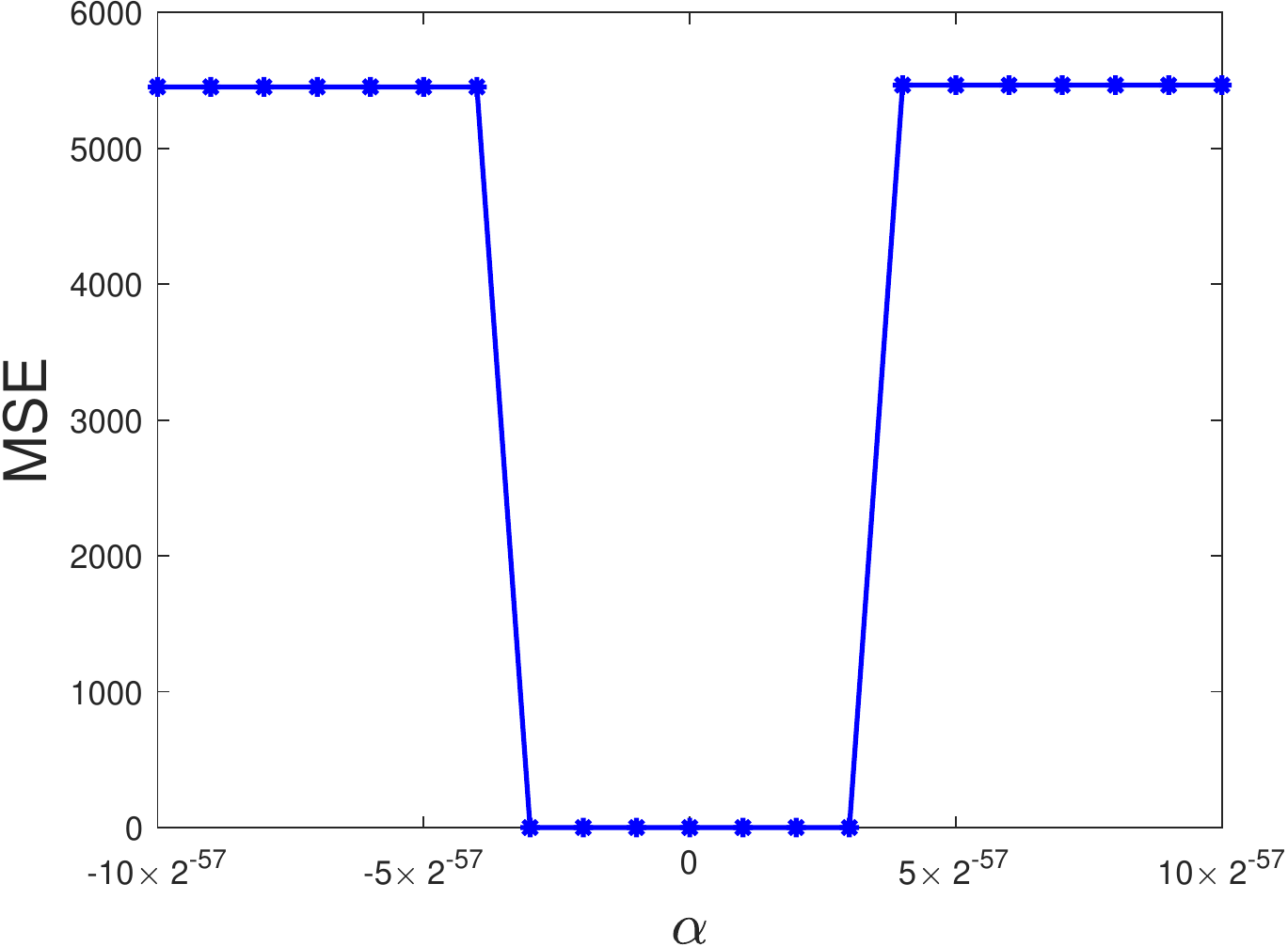}}
	\subfigure[$x_3$]{
		\includegraphics[width=0.3\textwidth]{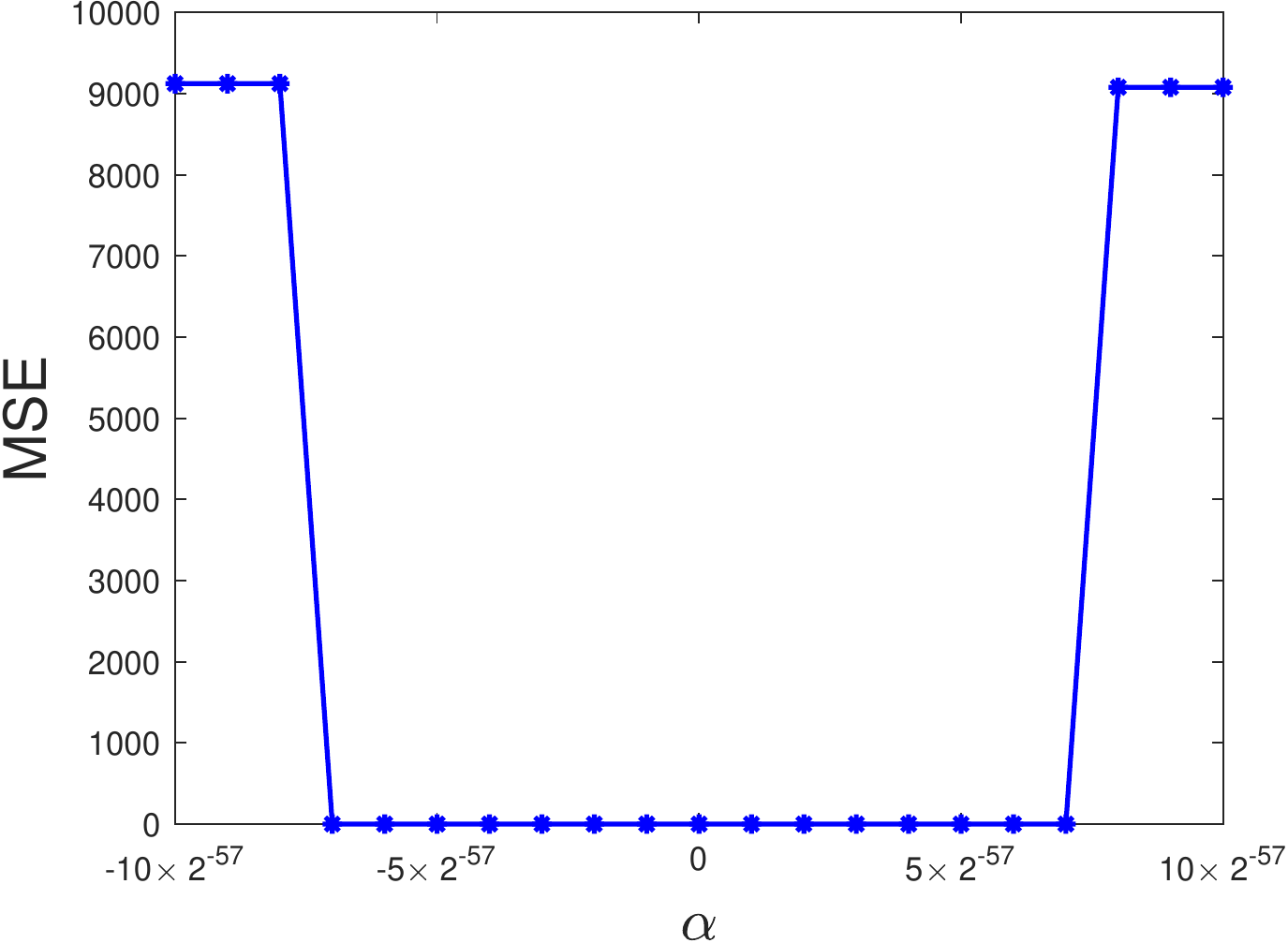}}
	\subfigure[$r_3$]{
		\includegraphics[width=0.3\textwidth]{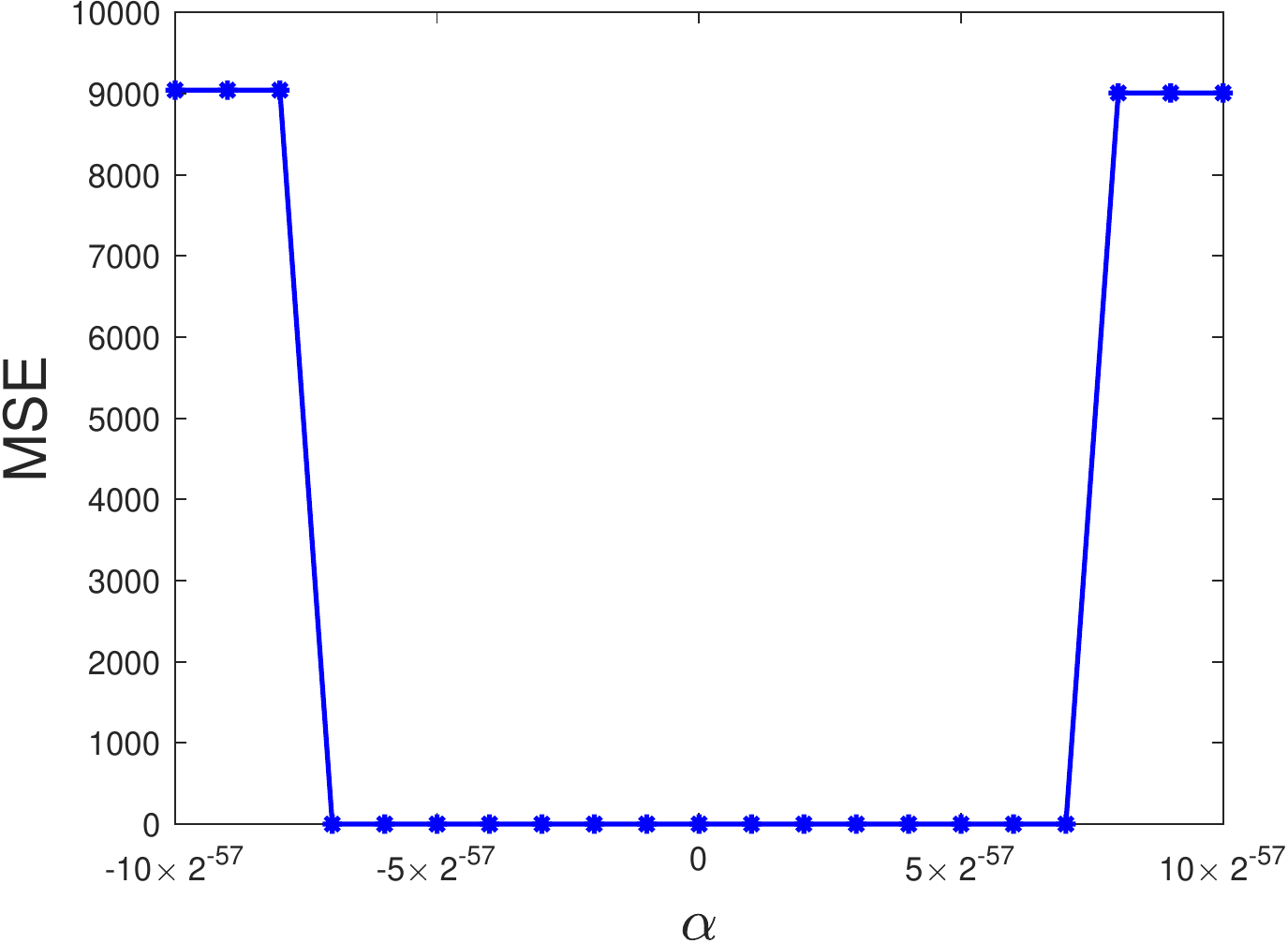}}
	\caption{MSEs between decrypted image and original image with the subtle change $\alpha$ on various keys}
	\label{fig:key_sensitivity}
\end{figure}

\subsection{Key space and key sensitivity}
Key space size is the total number of the different key groups that can be used in encryption. The key space should be large enough to resist the brute-force attack. Key sensitivity tells us how slightly a key changes will make difference to the encryption system. In our proposed algorithm, the secret keys include $x_1$, $r_1$ $x_2$, $r_2$, $x_3$, $r_3$. The range of these keys are $(0,1)$.

In order to test the sensitivities of these keys, one common approach is to decrypt cipher image with a key slightly changed by $\alpha$ (i.e. change $x_1$ to $x_1+\alpha$). Then we adjust the change intensity $\alpha$ until we decrypt the correct plaintext image. We use mean square error (MSE) to assess the difference between decrypted image and original image with the variety of $\alpha$. The MSE of two $M\times N$ images $F_1$, $F_2$ is defined as:
\begin{equation}
	{\rm MSE}=\frac{1}{MN}\sum_{i=1}^M{\sum_{j=1}^N{\left( F_1\left( i, j \right) -F_2\left( i, j \right) \right) ^2}}
\end{equation}
MSE $=0$ means that the two images are exactly the same. Considering that the secret keys in our proposed encryption scheme are mutual independent, we can test the key sensitivity by changing one key and leave others unchanged.

The $\alpha$-MSE plots of six keys are shown as Figure \ref{fig:key_sensitivity}. Form these diagrams we obtain the sensitivities of secrete key listed as Table \ref{tab:key_sensitivity}. Since the typical value of smallest computational precision is $2^{-52}$ according to IEEE floating-point standard \cite{kahanIEEEStandard7541996}, the total number of different value for each key above is more than $2^{52}$ as it is sensitive enough. Therefore, the key space size of the proposed encryption scheme is more than $2^{52\times 6} = 2^{312}$, which is completely enough to resist all kinds of brute-force attack.


\begin{table}[htbp]
	\centering
	\caption{Sensitivities of secret keys}
	\begin{tabular}{c|cccccc}
		\toprule
		Keys &$x_1$ &$r_1$ &$x_2$ &$r_2$ &$x_3$ &$r_3$ \\
		\midrule
		Sensitivity &$2^{-55}$ &$2^{-56}$ &$2^{-53}$ &$2^{-54}$ &$2^{-54}$ &$2^{-54}$ \\
		\bottomrule
	\end{tabular}
	\label{tab:key_sensitivity}
\end{table}

\subsection{Resistance to differential attack}
By constructing special plain image to encrypt and find the correspondence among pixels, the chosen plaintext attack is a powerful attack technique. However, most chosen plaintext attack will be useless faced with our proposed encryption algorithm since we will get totally different cipher images once the plaintext image changes slightly. When we say "totally different", we tend to use the Number of Pixel Change Rate (NPCR) and the Unified Average Changing Intensity (UACI) to measure the difference between two images. Let $F_1\left(i,j\right)$ and $F_2\left(i,j\right)$ be the $\left(i,j\right)th$ pixel of two image $F_1$ and $F_2$ with size $M \times N$,respectively, the definitions of NPCR and UACI are as follows \cite{wuNPCRUACIRandomness2011}:
\begin{equation}
\label{eq:NPCR}
NPCR=\frac{\sum_{i,j}D\left(i,j\right)}{MN} \times 100\%,
\end{equation}
where $D\left(i,j\right)$ is defined as
\begin{equation}
D\left(i,j\right)=\left\{
\begin{aligned}
&0, &F_1\left(i,j\right) = F_2\left(i,j\right), \\
&1, &F_1\left(i,j\right) \neq F_2\left(i,j\right), 
\end{aligned}\right.
\end{equation}

\begin{equation}
\label{eq:UACI}
UACI = \frac{1}{MN} \sum_{i,j} \frac{|F_1\left(i,j\right) - F_2\left(i,j\right)|}{L-1} \times 100\%,
\end{equation}
where $L$ denotes the largest supported gray value compatible with the image format.

For 8-bit grayscale images, it has been proved that the expect values of NPCR and UACI between two random image (each pixel takes a random number independently) are 99.6094\% and 33.4635\%, respectively \cite{wuNPCRUACIRandomness2011}.

In the test, We change only one bit in plain image and encrypt the two images with the same keys. Then NPCR and UACI value between two cipher images are calculated by Eq.(\ref{eq:NPCR})-(\ref{eq:UACI}). Of course, we have to do the experiments for many times and take the average value as the ultimate result. The results of sensitivity test for six images and the theoretical values defined by the 0.1\%, 1\% and 5\% confidence intervals are listed as Table \ref{tab:NPCR_my} and Table \ref{tab:UACI_my}.
As we can see from the above table, NPCR and UACI are close to their expected values, which has proved that we will get a totally different cipher image every time we encrypt the same plain image with one-bit-change. Thus, our proposed algorithm has achieved high resistance to all kinds of differential attack. 

Moreover, for Lena iamge, Table \ref{tab:NPCR_cmp} and Table \ref{tab:UACI_cmp} list the NPCR and UACI results  of some existing image encryption algorithms. It can be seen that our proposed method passes all the test while some existing methods fail some items. Thus, the proposed scheme exhibits excellent performance by comparison, which proves the progress of this work.
	
\begin{table}[htbp]
	\scriptsize
	\centering
	\caption{NPCR test results of the proposed image encryption algorithm}
	\begin{tabular}{ccccc}
		\hline
		Image &NPCR &NPCR critical values \\
		\hline
		& &${\rm NPCR}_{0.05}^{*} = 99.5693\%$ &${\rm NPCR}_{0.01}^{*} = 99.5527\%$ &${\rm NPCR}_{0.001}^{*} = 99.5341\%$ \\
		\hline
		Lena &99.6964\%  &Passed &Passed &Passed\\
		Cameraman &99.6002\% &Passed &Passed &Passed\\
		Peppers &99.6063\% &Passed &Passed &Passed\\
		Starfish &99.6351\% &Passed &Passed &Passed\\
		Airplane &99.6228\% &Passed &Passed &Passed\\
		Baboon &99.6191\% &Passed &Passed &Passed\\
		\hline
	\end{tabular}
	\label{tab:NPCR_my}
\end{table}

\begin{table}[htbp]
	\scriptsize
	\centering
	\caption{UACI test results of the proposed image encryption algorithm}
	\begin{tabular}{ccccc}
		\hline
		Image &UACI &UACI critical values \\
		\hline
		& &\tabincell{c}{${\rm UACI}_{0.05}^{*-} = 33.2824\%$ \\
	${\rm UACI}_{0.05}^{*+} = 33.6447\%$} &\tabincell{c}{${\rm UACI}_{0.01}^{*-} = 33.2255\%$ \\
	${\rm UACI}_{0.01}^{*+} = 33.7016\%$} &\tabincell{c}{${\rm UACI}_{0.001}^{*-} = 33.1594\%$ \\
	${\rm UACI}_{0.001}^{*+} = 33.7677\%$} \\
		\hline
		Lena &33.4758\%  &Passed &Passed &Passed\\
		Cameraman &33.5031\% &Passed &Passed &Passed\\
		Peppers &33.4523\% &Passed &Passed &Passed\\
		Starfish &33.4875\% &Passed &Passed &Passed\\
		Airplane &33.5115\% &Passed &Passed &Passed\\
		Baboon &33.4624\% &Passed &Passed &Passed\\
		\hline
	\end{tabular}
	\label{tab:UACI_my}
\end{table}

\begin{table}[htbp]
	\scriptsize
	\centering
	\caption{NPCR test results of different image encryption algorithms}
	\begin{tabular}{ccccc}
		\hline
		Algorithm &NPCR &NPCR critical values \\
		\hline
		& &${\rm NPCR}_{0.05}^{*} = 99.5693\%$ &${\rm NPCR}_{0.01}^{*} = 99.5527\%$ &${\rm NPCR}_{0.001}^{*} = 99.5341\%$ \\
		\hline
		Proposed &99.6964\%  &Passed &Passed &Passed\\
		Ref.\cite{asgari-chenaghluNovelImageEncryption2019} &99.6191\% &Passed &Passed &Passed\\
		Ref.\cite{gayathriEfficientSpatiotemporalChaotic2019} &99.6096\% &Passed &Passed &Passed\\
		Ref.\cite{alawidaNewHybridDigital2019} &99.620\% &Passed &Passed &Passed\\
		Ref.\cite{wangNovelChaoticBlock2015} &99.5865\% &Passed &Passed &Passed\\
		Ref.\cite{wangFastImageAlgorithm2015} &99.61\% &Passed &Passed &Passed\\
		Ref.\cite{xuNovelBitlevelImage2016} &99.62\% &Passed &Passed &Passed\\
		Ref.\cite{liuFastImageEncryption2016} &99.61\% &Passed &Passed &Passed\\
		Ref.\cite{wuDesignImageCipher2014} &99.6689\% &Passed &Passed &Passed\\
		Ref.\cite{ahmadChaosbasedDiffusionHighly2015} &99.36\% &\textbf{Failed} &\textbf{Failed} &\textbf{Failed}\\
		\hline
	\end{tabular}
	\label{tab:NPCR_cmp}
\end{table}

\begin{table}[htbp]
	\scriptsize
	\centering
	\caption{UACI test results of different image encryption algorithms}
	\begin{tabular}{ccccc}
		\hline
		Image &UACI &UACI critical values \\
		\hline
		& &\tabincell{c}{${\rm UACI}_{0.05}^{*-} = 33.2824\%$ \\
	${\rm UACI}_{0.05}^{*+} = 33.6447\%$} &\tabincell{c}{${\rm UACI}_{0.01}^{*-} = 33.2255\%$ \\
	${\rm UACI}_{0.01}^{*+} = 33.7016\%$} &\tabincell{c}{${\rm UACI}_{0.001}^{*-} = 33.1594\%$ \\
	${\rm UACI}_{0.001}^{*+} = 33.7677\%$} \\
		\hline
		Proposed &33.4758\%  &Passed &Passed &Passed\\
		Ref.\cite{asgari-chenaghluNovelImageEncryption2019} &33.6751\% &\textbf{Failed} &Passed &Passed\\
		Ref.\cite{gayathriEfficientSpatiotemporalChaotic2019} &33.4283\% &Passed &Passed &Passed\\
		Ref.\cite{alawidaNewHybridDigital2019} &33.505\% &Passed &Passed &Passed\\
		Ref.\cite{wangNovelChaoticBlock2015} &33.2533\% &\textbf{Failed} &Passed &Passed\\
		Ref.\cite{wangFastImageAlgorithm2015} &33.13\% &\textbf{Failed} &\textbf{Failed} &\textbf{Failed}\\
		Ref.\cite{xuNovelBitlevelImage2016} &33.51\% &Passed &Passed &Passed\\
		Ref.\cite{liuFastImageEncryption2016} &33.33\% &\textbf{Failed} &Passed &Passed\\
		Ref.\cite{wuDesignImageCipher2014} &33.4936\% &Passed &Passed &Passed\\
		Ref.\cite{ahmadChaosbasedDiffusionHighly2015} &32.72\% &\textbf{Failed} &\textbf{Failed} &\textbf{Failed}\\
		\hline
	\end{tabular}
	\label{tab:UACI_cmp}
\end{table}

\subsection{Correlation analysis}
A meaningful image has high correlation among adjacent pixels. Thus an important feature for image encryption algorithm is to break this high correlation and make it possible to resist all kinds of statistical attack. To test correlation of image, we randomly select 5000 pairs of adjacent pixels from three different directions (horizontal, vertical and diagonal) in plaintext images and ciphertext images, and calculate the correlation coefficient by:
\begin{equation}
CC=\frac{\sum_{i=1}^{N}\left(x_i-\bar{x}\right)\left(y_i-\bar{y}\right)}{\sqrt{\left(\sum_{i=1}^{N}\left(x_i-\bar{x}\right)^2\right)\left(\sum_{i=1}^{N}\left(y_i-\bar{y}\right)^2\right)}},
\end{equation}
where $\bar{x}=\frac{1}{N}\sum_{i=1}^{N}x_i$, $\bar{y}=\frac{1}{N}\sum_{i=1}^{N}y_i$.
The result is listed in Table.\ref{tab:correlation}. As we can see, the plain images  hold the correlation coefficient close to 1 in all direction, which demonstrates high correlation. After encryption, the correlation coefficient is close to 0. That is, our proposed algorithm has broken the correlation between adjacent pixels. The distribution of gray values in Lena for two adjacent pixels is shown in Figure.\ref{Fig:CC}, which has proved the conclusion above visually. Besides, Table \ref{tab:correlation_cmp} presents the correlation coefficients of encrypted Lena using different encryption algorithms. It can be seen that the proposed encryption scheme gets relatively favorable results among these methods.

\begin{figure}[htbp]
	\centering
	\subfigure[]{
		\label{Fig:CC_src_hor}
		\includegraphics[width=0.3\textwidth]{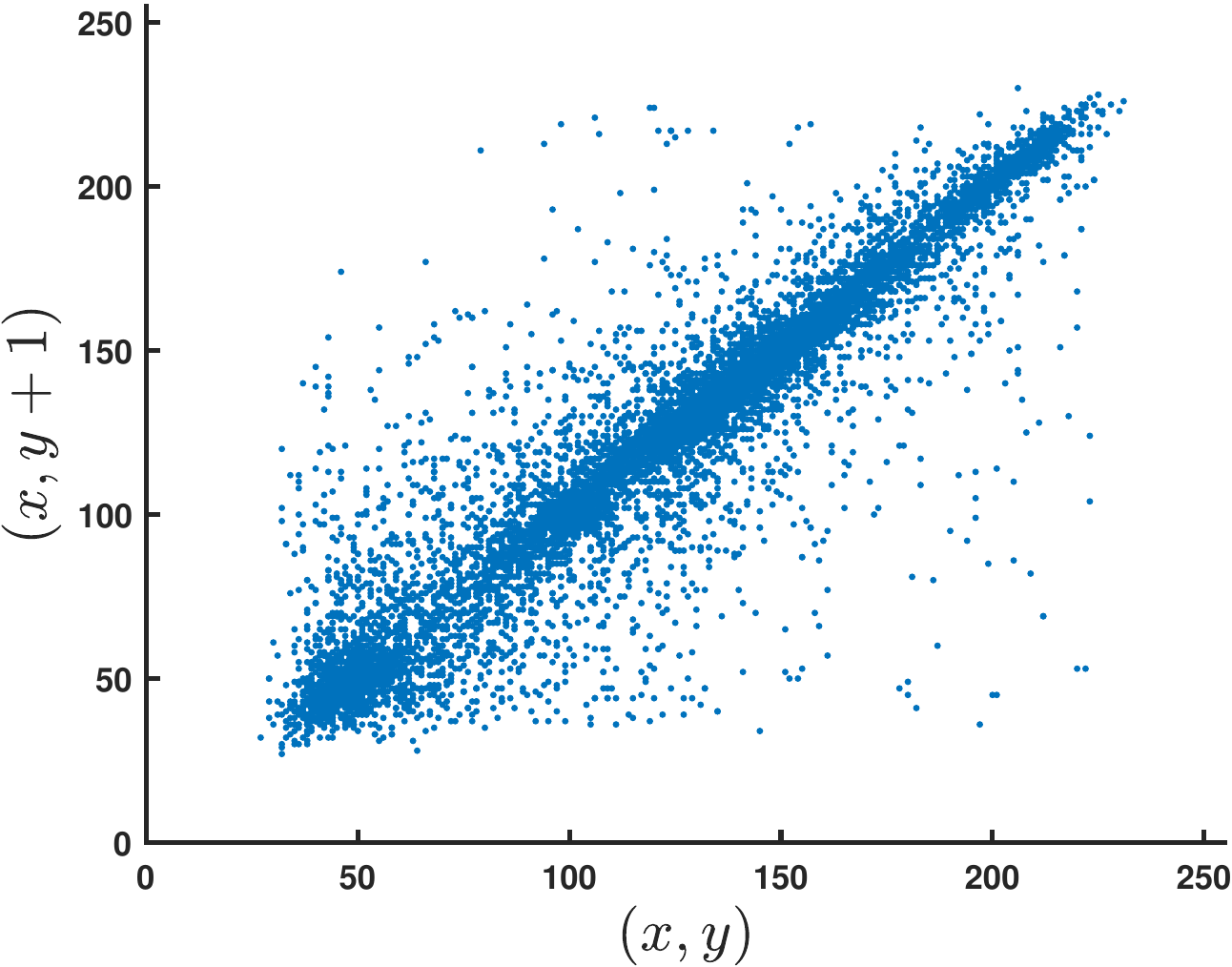}}
	\subfigure[]{
		\label{Fig:CC_src_ver}
		\includegraphics[width=0.3\textwidth]{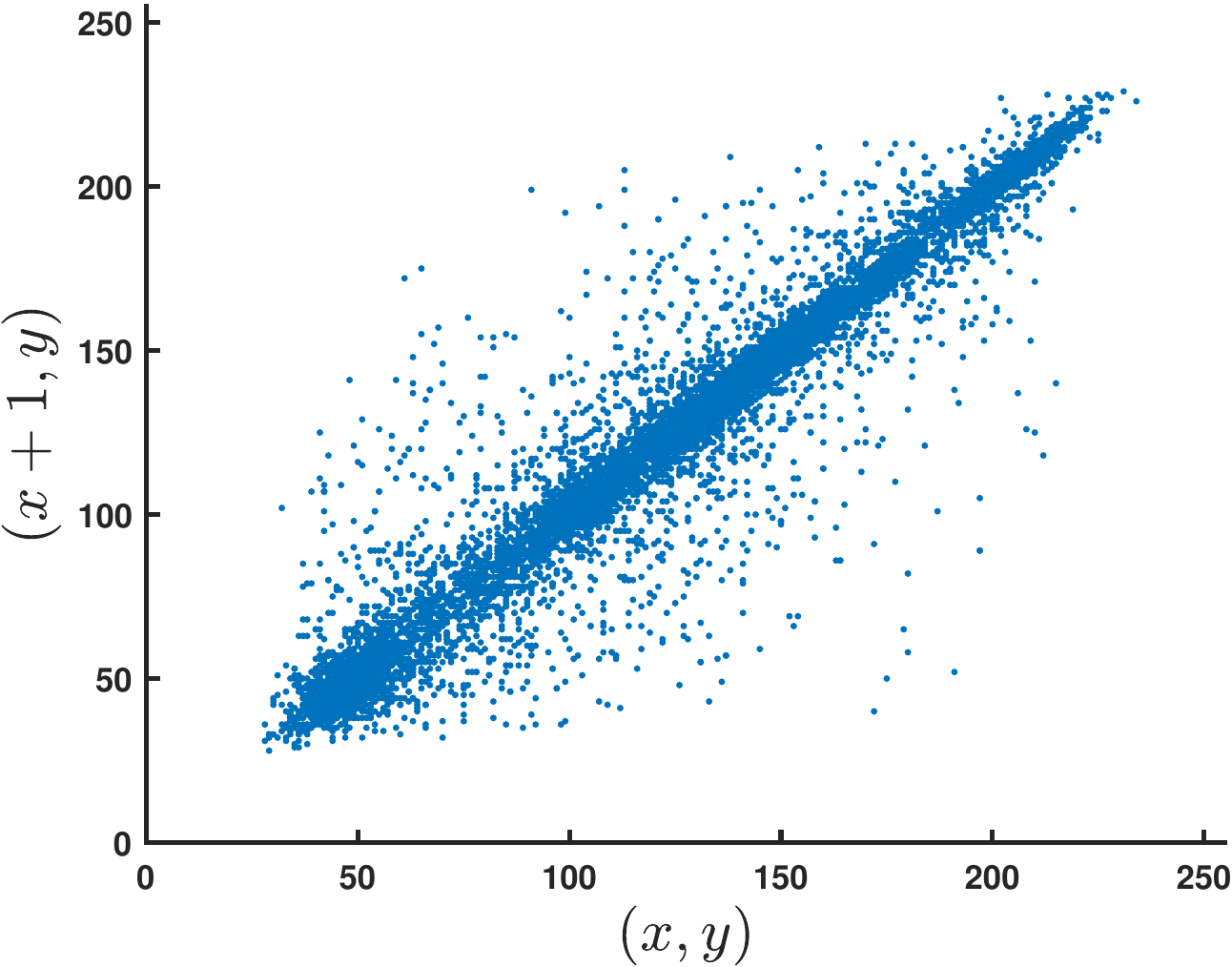}}
	\subfigure[]{
		\label{Fig:CC_src_dia}
		\includegraphics[width=0.3\textwidth]{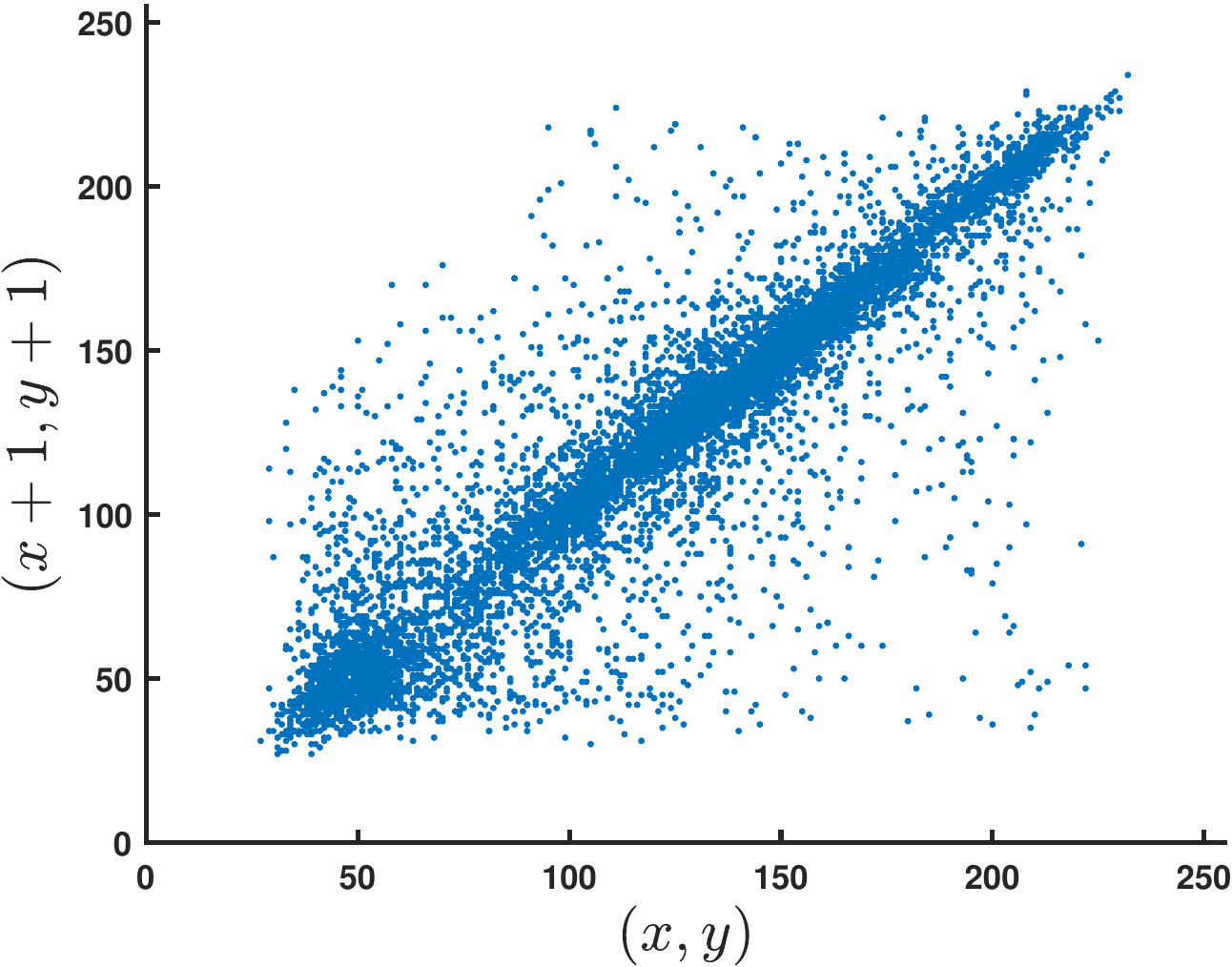}}
	\subfigure[]{
		\label{Fig:CC_en_hor}
		\includegraphics[width=0.3\textwidth]{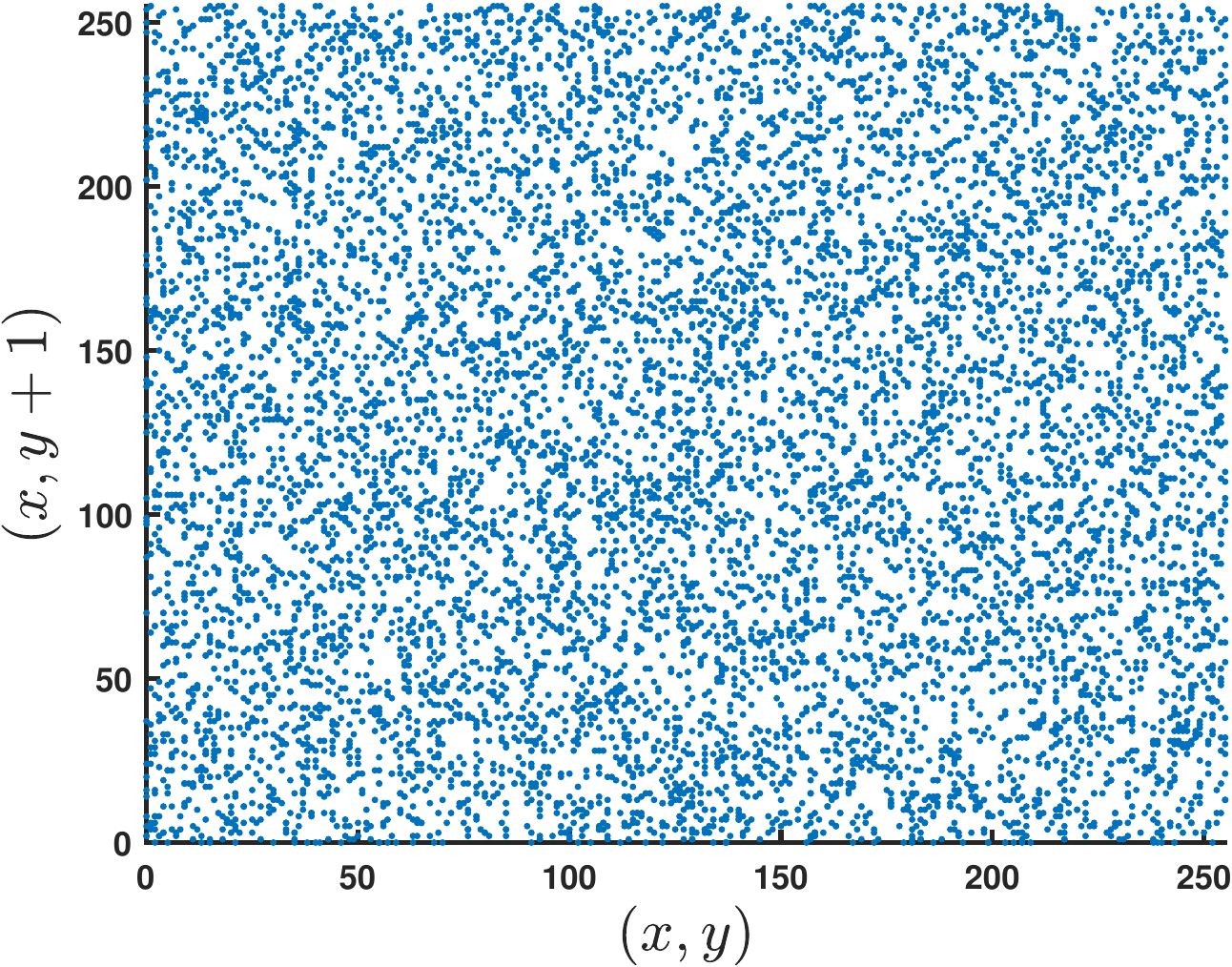}}
	\subfigure[]{
		\label{Fig:CC_en_ver}
		\includegraphics[width=0.3\textwidth]{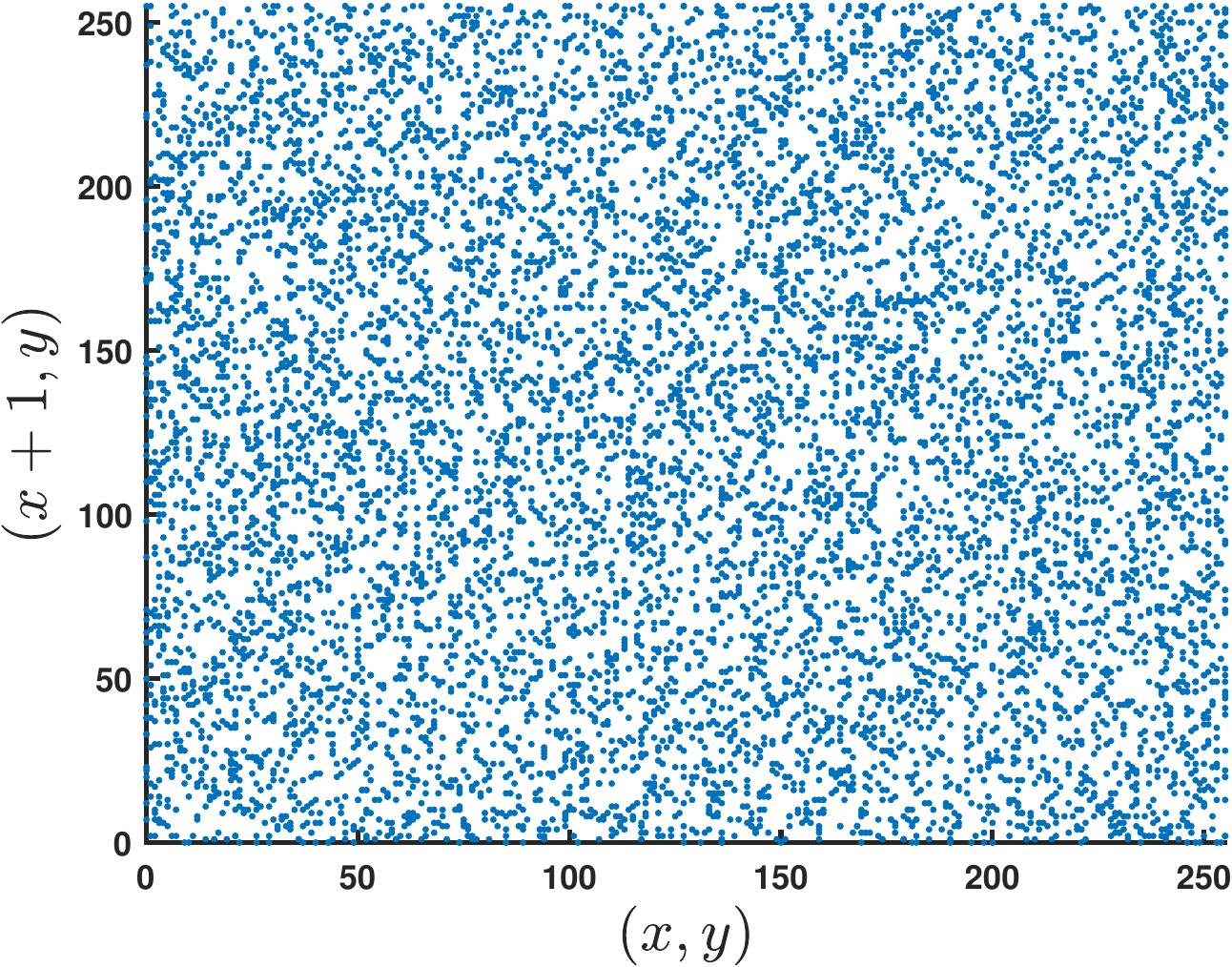}}
	\subfigure[]{
		\label{Fig:CC_en_dia}
		\includegraphics[width=0.3\textwidth]{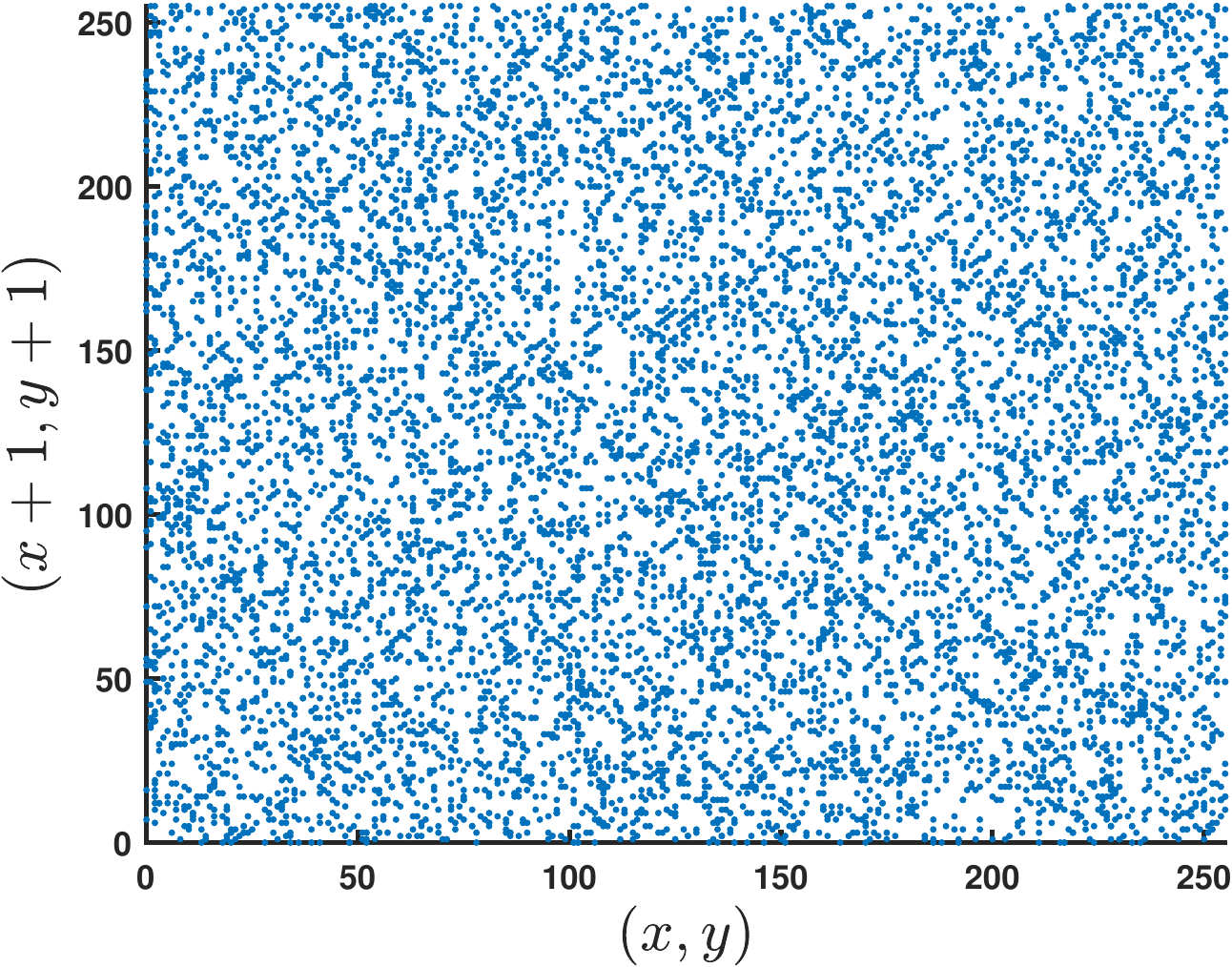}}
	\caption{Adjacent pixel correlation in different directions: (a) Horizontal direction of Lena; (b) Vertical direction of Lena; (c) Diagonal direction of Lena; (d) Horizontal direction of Encrypted Lena; (e) Vertical direction of Encrypted Lena; (f) Diagonal direction of Encrypted Lena;}
	\label{Fig:CC}
\end{figure}

\begin{table}[ht]
	\scriptsize
	\centering
	\caption{Correlation coefficients}
	\begin{tabular}{c|cc|cc|cc}
		\toprule
		\multirow{2}{*}{Image} & \multicolumn{2}{|c|}{Horizontal}& \multicolumn{2}{|c|}{Vertical}& \multicolumn{2}{|c}{Diagonal} \\
		\cline{2-7}
		&plain-image &cipher-image &plain-image &cipher-image &plain-image &cipher-image \\
		\midrule
		Lena& 0.973187& 0.001135& 0.960615& 0.002457& 0.945709& 0.003213 \\
		Cameraman& 0.807831& 0.001258& 0.842404& 0.007285& 0.802359& 0.003244 \\
		Peppers& 0.923041& 0.002659& 0.921635& 0.004842& 0.910057& 0.002405 \\
		Starfish& 0.963034&  0.001167& 0.963555& 0.008413& 0.942439& 0.000175 \\
		Woman& 0.882540&  0.002617& 0.893154& 0.008413& 0.884573& 0.003216 \\
		Baboon& 0.935587&  0.003578& 0.942316& 0.000311& 0.939928& 0.006542 \\
		\bottomrule
	\end{tabular}
	\label{tab:correlation}
\end{table}

\begin{table}[htbp]
	\centering
	\caption{Correlation coefficients of various algorithms}
	\begin{tabular}{c|ccc}
		\toprule
		Algorithm &Horizontal &Vertical &Diagonal \\
		\midrule
		Proposed &0.001135 &0.002457 &0.003213 \\
		Ref.\cite{wangFastImageAlgorithm2015} &-0.0331 &0.0057 &0.0169 \\
		Ref.\cite{kangRealityPreservingMultipleParameter2019} &0.0015 &0.0017 &-0.0033 \\
		Ref.\cite{liaoNovelImageEncryption2010} &0.0127 &-0.0190 &-0.0012 \\
		Ref.\cite{asgari-chenaghluNovelImageEncryption2019} &-0.00222 &0.00137 &0.00291 \\
		Ref.\cite{gayathriEfficientSpatiotemporalChaotic2019} &0.01410 &0.02960 &0.00540 \\
		Ref.\cite{alawidaNewHybridDigital2019} &-0.0017 &-0.0084 &-0.0019 \\
		Ref.\cite{zhangUnifiedImageEncryption2018} &-0.009448 &0.024064 &-0.041171 \\
		Ref.\cite{huaCosinetransformbasedChaoticSystem2019} &-0.003280 &-0.000777 &-0.000181 \\
		Ref.\cite{xuNovelBitlevelImage2016} &-0.0230 &0.0019 &-0.0034 \\
		\bottomrule
	\end{tabular}
	\label{tab:correlation_cmp}
\end{table}

\subsection{Robustness to data loss}
In actual communication system, some information of the encrypted image may be lost when it is transmitted in channel for all kinds of reasons. And attackers may intercept the information and cause data loss as well. Then receivers would like to recover the original image as much as possible from the partial information received in such situation. Thus it is useful if the encryption algorithm shows robustness to data loss.

In our experiment, we let the encrypted Lena lose 25\% and 50\% information in different pixels, respectively, and  decrypt them with correct keys. The result is shown in Figure.\ref{Fig:data_loss}. As we can see, the decrypted images have restored most of the original information visually, which has proved the robustness to data loss of our proposed encryption scheme.

\begin{figure}[htbp]
	\centering
	\subfigure[]{
		\label{Fig:en_loss_1}
		\includegraphics[width=0.22\textwidth]{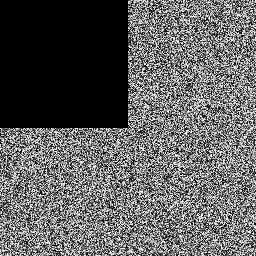}}
	\subfigure[]{
		\label{Fig:en_loss_2}
		\includegraphics[width=0.22\textwidth]{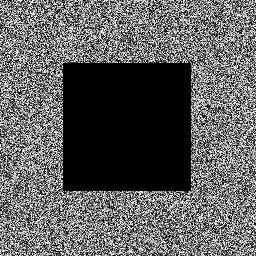}}
	\subfigure[]{
		\label{Fig:en_loss_3}
		\includegraphics[width=0.22\textwidth]{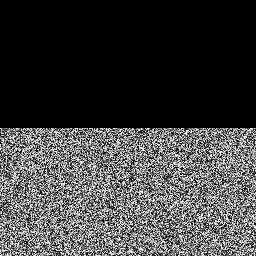}}
	\subfigure[]{
		\label{Fig:en_loss_4}
		\includegraphics[width=0.22\textwidth]{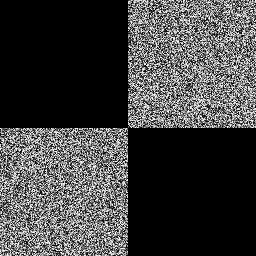}}
	\subfigure[]{
		\label{Fig:de_loss_1}
		\includegraphics[width=0.22\textwidth]{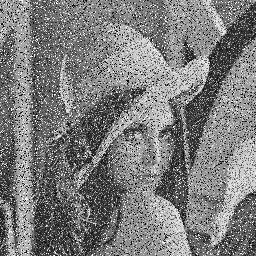}}
	\subfigure[]{
		\label{Fig:de_loss_2}
		\includegraphics[width=0.22\textwidth]{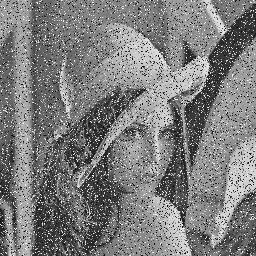}}
	\subfigure[]{
		\label{Fig:de_loss_3}
		\includegraphics[width=0.22\textwidth]{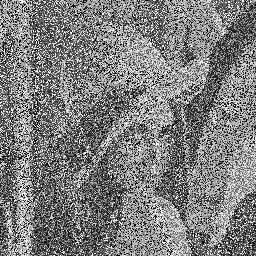}}	
	\subfigure[]{
		\label{Fig:de_loss_4}
		\includegraphics[width=0.22\textwidth]{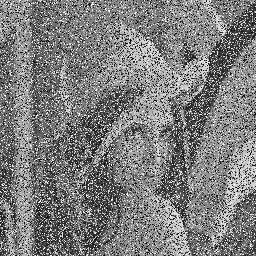}}
	\caption{Data loss analysis: (a) 25\% data loss of encrypted Lena; (a) 25\% data loss of encrypted Lena; (c) 50\% data loss of encrypted Lena; (d) 50\% data loss of encrypted Lena; (e) decrypted Lena from (a); (f) decrypted Lena from (b); (g) decrypted Lena from (c); (h) decrypted Lena from (d).}
	\label{Fig:data_loss}
\end{figure}

%

\section{Conclusion}
\label{sec:conclusion}
In this paper, a novel chaotic system named as UT-CCS is proposed as a general chaotic framework that covers several existing chaotification models. The UT-CCS can generate new chaotic maps with excellent performance using any two existing chaotic maps. A theorem leading us to choose better UTFs, is given and proved. Several new chaotic maps are generated as examples and the results of chaotic behavior analysis demonstrate the superiority of UT-CCS. Besides, a new PRNG based on generated chaotic maps is designed for image encryption, of which the uniformity of numerical distribution is also proved mathematically. At last, a digital image encryption algorithm based on CBPRNG is proposed. The confusion phase and diffusion phase of the encryption scheme are based on different CBPRNG to achieve better randomness. We implement the algorithm in simulation environment and analyze the security through general tests. The results indicates the reliability of ciphertext images and the competitiveness of our algorithm compared with similar methods.

This paper makes contributions to chaos application and image information security. In future work, we will continue improving our model such as extending the UT-CCS to two-dimensional case.


\section*{Acknowledgement}
\label{sec:ack}
This work was supported in part by NSFC under Grant U1536104, 11301504.

\bibliographystyle{elsarticle-num}
\bibliography{chaos_IE}


%
%
%

\end{document}